\documentclass[12pt]{article}
\usepackage{amsmath}
\usepackage{graphicx}
\usepackage{enumerate}
\usepackage{natbib} 

\newcommand{\blind}{1}

\usepackage{hyperref,amssymb,amsmath,amsthm,bm,bbm,mathtools}

\hypersetup{colorlinks=true,
citecolor=blue,
anchorcolor = red,
}

\usepackage{algorithm}
\usepackage{algpseudocode}
\usepackage{fullpage}
\usepackage{epsfig}
\usepackage{enumerate}
\usepackage{xcolor}
\usepackage{extarrows}
\usepackage{tabularx}
\usepackage{makecell}
\usepackage{multirow}
\usepackage{xr}
\usepackage{subcaption}

\definecolor{darkred}{RGB}{222,0,32}

\newtheorem{lemma}{Lemma}[section]

\newtheorem{thm}{\bf Theorem}[section]
\newtheorem{cor}{\bf Corollary}[section]
\newtheorem{remark}{Remark}

\renewcommand{\baselinestretch}{1.2}

\newcommand{\R}{\mathbb{R}}
\newcommand{\E}{\mathbb{E}}

\newcommand{\Prob}{\mathbb{P}}
\newcommand{\tr}{\mathrm{tr}}
\newcommand{\ve}{\mathrm{vec}}
\newcommand{\col}{\mathrm{Col}}

\newcommand{\bu}{\bm{u}}

\newcommand{\hU}{\widehat{U}}
\newcommand{\hV}{\widehat{V}}
\newcommand{\hW}{\widehat{W}}

\newcommand{\tZ}{\widetilde{Z}}

\newcommand{\calG}{\mathcal{G}}
\newcommand{\calM}{\mathcal{M}}
\newcommand{\calN}{\mathcal{N}}
\newcommand{\calT}{\mathcal{T}}
\newcommand{\calZ}{\mathcal{Z}}

\newcommand{\frakE}{\mathfrak{E}}
\newcommand{\frakS}{\mathfrak{S}}
\newcommand{\frakR}{\mathfrak{R}}

\DeclareFontFamily{U}{mathx}{}
\DeclareFontShape{U}{mathx}{m}{n}{<-> mathx10}{}
\DeclareSymbolFont{mathx}{U}{mathx}{m}{n}
\DeclareMathAccent{\widecheck}{0}{mathx}{"71}

\newcommand{\abs}[1]{\left\lvert#1\right\rvert}
\newcommand{\norm}[1]{\left\lVert#1\right\rVert}
\newcommand{\angles}[1]{\left\langle #1 \right\rangle}

\newcommand{\braces}[1]{\left\{ #1 \right\}}

\DeclareMathOperator*{\argmax}{\arg\!\max}

\definecolor{DSgray}{cmyk}{0,1,0,0}

\begin{document}
\pagenumbering{arabic}

\def\spacingset#1{\renewcommand{\baselinestretch}%
{#1}\small\normalsize} \spacingset{1}

%
%
%

\def\TITLE{Distributed Tensor Principal Component Analysis with Data Heterogeneity}

\if1\blind
{
	\title{\bf \TITLE}
	\author{ \medskip
		Elynn Chen$^\flat$ \hspace{2ex}
		Xi Chen$^\natural$
		\hspace{2ex}
		Wenbo Jing$^\sharp$ \hspace{2ex}
		Yichen Zhang$^\dag$ \\ \normalsize
		$^{\flat,\natural,\sharp}$Stern School of Business, New York University  \\
		\normalsize $^{\dag }$ Mitchell E. Daniels, Jr. School of Business, Purdue University
	}
	\date{}
	\maketitle
} \fi

\if0\blind
{
	\title{\bf \TITLE}
	\bigskip
	\bigskip
	\bigskip
	\date{}
	\maketitle
	\medskip
} \fi

\bigskip
\begin{abstract}
	As tensors become widespread in modern data analysis, Tucker low-rank Principal Component Analysis (PCA) has become essential for dimensionality reduction and structural discovery in tensor datasets. Motivated by the common scenario where large-scale tensors are distributed across diverse geographic locations, this paper investigates tensor PCA within a distributed framework where direct data pooling is theoretically suboptimal or practically infeasible.
	
	We offer a comprehensive analysis of three specific scenarios in distributed Tensor PCA: a homogeneous setting in which tensors at various locations are generated from a single noise-affected model; a heterogeneous setting where tensors at different locations come from distinct models but share some principal components, aiming to improve estimation across all locations; and a targeted heterogeneous setting, designed to boost estimation accuracy at a specific location with limited samples by utilizing transferred knowledge from other sites with ample data.
	
	We introduce novel estimation methods tailored to each scenario, establish statistical guarantees, and develop distributed inference techniques to construct confidence regions. Our theoretical findings demonstrate that these distributed methods achieve sharp rates of accuracy by efficiently aggregating shared information across different tensors, while maintaining reasonable communication costs. Empirical validation through simulations and real-world data applications highlights the advantages of our approaches, particularly in managing heterogeneous tensor data.
\end{abstract}

\noindent%
\vfill


\newpage
\spacingset{1.78} 

\addtolength{\textheight}{.1in}%

\vspace{-.8em}
\section{Introduction}  \label{sec:intro}

In recent years, the prevalence of large-scale tensor datasets across diverse fields has garnered significant attention. Tucker low-rank tensor Principal Component Analysis (PCA) and Singular Value Decomposition (SVD) have become indispensable for unsupervised learning and dimension reduction, proving crucial in deriving valuable insights across various applications. These techniques are extensively utilized in recommendation systems \citep{karatzoglou2010multiverse}, biomedicine \citep{wen2024tensor}, natural image and video processing \citep{gatto2021tensor},  and health care systems \citep{ren2022blockchain}, among others.

Despite their wide applicability, the geographical dispersion of large-scale tensor datasets presents substantial challenges in data aggregation and analysis due to high communication costs, privacy concerns, and data security and ownership issues. For example, IT corporations face practical limitations in centralizing globally gathered customer data for recommendation systems due to communication budgets and network bandwidth constraints. Similarly, health records spread across multiple hospitals or jurisdictions pose significant privacy and ownership challenges for centralized processing.

Addressing these challenges, our paper makes several critical contributions to the field of tensor decomposition, particularly focusing on its application in distributed and heterogeneous environments. While recent studies have begun exploring distributed Tucker decomposition, they primarily concentrate on computational aspects such as parallelization and memory usage, focusing on environments where tensors share identical principal spaces \citep{shin2016fully,chakaravarthy2018optimizing,jang2020d}. These studies often overlook the importance of providing statistical guarantees and also fail to address the challenges posed by heterogeneous tensors, especially prevalent in fields like medical care, where tensor data from various locations differ in their underlying structures. Our research fills this gap by not only providing statistical convergence and inference theories for distributed Tucker decomposition but also extending its applicability beyond homogeneous data settings.

In this paper, we address the challenge of tensor PCA within a distributed framework, where tensors are stored across different machines without the feasibility of pooling them together. Our comprehensive analysis delineates three distinct scenarios: In the homogeneous scenario in Section \ref{sec:homo}, we develop a method that involves computing local estimators for each tensor's subspace, then transmitting these estimators to a central processor, and finally aggregating them to produce a global estimate, reducing communication by only transmitting essential subspace information. 

In the second scenario in Section \ref{sec:hetero}, tensors observed on different machines are allowed to be generated from different underlying models that share common leading principal components. Our focus is to improve estimation across all machines handling heterogeneous data. 
We formally define the partition of the common and individual components for each tensor and develop a generalized distributed Tensor PCA method to efficiently estimate these shared components and also identify and extract unique components specific to each tensor. This approach allows us to accommodate the distinct characteristics of each dataset within the collective framework. 

The third scenario in Section \ref{sec:transfer} focuses on another setting of heterogeneous distributed learning, aiming at enhancing the estimation accuracy of a tensor at a specific location by intelligently integrating abundant data from other locations. We develop a novel knowledge transfer algorithm based on a weighted averaging scheme that involves calibrating the influence of data based on their noise levels and relevance, thereby improving the precision of the target site's tensor estimation.

We have established comprehensive statistical error results for our proposed estimators across different scenarios. In Section \ref{sec:theory-homo}, our analysis involves a detailed decomposition and quantification of the bias and variance terms. We demonstrate that our estimator achieves the optimal minimax rate when the signal-to-noise ratio (SNR) exceeds a certain threshold, thus equating the performance of the pooled estimator that aggregates tensors directly (further details in Section \ref{sec:homo-setup}). Essentially, our distributed algorithm attains the highest possible estimation accuracy achievable in a non-distributed setting.

In the heterogeneous scenario detailed in Section \ref{sec:theory-hetero}, we show that our estimator for common components matches the performance of the homogeneous scenario, while the estimators for individual components maintain rates consistent with individual tensor PCA. This underscores the effectiveness of our method in simultaneously learning shared and unique tensor structures. Importantly, even when pooling of tensors is possible, creating an efficient estimator through direct aggregation is challenging due to data heterogeneity—a fact corroborated by our numerical studies where our distributed methods surpass the pooled estimator in various heterogeneous configurations.

Furthermore, we provide theoretical guarantees for our knowledge-transfer estimator, optimizing the weight allocations in our algorithm to enhance estimation accuracy beyond what is achievable with the target dataset alone. In addition, we derive the asymptotic distribution of the distance between our proposed estimator and the truth in Section E of the supplementary material, facilitating statistical inference and enabling the construction of confidence regions for the singular subspaces.

Our extensive numerical evaluations, detailed in Section \ref{sec:numerical}, assess the empirical performance of our methods. The results demonstrate that our approaches not only significantly outperform the ``single'' method, which applies PCA on individual tensors without aggregation, but also exceed the performance of the pooled method in both simulated heterogeneous settings and real data analyses. These findings highlight the superiority of our distributed approaches, which enhance estimation accuracy by effectively aggregating common information across diverse tensors amid data heterogeneity.

Our work, while related to distributed matrix PCA (e.g., \citealp{fan2019distributed, chen2022distributed}), tackles the more complex issue of distributed tensor PCA, which is inherently more challenging both practically and technically. Unlike matrices, tensors often involve three or more modes, increasing their dimensionality and complicating their analysis. The estimators of tensor PCA are often calculated from optimization algorithms iteratively over multiple modes of low-rank components and a core tensor, whereas the estimator of matrix PCA is based on non-iterative schemes. Due to the intricate statistical dependencies involved in iterative optimization algorithms, the establishment of theoretical guarantees for the distributed estimators in tensor PCA, including statistical convergence and inference, presents significantly greater challenges than its matrix counterpart. 

The contributions and novelty of our work are summarized as follows. 
\begin{itemize}
\item \textit{Modeling}: We introduce a new model to represent the distributed and heterogeneous environments encountered in tensor PCA. This model addresses a gap in the current literature, which has largely overlooked the complexities of real-world distributed data analysis where tensors are heterogeneous on different machines or servers.

\item \textit{Methodological Advances}: We propose novel distributed methods for tensor PCA that function effectively under both homogeneous and heterogeneous settings. These methods are designed to aggregate common components shared across different tensors, thereby enhancing estimation accuracy. Additionally, we expand our methodology to address scenarios akin to transfer learning, making considerable strides in distributed tensor analysis. To the best of our knowledge, we are the first to develop such methods for distributed tensor analysis, which, as discussed above, are highly different from the existing methods for matrix PCA.  

\item \textit{Theoretical Contributions}: We establish statistical guarantees for our distributed methods, demonstrating that they achieve a sharp statistical error rate that aligns with the minimax optimal rate possible in a non-distributed setting. Additionally, we calculate the asymptotic distribution of the proposed distributed estimator. This calculation aids in statistical inference and supports the construction of confidence regions for the singular subspaces. These contributions represent a theoretical advancement in utilizing aggregated common components, a finding that has not been previously documented in tensor analysis literature and thus enhances the existing computational approaches to distributed Tucker decomposition.
\end{itemize}

\vspace{-.6em}
\subsection{Related works}
\label{sec:literature}

This paper is situated at the intersection of two bodies of literature: tensor decomposition and distributed estimation and inference. 
The literature on both areas is broad and vast. 
Here we only review the closely related studies, namely tensor decomposition and distributed learning that provides
statistical guarantees. 
The readers are referred to \cite{bi2021tensors} and \cite{sun2021tensors} 
for recent surveys of statistical tensor learning. 

\paragraph{Tensor Decomposition.} 
Tensor decompositions have become increasingly important in machine learning, electrical engineering, and statistics. While low-rank matrix decomposition theory is well-established, tensors present unique challenges, with multiple notions of low-rankness \citep{kolda2009tensor}, primarily CANDECOMP/PARAFAC (CP) and multilinear/Tucker decompositions. Early work by \cite{richard2014statistical}, \cite{hopkins2015tensor}, and \cite{perry2016statistical} focused on rank-1 spiked tensor PCA models, proposing methods like tensor unfolding and sum of square optimization. \cite{zhang2018tensor} later developed a general tensor SVD framework based on Tucker decomposition, with \cite{xia2022inference} extending this to inference for tensor PCA. Recent advances include sparse tensor SVD \citep{zhang2019optimal}, binary tensor decomposition \citep{wang2020learning}, semi-parametric tensor factor models with covariates \citep{chen2024semiparametric}, and theoretical guarantees for Tucker/CP decomposition with correlated entries and missing data \citep{chen2024highCP,chen2024highTucker}.

{Part of our analysis under the homogeneous scenario utilizes the decomposition technique in \cite{xia2022inference} established for a non-distributed environment. 
Compared to their work, we highlight several major innovations in our study from the theoretical perspective.   In the homogeneous setting, we conduct a delicate bias-variance analysis for the error decomposition to establish the theoretical results for the distributed estimator. 
For the heterogeneous setting, we develop the decomposition for the common-component estimators using a different technique from \cite{xia2022inference}, due to the existence of distinct individual components. 
Additionally, the procedure for determining the common component ranks and the extension to data-adaptive weighted estimators in the knowledge transfer setting are, to our best knowledge, new to the literature.}

\paragraph{Distributed Estimation and Inference.}  To handle the challenges posed by massive and decentralized data, there has been a significant amount of recent literature developing distributed estimation or inference techniques for a variety of statistical problems \citep{li2013statistical, chen2014split,garber2017communication, lee2017communication,shi2018massive, jordan2019communication,fan2019distributed, volgushev2019distributed, chen2019quantile, li2020communication, yu2020simultaneous, 
tan2022communication,luo2022distributed, yu2022distributed, chen2022first, chen2024distributed}. We refer readers to \cite{gao2022review} for a comprehensive review.  
Among these works, our paper is most closely related to the literature on distributed matrix PCA. 
\cite{fan2019distributed} notably proposed a convenient one-shot approach that computes the top-$K$-dim eigenspace of the covariance matrix on each local machine and aggregates them on a central machine. They established a rigorous statistical guarantee showing that this approach achieves a sharp error rate as long as the sample size on each local machine is sufficiently large. An alternative multi-round method was simultaneously developed by \cite{garber2017communication}, which mainly estimates the first eigenspace by leveraging shift-and-invert preconditioning. Subsequently, \cite{chen2022distributed} proposed an improved multi-round algorithm for estimating the top-$K$-dim eigenspace that enjoys a fast convergence rate under weaker restrictions compared to \cite{fan2019distributed}. Other works focused on providing more general statistical error bounds \citep{zheng2022limit} or improving the communication efficiency \citep{huang2021communication,charisopoulos2021communication}.



\paragraph{Organizations.} The structure of this paper is as follows: Section \ref{sec:homo} introduces our distributed tensor PCA algorithm for homogeneous settings and establishes its statistical error rates. Section \ref{sec:hetero} adapts the algorithm for heterogeneous environments and includes theoretical analysis. Section \ref{sec:transfer} explores knowledge transferring from source tensors to a target tensor within a heterogeneous setting. Section \ref{sec:numerical} presents simulations and real data analyses to evaluate the performance of our methods. Some commonly used notations in this paper are defined in Section A of the supplementary material. Sections B and C provide supplemental theoretical results regarding Theorems \ref{thm:1} and \ref{thm:transfer}, respectively.  Section E, notably, presents the asymptotic distribution of our proposed estimator, featuring distributed inference. All technical proofs and additional numerical experiments are presented in Sections D and F of the supplementary material, respectively.

\vspace{-.8em}
\section{Estimation for Homogeneous Tensors}
\label{sec:homo}
In this section, we present the distributed tensor PCA algorithm for the homogeneous setting where each machine observes a tensor generated from the same underlying model. We first formulate the problem setup and illustrate the pooled estimator as a benchmark in Section \ref{sec:homo-setup}, followed by our estimator and the intuition behind its construction detailed in Section \ref{sec:method-homo}. Finally, Section \ref{sec:theory-homo} establishes theoretical guarantees on the statistical error rate for our estimator. 

\vspace{-.6em}
\subsection{Problem Setup and the Pooled Estimator}\label{sec:homo-setup}

Assume a $J$-mode tensor $\calT^* \in \R^{p_1 \times p_2 \times \cdots \times  p_J}$ has a Tucker decomposition
\begin{equation}\label{eq:model}
\calT^* = \calG \times_1 U_1 \times_2 U_2 \cdots \times_J U_J , \quad U_j \in \mathbb{O}^{p_j \times r_j}, \calG \in \R^{r_1 \times r_2 \times \cdots \times r_J},
\end{equation}
which decomposes $\calT^*$ into a core tensor $\calG$ multiplied by factor matrices $\{U_j\}$ with orthogonal columns.
Suppose we observe $L$ tensors $\{\calT_{\ell}\}_{\ell=1}^L$ that are distributed on $L$ machines and cannot be pooled together. For each machine $\ell$, the tensor $\calT_{\ell}$ is generated by $\calT_{\ell}=\calT^*+\calZ_{\ell}$, where $\calT^*$ is a common low-rank tensor of interest given by \eqref{eq:model}, and the noise tensor $\calZ_{\ell}$ has i.i.d. normal entries with mean zero and variance $\sigma^2$. Our goal is to estimate the singular subspace $\col(U_j)$, $j \in [J]$, i.e., the linear space spanned by the columns of $U_j$, under the distributed setting. 

The model defined in \eqref{eq:model} is non-identifiable since it is equivalent to the model with $\widetilde{\calG}=\calG \times_1 O^{\top}_{1} \times_2 O^{\top}_{2}\cdots \times_J O^{\top}_{J}$ and $\widetilde U_j = U_j O_j$, for any $O_j \in \mathbb{O}^{r_j \times r_j}$, $j\in[J]$. However, the singular subspace $\col(U_j)$ remains invariant under such orthogonal transformation and thus is identifiable. The singular subspace $\col(U_j)$ can further be represented by the \textit{projection matrix } $U_jU_j^{\top}$, which projects $\R^{p_j}$ onto $\col(U_j)$ and satisfies $\widetilde U_j \widetilde U_j ^{\top}=U_jO_jO_j^{\top}U_j^{\top}=U_jU_j^{\top}$.  To measure the estimation error between a singular subspace spanned by $U \in \mathbb{O}^{p \times r}$ and that spanned by an estimator $\hU \in \mathbb{O}^{p \times r}$, we use a metric $\rho\big(\hU,  U\big):=\big\|\widehat U \widehat U^{\top}- UU^{\top}\big\|_{\rm F}$, which is the Frobenius norm of the difference between the projection matrices of $U$ and $\hU$. We note that $\rho$ is equivalent to the well-known $\sin \Theta$ distance \citep{davis1970rotation} that measures the distance between the subspaces $\col(U)$ and $\col\big(\hU\big)$ using principal angles,  defined as  
\[\big\|\sin \Theta\big(U, \hU\big)\big\|_{\rm F} = \norm{\mathrm{diag} \big(\sin (\cos^{-1}\sigma_1), \dots, \sin(\cos^{-1}\sigma_r)\big)}_{\rm F}=\sqrt{r-\sum_{i=1}^{r}\sigma_i^2},\]
where $\sigma_1, \dots, \sigma_r$ are the singular values of $U^{\top}\hU$. The equivalence between $\rho$ and $\sin \Theta$ distance can be established by 
\[\rho^2\big(U, \widehat U\big)=\big\|UU^{\top}\big\|^2_{\rm F}+\big\| \widehat U\widehat U^{\top}\big\|_{F}^2-2\big\|U^{\top}\hU\big\|_{\rm F}^2=2r-2\sum_{i=1}^{r}\sigma_i^2=2\big\|\sin \Theta\big(U, \widehat U\big)\big\|_{F}^2.\] 

If the tensors $\{\calT_{\ell}\}$ were allowed to be pooled onto a central machine, a standard way for estimating the singular space $\col(U_j)$ would be to conduct a Tucker decomposition on the averaged tensor $\overline{\calT}=\frac{1}{L}\sum\limits_{\ell=1}^{L}\calT_{\ell}=\calT^*+\overline{\calZ}$, where $\overline{\calZ} =\frac{1}{L}\sum\limits_{\ell=1}^{L}\calZ_{\ell}$ is a tensor with i.i.d. $\calN(0, \frac{\sigma^2}{L})$ entries. Denote the estimator obtained through this method as $\widehat U_{ \mathrm{pooled}, j}$, $j \in [J]$, which serves as a benchmark against which we evaluate the effectiveness of our approach in the distributed environment.

\vspace{-.6em}
\subsection{Distributed Tensor PCA for Homogeneous Tensors}\label{sec:method-homo}

\begin{algorithm}[!t] 
\spacingset{1.2}
\caption{Distributed Tensor PCA for Homogeneous Data} 
\label{alg:homo} 
\vspace*{0.08in} {\bf Input:}
Tensors distributed on local machines $\{\calT_{\ell}\}$ and initial estimators $\big\{\hU_{1,\ell}^{(0)}, \hU_{2,\ell}^{(0)}, \dots, \hU_{J,\ell}^{(0)}\big\}$, for all $\ell \in [L]$. 

{\bf Output:}  Estimators $\big\{\hU_1, \hU_2, \dots, \hU_J\big\}$.
\begin{algorithmic}[1]  
\For{$\ell=1,2,\dots,L$} \do \\
\For{$j=1,2,\dots,J$}
\State Compute a local estimator $	\hU_{j,\ell} = \mathrm{svd}_{r_j}(M_{j,\ell})$, where $M_{j,\ell}$ is defined in \eqref{eq:local_homo};
\State Send  $\hU_{j,\ell}$ to the central machine; 
\EndFor
\EndFor
\For{$j=1,2,\dots, J$}
\State On the central machine, compute  $\hU_j = \mathrm{svd}_{r_j}\Big[\frac{1}{L}\sum\limits_{\ell=1}^{L}\hU_{j, \ell}\hU^{\top}_{j, \ell}\Big]$;
\EndFor
\end{algorithmic} 
\end{algorithm} 
We first propose Algorithm \ref{alg:homo} for estimating the singular subspace spanned by $U_j$ in the distributed setting. Algorithm \ref{alg:homo} starts with an initial estimate $\{\hU_{j,\ell}^{(0)}\}$ that can be obtained from a prototypical Tucker decomposition algorithm, for instance, higher-order SVD (HOSVD) or higher order orthogonal iteration (HOOI) \citep{de2000multilinear, de2000best}, on the corresponding individual tensors $\{\calT_\ell\}$. 

Given initial estimators $\{\hU_{j,\ell}^{(0)}\}$, we obtain a local estimator for $U_j$ by computing the left singular value matrix of  $M_{j,\ell}$ on each machine $\ell$, where the matrix $M_{j,\ell} \in \R^{p_j \times (r_1r_2\cdots r_J / r_j)}$ is defined as the matricization of a projected version of $\calT_{\ell}$, given by 
\begin{equation}\label{eq:local_homo}
M_{j,\ell} = \calM_j\big(\calT_{\ell} \times_{1} \hU_{1,\ell}^{(0)\top} \times_{2} \hU_{2,\ell}^{(0)\top} \cdots \times_{j-1}\hU_{j-1,\ell}^{(0)\top} \times_{j+1}\hU_{j+1,\ell}^{(0)\top}\cdots \times_{J}\hU_{J,\ell}^{(0)\top}\big).
\end{equation}
Under mild conditions, we show that $M_{j,\ell}$ approximately equals $U_j \calM_j(\calG)$, and therefore, the left singular vector matrices of $\{M_{j,\ell}\}$ provide estimators $\{\hU_{j,\ell}\}$ for $U_j$. The local estimators  $\{\hU_{j,\ell}\}$ are then sent to a central machine and further aggregated by averaging the projection matrices $\hU_{j,\ell}\hU_{j,\ell}^{\top}$ over all $\ell$. Finally, we compute the left singular vectors of the averaged matrix as the output estimator $\hU_{j}$. The communication cost of Algorithm \ref{alg:homo} is of the order $O(\sum_{j=1}^{J}p_jr_j)$, which is a significant reduction from $O(\prod_{j=1}^{J}p_j)$, the communication cost for transferring the individual tensors themselves across machines.
\begin{remark}\label{rmk:UU'}
In the aggregation step of Algorithm \ref{alg:homo}, we average the projection matrices  $\hU_{j,\ell}\hU_{j,\ell}^{\top}$ instead of the singular vectors $\hU_{j,\ell}$ due to the non-identifiability of $\hU_{j,\ell}$, as discussed in Section \ref{sec:homo-setup}. 
For instance, if a singular value has a multiplicity greater than 1, the corresponding singular vectors $\hU_{j,\ell}$ can be any orthonormal basis spanning the same singular subspace associated with that repeated singular value. Even if all singular values are distinctive,  there is still a sign ambiguity issue, i.e., both $\hU_{j,\ell}$  and $-\hU_{j,\ell}$ may be obtained from the SVD of the same matrix, which may lead to cancellations in the averaging of $\hU_{j,\ell}$. In contrast, averaging the projection matrices avoids these issues and provides a valid estimate for the singular space spanned by $U_j$. Moreover, since the estimated projection matrices $\big\{\hU_{j,\ell}\hU_{j,\ell}^{\top}\big\}_{\ell \in [L]}$ are not guaranteed to represent the same subspace, their average may have a rank larger than $r_j$. Therefore, we add an additional SVD in the final step to obtain a low-rank approximation of the averaged projection matrices, denoted as $\hU_j$, and output it as the final estimator. 
\end{remark}

\begin{remark}\label{rmk:rank} In practice, the ranks $\{r_j\}_{j\in [J]}$ may be unknown and need to be specified for Algorithm \ref{alg:homo} in the distributed way. Since the local matrix $M_{j,\ell}$ in \eqref{eq:local_homo} provides an approximation for $U_j\calM_j(G)$ which has rank $r_j$, one may first compute a consistent rank estimator $\widehat{r}_{j,\ell}$ for $M_{j,\ell}$ using existing rank determination methods (e.g., \citealp{choi2017selecting, charkaborty2022testing, han2022rank}) and further aggregate the locally estimated ranks $\{\widehat{r}_{j,\ell}\}_{\ell \in [L]}$, for example, by averaging, to obtain a more accurate estimate for $r_j$. Another approach is to overparametrize each local model by specifying a conservative rank $\widehat{r}_{j,\ell} \geq r_j$, as studied in \cite{xu2023power}, and then aggregate conservatively, for example, by choosing $\widehat{r}_j$ as the maximum or certain quantile of $\{\widehat{r}_{j,\ell}\}_{\ell \in [L]}$. It is a potentially interesting future direction to investigate the performance of Algorithm \ref{alg:homo} under overparametrization.
\end{remark}

\vspace{-.6em}
\subsection{Theoretical Guarantee}
\label{sec:theory-homo}

In this section, we provide theoretical guarantees for the statistical performance of Algorithm \ref{alg:homo}.
For $j \in [J]$, let $\Lambda_j$ be the
$r_j \times r_j$ singular value matrix of $\calM_j(\calG)$. Let $r=\max_j r_j$ and $p=\max_j p_j$. Moreover, define $\lambda_{\max}$, $\lambda_{\min}$ to be the maximum and minimum singular value of $\Lambda_j$ across all $j \in [J]$, and let $\kappa_0 = \lambda_{\max} \lambda^{-1}_{\min}$. 

\begin{thm}
Assume that there exist constants $C_1$, $c_1$, and $C_2$ such that, with probability at least $1-C_1e^{-c_1p}$, $\big\|\hU_{j,\ell}^{(0)}\hU_{j,\ell}^{(0)\top}- U_jU_j^{\top}\big\|_{2}\leq C_2\sqrt{p}\sigma\lambda_{\min}^{-1}$ for any $j \in [J]$ and $\ell \in [L]$. If $p_j \asymp p$ for all $j$, $L =O(p^{c_3})$ for some $c_3>0$, $r^{J-1}=O(p)$, $\kappa_0 = O(1)$, and $ \sqrt{pr}=o(\lambda_{\min}/\sigma)$, then we have
\begin{equation}
\label{eq:thm1-bound}
\sup_{j}\rho\big(\hU_j, U_j\big) \leq \widetilde C_2\bigg( \frac{\sigma}{\lambda_{\min}}\sqrt{\frac{pr}{L}}+ \frac{pr\sigma^2}{\lambda_{\min}^{2}}\bigg),
\end{equation}
with probability at least $1- \widetilde C_1e^{- \widetilde c_1 p}$ for some constants $ \widetilde C_1$ , $ \widetilde c_1$, and $ \widetilde C_2$, where $\{\hU_j\}$ is the output of Algorithm \ref{alg:homo}.
\label{thm:1}
\end{thm}

Theorem \ref{thm:1} establishes the error rate of the estimators $\big\{\hU_j\big\}$ obtained by Algorithm \ref{alg:homo}, which can be explained by a bias-variance decomposition. We take $J=3$, $j=1$ for an example. On each machine $\ell$, the estimation error of the local estimator $\hU_{1,\ell} $ can be decomposed as
\begin{equation}\label{eq:b-v-decomp}
\hU_{1,\ell}\hU^{\top}_{1,\ell} - U_1U_1^{\top} = U_1\Lambda_{1}^{-2}G_1 \big(U_{2}^{\top} \otimes U_{3}^{\top}\big)Z^{\top}_{1,\ell}U_{1\perp}U_{1\perp}^{\top} +  U_{1\perp}U_{1\perp}^{\top}Z_{1,\ell} \left(U_{2} \otimes U_{3} \right)G_1^{\top}\Lambda_{1}^{-2}U_1^{\top}+ R_{1,\ell},
\end{equation}
where $G_1 = \calM_1(\calG)$, $Z_{1,\ell}=\calM_1(\calZ_{\ell})$, $U_{1\perp}$ is the orthogonal complement of $U_1$, and $R_{1,\ell}$ is a remainder term. When the signal-to-noise ratio (SNR) satisfies $\sqrt{pr}=o(\lambda_{\min}/\sigma)$, the Frobenius norm of the first two mean-zero terms on the RHS of \eqref{eq:b-v-decomp} is of the order $O\left(\sqrt{pr}\sigma \lambda_{\min}^{-1}\right)$, and the remainder term $R_{1,\ell}$ has a higher order $O\big(pr\sigma^2\lambda_{\min}^{-2}\big)$. By averaging the projection matrices on all machines, the order of the first two terms can be reduced to $O\big(\sqrt{pr}\sigma \lambda_{\min}^{-1}L^{-1/2}\big)$, while that of $R_{1,\ell}$ does not change, leading to the error rate in \eqref{eq:thm1-bound}. When the SNR is sufficiently large with respect to $L$, the first term in \eqref{eq:thm1-bound} dominates the second one, which leads to the following corollary.

\begin{cor}\label{cor:1}
Under the same assumptions in Theorem \ref{alg:homo}, if we further assume that $\lambda_{\min} / \sigma \gtrsim \sqrt{prL} $, then we have
\begin{equation}\label{eq:cor1-bound}
\sup_{j}\rho\big(\hU_j, U_j\big) \leq \frac{\widetilde{C}_2\sigma}{\lambda_{\min}}\sqrt{\frac{pr}{L}},
\end{equation}
with probability at least $1- \widetilde C_1e^{- \widetilde c_1 p}$ for some constants $ \widetilde C_1$ , $ \widetilde c_1$, and $ \widetilde C_2$.
\end{cor}

Corollary \ref{cor:1} shows that, when the SNR satisfies $ \lambda_{\min} / \sigma \gtrsim \sqrt{prL} $,  the estimator $\hU_j$ achieves the error rate of  $O(\sqrt{pr}\sigma \lambda_{\min}^{-1}L^{-1/2})$, which matches the minimax optimal lower bound. 
Concretely, for a tensor $\calT=\calT^*+\calZ$ with $\calT^*$ satisfying \eqref{eq:model} and $\calZ$ having i.i.d. $\calN(0, \sigma^2)$ entries, Theorem 3 in \cite{zhang2018tensor} shows that
\begin{equation}\label{eq:lower-bound}
\inf_{\hU}\sup_{\calT^* \in \mathcal{F}_{\bm{p}, \bm{r}}(\lambda)} \E r_j^{-1/2}\big\|\sin \Theta (\hU, U_j)\big\|_{\rm F} \gtrsim \Big(\frac{\sqrt{p_j}}{\lambda/\sigma} \wedge 1\Big),
\end{equation} 
where $\mathcal{F}_{\bm{p}, \bm{r}}(\lambda)$  is the class of tensors with dimension $\bm{p}=(p_1, \dots, p_J)$, rank $\bm{r}=(r_1,\dots,r_J)$, and minimum singular value $\lambda$ over all matricizations of $\calT^*$. 
Recall that if we are allowed to pool all the tensors $\{\calT_{\ell}\}$ on a single machine and average them, the averaged tensor satisfies $\overline{\calT} = \calT^* + \overline{\calZ}$, where $\overline{\calZ}$ has i.i.d. $\calN(0, \frac{\sigma^2}{L})$ entries. By \eqref{eq:lower-bound} and the equivalence between the $\rho$ distance and $\sin \Theta$ distance, the minimax error rate that the pooled estimator $\hU_{\mathrm{pooled}, j}$ (defined in Section \ref{sec:homo-setup}) can achieve is the same as the rate in \eqref{eq:cor1-bound}, which is the optimal rate one can expect in a non-distributed setting. Therefore, our proposed method enjoys a sharp rate when the SNR is sufficiently large.

\begin{remark}\label{rmk:lower-bound} 
{The second-order term $pr\sigma^2\lambda_{\min}^{-2}$  in Theorem \ref{thm:1} cannot be improved. We establish this through a lower bound result in Section B of the supplementary material, which shows that when the condition  $\lambda_{\min} / \sigma \gtrsim \sqrt{prL} $ is violated, there exists a target tensor $\calT^*$ and a set of initial estimators that, despite satisfy all the assumptions in Theorem \ref{thm:1}, achieve only the sub-optimal convergence rate $pr\sigma^2\lambda_{\min}^{-2}$. This result demonstrates the sharpness of the error rate in \eqref{eq:thm1-bound} and implies the necessity of the condition $\lambda_{\min} / \sigma \gtrsim \sqrt{prL} $ in Corollary \ref{cor:1} for achieving the minimax optimal rate.}
\end{remark}

The assumption for the initial estimators $\big\{\hU_{j, \ell}^{(0)}\big\}$ in Theorem \ref{thm:1} is consistent with Assumption 1 in \cite{xia2022inference}, which can be achieved by the HOOI algorithm under a requirement that $\lambda_{\min}/\sigma \gtrsim p^{J/4}$ \citep{zhang2018tensor}. In some scenarios, a certain number of tensors can be pooled together and averaged into a new tensor $\calT'_{\ell}$ with a noise level $\sigma'_{\ell} < \sigma$. The requirement can be relaxed into $\lambda_{\min}/\min_{\ell} \sigma'_{\ell} \gtrsim p^{J/4}$ in such scenarios by computing the initial estimators using the locally-aggregated tensor with the smallest noise level. On the other hand, with initial estimators that satisfy the assumption in Theorem \ref{thm:1}, our method only requires a weaker condition $\lambda_{\min} / \sigma \gtrsim \sqrt{prL} $ to achieve the optimal rate, as shown in Corollary \ref{cor:1}. 

Moreover, we do not require sample splitting for initialization, that is, the set of tensors $\{\calT_{\ell}\}$ used for initialization can be the same as that used for the distributed estimation procedure in Algorithm \ref{alg:homo}. Indeed, our theoretical error bound \eqref{eq:thm1-bound} uniformly holds for all initial estimators that satisfy the assumption in Theorem \ref{thm:1}, since the first two terms in decomposition \eqref{eq:b-v-decomp} do not rely on the initial estimators, and the remainder term $R_{1,\ell}$ can be uniformly bounded.

\vspace{-.8em}
\section{Estimation for Heterogeneous Tensors}
\label{sec:hetero}

In this section, we generalize Algorithm \ref{alg:homo} to a heterogeneous setting where we allow different truth tensors $\calT^*$ on different machines. Suppose we observe $L$ tensors $\{\calT_{\ell}\}_{\ell=1}^L$ distributed on $L$ machines, and $\calT_{\ell}=\calT_{\ell}^*+\calZ_{\ell}$, where $\calZ_{\ell}$ has i.i.d. $\calN(0, \sigma^2)$ entries. Assume the truth $\calT^*_{\ell} \in \R^{p_1 \times p_2 \cdots \times p_J}$ on machine $\ell$ satisfies a Tucker decomposition
\begin{equation}\label{eq:model_hetero}
\calT_{\ell}^* = \calG_{\ell} \times_1 [U_1 \; V_{1,\ell}] \times_2 [U_2 \; V_{2, \ell}] \cdots \times_J [U_J \; V_{J, \ell}],
\end{equation}
where $U_j \in \mathbb{O}^{p_j \times r_{j,U}}$, $V_{j,\ell}\in\mathbb{O}^{p_j\times r_{j,V,\ell}}$,  $U_j^{\top}V_{j,\ell}=0$, $\calG_{\ell} \in \R^{ r_{1,\ell} \times  r_{2,\ell} \cdots \times  r_{J,\ell}}$, and $r_{j,\ell}=r_{j, U}+r_{j,V,\ell}$. For each $j$, the component $U_j$ spans a common singular subspace shared by all tensors $\{\calT_{\ell}\}$, while  $V_{j,\ell}$ spans an individual subspace specific to each tensor. Moreover, the core $\calG_{\ell}$ is allowed to be different for different $\ell$. The goal is to estimate the shared singular subspace and the individual subspace separately for $j \in [J]$ and $\ell \in [L]$. Note that in this section, we assume the noise level $\sigma^2$ is identical on all machines for clear presentation. The heterogeneity on $\sigma^2$  will be further investigated in Section \ref{sec:transfer}.

In addition to \eqref{eq:model_hetero}, we need a regularity condition to ensure the identifiability of the model. As discussed in Section \ref{sec:homo-setup},  the model defined in \eqref{eq:model} is non-identifiable since it is equivalent to the model with $\widetilde{\calG}=\calG \times_1 O^{\top}_{1} \times_2 O^{\top}_{2}\cdots \times_J O^{\top}_{J}$ and $\widetilde U_j = U_j O_j$ , for any $O_j \in \mathbb{O}^{r_j \times r_j}$, $j\in[J]$. In the homogeneous case, the non-identifiability has no impact on estimation since the singular subspace $\col(U_j)$ remains invariant under orthogonal transformation. 
However, under the heterogeneous setting \eqref{eq:model_hetero}, the partition of the common component $U_j$ and the individual component $V_{j,\ell}$ is not orthogonally invariant, which necessitates identifying a fixed $O_j$ for each $j$. Therefore, we require that the core tensors $\calG_{\ell}$ satisfy
\begin{equation}\label{eq:identification} \calM_j(\calG_{\ell})\calM_j(\calG_{\ell})^{\top}=\Lambda_{j,\ell}^2,
\end{equation}
where $\Lambda_{j,\ell}$ is a diagonal matrix with decreasing singular values for all $j \in [J]$, $\ell \in [L]$. 

\begin{remark}\label{rmk:IC} 
{Condition \eqref{eq:identification} is closely related to the all-orthogonality condition in the HOSVD framework \citep{de2000multilinear}. The all-orthogonality condition states that the mode-$j$ subtensors of the core tensor (i.e., the subtensor obtained by fixing the $j$-th index and letting the other indices be free) are mutually orthogonal with respect to the scalar product, which leads to orthogonal rows in $\calM_j(\calG_{\ell})$  and implies \eqref{eq:identification}.}

{To clearly illustrate how this condition guarantees the uniqueness of $\col(U_j)$, we first restate the non-identifiablility of our heterogeneous model,
$
\calT^*_{\ell} = \calG_{\ell} \times_1 [U_{1,\ell} \, V_{1,\ell}]  \times_2 [U_{2,\ell} \, V_{2,\ell}] \cdots \times_J [U_{J, \ell} \,V_{J, \ell}]
$,
where we initially do not assume that the components $U_{j,\ell}$ are shared across the tensors. This model is non-identifiable: it remains unchanged under the transformation $\widetilde{\calG}_{\ell}=\calG_{\ell} \times_1 O^{\top}_{1,\ell} \times_2 O^{\top}_{2,\ell}\cdots \times_J O^{\top}_{J,\ell}$ and $\big[\widetilde U_{j,\ell} \, \widetilde V_{j,\ell}\big]= [U_{j,\ell} V_{j,\ell}] O_{j,\ell}$, for any orthogonal matrices $\{O_{j,\ell}\in \mathbb{O}^{r_{j,\ell} \times r_{j,\ell}}\}_{j \in [J]}$.  
Note that $\calM_j\big(\calG_{\ell}\big)=O_{j,\ell}\calM_j\big(\widetilde \calG_{\ell}\big)\big(\bigotimes\limits_{j' \neq j} O_{j',\ell}^{\top}\big)$. There exists a unique $O_{j,\ell}$---the left singular vector matrix of $\calM_j\big(\widetilde \calG_{\ell}\big)$---such that $\calM_j(\calG_{\ell})\calM_j(\calG_{\ell})^{\top}=O_{j,\ell}^{\top}\calM_j\big(\widetilde \calG_{\ell}\big)\calM_j\big(\widetilde \calG_{\ell}\big)^{\top}O_{j,\ell}$ is a diagonal matrix with decreasing singular values. This uniquely determines $[U_{j,\ell} \, V_{j,\ell}]=\big[\widetilde U_{j,\ell} \, \widetilde V_{j,\ell}\big]O_{j,\ell}^{\top}$. Based on this condition, we further assume that $U_{j,\ell}=U_j$ to be the common components shared across all tensors, yielding model \eqref{eq:model_hetero}.} 
\end{remark}

\vspace{-.6em}
\subsection{Distributed Tensor PCA for Heterogeneous Tensors}

\begin{algorithm}[!t] 
\spacingset{1.4}
\caption{Distributed Tensor PCA for Heterogeneous Data} 
\label{alg:hetero} 
\vspace*{0.08in} {\bf Input:}
Tensors distributed on local machines $\{\calT_{\ell}\}$ and initial estimators $\big\{\big[\hU_{1,\ell}^{(0)} \, \hV_{1,\ell}^{(0)}\big], \big[\hU_{2,\ell}^{(0)} \, \hV_{2,\ell}^{(0)}\big], \dots, \big[\hU_{J,\ell}^{(0)} \, \hV_{J,\ell}^{(0)}\big]\big\}$, where $\ell=1,2,\dots,L$. 

{\bf Output:}  Estimators $\big\{\hU_1, \hV_{1, \ell}, \hU_2, \hV_{2, \ell}, \dots, \hU_J, \hV_{J, \ell}\big\}$.
\begin{algorithmic}[1]  
\For{$\ell=1,2,\dots,L$} \do \\
\For{$j=1,2,\dots,J$}
\State Compute a local estimator 
$\hU_{j, \ell}=\mathrm{svd}_{r_{j, U}}\big(\widetilde M_{j,\ell}\big)$, where $\widetilde M_{j,\ell}$ is defined in \eqref{eq:local_hetero};
\State Send  $\hU_{j,\ell}$ to the central machine; 
\EndFor
\EndFor
\For{$j=1,2,\dots, J$}
\State On the central machine, compute  $\hU_j = \mathrm{svd}_{r_{j, U}}\Big[\frac{1}{L}\sum\limits_{\ell=1}^{L}\hU_{j, \ell}\hU^{\top}_{j, \ell}\Big]$;
\EndFor
\State Send $\hU_j$ to all machines;
\For{$\ell=1,2,\dots, L$}
\State Compute 
$\hV_{j, \ell}=\mathrm{svd}_{r_{j, V, \ell}}\Big[\big(I_{p_j}-\hU_{j}\hU_{j}^{\top}\big)\widetilde M_{j,\ell}\Big]$;
\EndFor
\end{algorithmic} 
\end{algorithm} 

We propose Algorithm \ref{alg:hetero} for estimating the subspace spanned by $U_j$ and $V_{j,\ell}$ under the heterogeneous setting. The initial estimators $\big\{[\hU_{j,\ell}^{(0)}\; \hV_{j,\ell}^{(0)}]\big\}$ can be obtained in the same way as the homogeneous case, for example, by the HOOI algorithm, and do not need to be partitioned into $\hU_{j,\ell}^{(0)}$ and $\hV_{j,\ell}^{(0)}$. Similar to Algorithm \ref{alg:homo}, we first obtain a local estimator $\hU_{j,\ell}$ for $U_j$ on each machine $\ell$ by taking the top singular vectors of $\widetilde{M}_{j,\ell}$, where
\begin{align}
\widetilde{M}_{j,\ell}&:=\calM_j\big(\calT_{\ell} \times_{1} \big[\hU_{1,\ell}^{(0)} \; \hV_{1,\ell}^{(0)}\big]^{\top} \cdots \times_{j-1} \big[\hU_{j-1,\ell}^{(0)} \; \hV_{j-1,\ell}^{(0)}\big]^{\top} \times_{j+1} \big[\hU_{j+1,\ell}^{(0)} \; \hV_{j+1,\ell}^{(0)}\big]^{\top} \cdots 
\times_{J}\big[\hU_{J,\ell}^{(0)} \; \hV_{J,\ell}^{(0)}\big]^{\top} \big)\notag\\
&\approx [U_j \; V_{j,\ell}] \calM_j(\calG_{\ell}).	\label{eq:local_hetero}
\end{align}
We then send $\hU_{j,\ell}$ to the central machine and aggregate the projection matrices $\hU_{j,\ell}\hU_{j,\ell}^{\top}$ to compute a global estimator $\hU_j$.  To estimate $\col(V_{j,\ell})$, we send $\hU_{j}$ back to each machine and\\
compute the top singular vectors of $\big(I_{p_j}-\hU_{j}\hU_{j}^{\top}\big)\widetilde M_{j,\ell}$, the projection of $\widetilde{M}_{j,\ell}$ onto the orthogonal space of $\col(\hU_j)$, whose top singular subspace provides a local estimator for $\col(V_{j,\ell})$.

\vspace{-.6em}
\subsection{Theoretical Guarantee}
\label{sec:theory-hetero}

In this section, we establish the statistical error rate for the estimators in Algorithm \ref{alg:hetero}. For $j\in[J]$, $\ell \in [L]$, let $\Lambda_{j,\ell}$ be the
$ r_{j,\ell} \times  r_{j,\ell} $ singular value matrix of $\calM_j(\calG_{\ell})$ defined in \eqref{eq:identification}.  Define $\lambda_{\max}$, $\lambda_{\min}$ to be the maximum and minimum singular value over all $\Lambda_{j,\ell}$, respectively, and let $\kappa_0 = \lambda_{\max} \lambda^{-1}_{\min}$.  Moreover, define $\Delta=\min_{j,\ell}\{\lambda_{r_{j,U},j,\ell}-\lambda_{r_{j,U}+1,j,\ell}\}$ and $\kappa=\lambda_{\max}/\Delta$, where $\lambda_{r,j,\ell}$ denotes the $r$-th largest singular value of $\Lambda_{j,\ell}$. Additionally, let $r=\max_{j,\ell}  r_{j, \ell}$ and $r_V=\max_{j, \ell}r_{j, V, \ell}$.

\begin{thm}
Assume that there exist constants $C_1, c_1, C_2$ such that, with probability at least $1-C_1e^{-c_1p}$, $\big\|\hU_{j,\ell}^{(0)}\hU_{j,\ell}^{(0)\top}+\hV_{j,\ell}^{(0)}\hV_{j,\ell}^{(0)\top}- U_jU_j^{\top}-V_{j, \ell}V_{j, \ell}^{\top}\big\|_{2}\leq C_2\sqrt{p}\sigma\lambda_{\min}^{-1}$ for all $j, \ell$. If $p_j \asymp p$ for all $j$, $L =O(p^{c_3})$ for some $c_3>0$, $r^{J-1} =O(p)$, $\kappa_0 =O(1)$, $\kappa=O(1)$, and $\sqrt{pr}=o(\min(\Delta, \lambda_{\min})/\sigma)$, then we have
\begin{equation}
\label{eq:thm2-U}
\sup_j \rho\big(\hU_{j}, U_j\big)\leq \widetilde C_2\left(\sqrt{\frac{pr}{L}}\frac{\sigma}{\Delta}+\frac{pr\sigma^2}{\Delta^2}\right),
\end{equation}
and
\begin{equation}\label{eq:thm2-V}
\sup_{j, \ell}	\rho\big(\hV_{j, \ell}, V_{j, \ell}\big) \leq \widetilde C_2\left(\frac{\sqrt{pr_{V}}\sigma(1+\sqrt{r/L})}{\lambda_{\min}}+\frac{\sqrt{r_V}pr\sigma^2}{\lambda_{\min}^2}\right),
\end{equation}
with probability at least $1- \widetilde C_1e^{- \widetilde c_1 p}$ for some constants $ \widetilde C_1$ , $ \widetilde c_1$, and $ \widetilde C_2$, where $\hU_j$ and $\hV_{j, \ell}$'s are the output of Algorithm \ref{alg:hetero}.
\label{thm:hetero}
\end{thm}

The statistical rate in \eqref{eq:thm2-U} is similar to the rate established in Theorem \ref{thm:1} in the homogeneous case with the SNR changing from $\lambda_{\min}/\sigma$ to $\Delta/\sigma$. Indeed, the quantity $\Delta$ denotes the minimum gap between the minimum singular value corresponding to $U_j$ and the maximum singular value corresponding to $V_{j,\ell}$ over all $j, \ell$, which indicates the strength of the signal for estimating $U_j$ in this heterogeneous setting and is equal to $\lambda_{\min}$  in the homogeneous setting where $V_{j,\ell}=0$. Similar to Corollary \ref{cor:1}, when SNR is sufficiently large, i.e., $\Delta/\sigma \gtrsim \sqrt{prL}$, our estimator $\hU_j$ enjoys the sharp rate $O\big(\sqrt{pr}\sigma\Delta^{-1}L^{-1/2}\big)$. For the local estimator $\hV_{j,\ell}$, the rate established in \eqref{eq:thm2-V} matches the local rate for $[U \; V]$, $\sqrt{pr}\sigma\lambda_{\min}^{-1}$, under a mild condition that $r_V \lesssim L$ and $(\lambda_{\min} /\sigma) \gtrsim \sqrt{prr_V}$.

{It is worth noting that, the definition of $\Delta$ is related to our theoretical framework that assumes the common components $U_j$ correspond to the large singular values for clear presentation. While extension to settings where $U_j$'s are not necessarily the top singular vectors is technically feasible, we leave it to future work as the analysis would involve introducing numerous additional notations and case-specific definitions without providing significant new insights. We anticipate that the error rate will remain the same as Theorem \ref{thm:hetero}, requiring only a redefinition of the signal strength $\Delta$. This redefined $\Delta$ still represents the minimum eigenvalue gap between singular values of common and individual components, but writing out its explicit form would be notationally cumbersome. } 
\vspace{-.6em}
\subsection{Rank Determination}\label{sec:rank-hetero}

{In this section, we provide an estimation method for determining the ranks $r_{j, U}$ and $r_{j, V, \ell}$ in model \eqref{eq:model_hetero} if they are unknown in practice. Note that in Remark \ref{rmk:rank}, we discuss how to use the existing rank determination methods to obtain an accurate estimate for $r_j$ under the homogeneous setting. Similarly, these methods can be applied to estimate the rank of $\widetilde{M}_{j, \ell}$ defined in \eqref{eq:local_hetero}, i.e., $r_{j, U}+r_{j,V,\ell}$. 
Therefore, we focus on how to obtain a consistent estimate for $r_{j, U}$, assuming that we have already correctly estimated $\widehat{r}_{j, \ell}=r_{j, U}+r_{j,V,\ell}$ for all $\ell$. For each tensor $\calT_{\ell}$, compute the left $\widehat r_{j, \ell}$ singular vectors of $\widetilde{M}_{j,\ell}$ and
denote it by $\hW_{j, \ell}$. We then compute $\overline{W}_{j}=\frac{1}{L}\sum_{\ell=1}^{L}\hW_{j, \ell}\hW_{j, \ell}^{\top}$ and its singular values $\sigma_1\big(\overline{W}_{j}\big)\geq \sigma_2\big(\overline{W}_{j}\big)\geq\dots \geq \sigma_{p_j}\big(\overline{W}_{j}\big)$. Finally, we estimate $r_{j, U}$ as 
\[\widehat{r}_{j, U}=\max\braces{k \mid \sigma_{k}\big(\overline{W}_{j}\big) \geq 1 - \delta_0},\]
for some constant $\delta_0>0$. 
The motivation behind $\widehat{r}_{j, U}$ stems from the fact that $\overline{W}_{j}$ is a consistent estimator for the matrix $U_jU_j^{\top}+\frac{1}{L}\sum_{\ell=1}^{L}V_{j, \ell}V_{j, \ell}^{\top}$. Its eigenvalues exhibit a distinctive pattern: those corresponding to $U_jU_j^{\top}$ are 1, whereas those corresponding to $\frac{1}{L}\sum_{\ell=1}^{L}V_{j, \ell}V_{j, \ell}^{\top}$ should be less than 1 due to the heterogeneity among $\col(V_{j, \ell})$ subspaces. To formalize this intuition, we introduce a heterogeneity measure $\eta_{V}:=\min_{j, \ell}\frac{1}{L-1}\sum_{\ell'\neq \ell}(1-\norm{V_{j, \ell}^{\top} V_{j, \ell'}}_2)$, which quantifies the dissimilarity among the individual components $V_{j, \ell}$ across different tensors. The largest singular value of $V_{j, \ell}^{\top} V_{j, \ell'}$, i.e., $\norm{V_{j, \ell}^{\top} V_{j, \ell'}}_2$, represents the cosine of the smallest principal angle between the subspaces $\col(V_{j, \ell})$ and $\col(V_{j, \ell'})$.  A larger value of $\norm{V_{j, \ell}^{\top} V_{j, \ell'}}_2$ indicates a smaller distance between the corresponding principal vectors of these subspaces. Consequently, $\eta_V$ can be interpreted as the minimum (over $\ell$) of the average distance between the principal vectors of $\col(V_{j, \ell})$ and $\col(V_{j, \ell'})$, for $\ell' \neq \ell$. }

{We then present the consistency results for $\widehat{r}_{j, U}$, whose proof is provided in Section D.2 of the supplementary material.
\begin{thm}\label{thm:rankU}
Under the assumptions of Theorem \ref{thm:hetero}, further assume that  
\[\eta_V=\min_{j, \ell}\frac{1}{L-1}\sum_{\ell'\neq \ell}(1-\norm{V_{j, \ell}^{\top} V_{j, \ell'}}_2)\geq \eta_0, \] 
for some constant $\eta_0>0$. For any $0<\delta_0<\eta_0/3$, we have
\[\Prob\left(\widehat{r}_{j, U}=r_{j, U}\right) \geq 1-\widetilde{C}_1e^{-\widetilde c_1p},\]
for some constants $\widetilde c_1, \widetilde C_1>0$.
\end{thm}}

{Theorem \ref{thm:rankU} demonstrates that, for a sufficiently small constant $\delta_0$, the estimator $\widehat{r}_{j, U}$ is consistent under the
requirement $\eta_{V} \geq \eta_0$. This requirement ensures adequate heterogeneity among $V_{j, \ell}$ components. Indeed, the quantity $\eta_{V}=0$ if there exists a common principal vector across all $\col(V_{j, \ell})$, as this implies $\norm{V_{j, \ell}^{\top} V_{j, \ell'}}_2=1$ for all $\ell' \neq \ell$. Conversely, $\eta_{V}>0$ indicates that no principal vector is universally shared across all $\col(V_{j, \ell})$ subspaces, thereby quantifying the heterogeneity among the individual components.}

\vspace{-.8em}
\section{Knowledge Transfer in Distributed Tensor PCA}\label{sec:transfer}

In this section, we explore the task of transferring knowledge from source locations to a target location within a heterogeneous setting. Knowledge Transfer seeks to enhance learning performance at a target site by leveraging insights from related source tasks. To illustrate our approach clearly, we concentrate on transferring knowledge between a single source dataset and a target dataset. The generalization to knowledge transfer across multiple tasks is provided in Section C of the supplementary material. 

\vspace{-.6em}
\subsection{Transferred Tensor PCA}

\begin{algorithm}[!t] 
\spacingset{1.4}
\caption{Transferred Tensor PCA} 
\label{alg:transfer} 
\vspace*{0.08in} {\bf Input:}
Target tensor $\calT_{t}$, source tensor $\calT_{s}$, initial estimators $\big\{[\hU_{1,\ell}^{(0)} \; \hV^{(0)}_{1,\ell}],\dots,[\hU_{J,\ell}^{(0)} \; \hV^{(0)}_{J,\ell}]\big\}$, and weights $w_{\ell}$ for $\ell=s, t$ that satisfy $w_s+w_t=1$;

{\bf Output:}  Estimators $\big\{\hU_1, \hV_{1, t}, \hU_2, \hV_{2, t}, \dots, \hU_J, \hV_{J, t}\big\}$.
\begin{algorithmic}[1]  
\For{$\ell=s, t$} \do \\
\For{$j=1,2,\dots,J$}
\State Compute a local estimator 
$\hU_{j, \ell}=\mathrm{svd}_{r_{j, U}}\big(\widetilde M_{j,\ell}\big)$, where $\widetilde M_{j,\ell}$ is defined in \eqref{eq:local_hetero}.
\State Send  $\hU_{j,\ell}$ to the target machine; 
\EndFor
\EndFor
\For{$j=1,2,\dots, J$}
\State On the target machine, compute  $\hU_j = \mathrm{svd}_{r_{j, U}}\Big[w_s\hU_{j, s}\hU^{\top}_{j, s}+w_t\hU_{j, t}\hU^{\top}_{j, t}\Big]$;
\State Compute 
$\hV_{j, t}=\mathrm{svd}_{r_{j, V, t}}\Big[\big(I_{p_j}-\hU_{j}\hU_{j}^{\top}\big)\widetilde M_{j,t}\Big]$;
\EndFor
\end{algorithmic} 
\end{algorithm} 

Formally, suppose the source tensor $\calT_{s}=\calT_{s}^*+\calZ_s$ and the target tensor $\calT_{t}=\calT_{t}^*+\calZ_t$, where 
\begin{equation*}
\begin{aligned}
\calT_{s}^* &= \calG_{s} \times_1 [U_1 \; V_{1,s}] \times_2 [U_2 \; V_{2, s}] \cdots \times_J [U_J \; V_{J, s}], ~ U_j \in \mathbb{O}^{p_j \times r_{j,U}}, V_{j,s}\in\mathbb{O}^{p_j\times r_{j,V,s}}, \calG_{s} \in \R^{r_{1,s} \times  r_{2,s} \cdots \times r_{J,s}},\\
\calT_{t}^* &= \calG_{t} \times_1 [U_1 \; V_{1,t}] \times_2 [U_2 \; V_{2, t}] \cdots \times_J [U_J \; V_{J, t}],~ U_j \in \mathbb{O}^{p_j \times r_{j,U}}, V_{j,t}\in\mathbb{O}^{p_j\times r_{j,V,t}}, \calG_{t} \in \R^{r_{1,t} \times  r_{2,t} \cdots \times r_{J,t}},
\end{aligned}
\end{equation*}
with $U_j^{\top}V_{j,\ell}=0$ and $ r_{j,\ell}=r_{j, U}+r_{j, V, \ell}$ for $j \in [J]$ and $\ell=s, t$. In other words, we assume the source and target tensors share a common top-$r_{j, U}$ singular subspace spanned by $U_j$, but either task can have different individual components $V_{j, \ell}$. The goal is to estimate the singular subspaces spanned by $U_j$ and $V_{j, t}$ of the target tensor $\calT_t$. Meanwhile, we assume the noise $\calZ_{\ell}$ has i.i.d. $\calN(0, \sigma_{\ell}^2)$ entries for $\ell=s, t$, where the noise levels $\sigma_{s}$ and $\sigma_{t}$ are allowed to be different. 
To achieve knowledge transfer between $\calT_s$ and $\calT_t$,  we 
propose Algorithm \ref{alg:transfer}, which is carefully designed for dealing with the heterogeneity in the transfer setting. 
Different from Algorithm \ref{alg:hetero} who treats all tensors equally, Algorithm \ref{alg:transfer} aggregates the local estimators $\hU_{j, s}$ and $\hU_{j, t}$ through a weighted average. 
The weights $w_s$ and $w_t$ are designed to optimally balance the contributions from the source and target tensors, respectively, accounting for the potential heterogeneity in their noise levels $\sigma_s^2$ and $\sigma_t^2$. 
The choice for the weights will be specified in the next section.

\vspace{-.6em}
\subsection{Theoretical Guarantee}
Analogous to the notations in the heterogeneous settings, for $j \in [J]$ and $\ell=s,t$, let $\Lambda_{j,\ell}$ be the
$r_{j,\ell} \times  r_{j,\ell} $ singular value matrix of $\calM_j(\calG_{\ell})$.  Define $\lambda_{\max}$, $\lambda_{\min}$ to be the maximum and minimum singular value over all $\Lambda_{j,\ell}$, respectively, and let $\kappa_0 = \lambda_{\max} \lambda^{-1}_{\min}$.  Moreover, define $\Delta=\min_{j \in [J], \ell\in \{s, t\}}\{\lambda_{r_{j,U},j,\ell}-\lambda_{r_{j,U}+1,j,\ell}\}$ and $\kappa=\lambda_{\max}/\Delta$, where $\lambda_{r,j,\ell}$ denotes the $r$-th largest singular value of $\Lambda_{j,\ell}$. Additionally, let $r=\max_{j,\ell}  r_{j, \ell}$ and $r_V=\max_{j, \ell}r_{j, V, \ell}$.
\begin{thm}\label{thm:transfer}
Assume that there exist constants $C_1, c_1, C_2$ such that, with probability at least $1-C_1e^{-c_1p}$, $\big\|\hU_{j,\ell}^{(0)}\hU_{j,\ell}^{(0)\top}+\hV_{j,\ell}^{(0)}\hV_{j,\ell}^{(0)\top}- U_jU_j^{\top}-V_{j, \ell}V_{j, \ell}^{\top}\big\|_{2}\leq C_2\sqrt{p}\sigma_{\ell}\lambda_{\min}^{-1}$ for all $j \in [J], \ell \in \{s, t\}$. Assume $p_j \asymp p$ for all $j$,  $r^{J-1}=O(p)$, $\kappa_0=O(1)$, $\kappa=O(1)$, and $\sqrt{pr}=o\big(\min(\Delta, \lambda_{\min})/\max(\sigma_t, \sigma_s)\big)$. 
There exists constants $\widetilde{C}_1, \widetilde{c}_1, \widetilde{C}_2$ such that
\begin{align}\label{eq:transfer-U}
\sup_j\norm{\hU_{j}\hU_{j}^{\top}-U_jU_j^{\top}}_{\rm F}&\leq\widetilde{C}_2\frac{\sqrt{pr}\sqrt{w_{s}^2\sigma_{s}^2+w_{t}^2\sigma_{t}^2}}{\Delta},
\\
\label{eq:transfer-V}
\sup_{j}\norm{\hV_{j,t} \hV_{j,t}^{\top}-V_{j,t} V_{j,t}^{\top}}_{\rm F} 
&\leq\widetilde{C}_2\left(\frac{\sqrt{pr_{V}}}{\lambda_{\min}}\Big(\sigma_t+\sqrt{r}\sqrt{w_s^2\sigma^2_s+w_t^2\sigma_t^2}\Big)\right),
\end{align}
with probability at least $1-\widetilde{C}_1e^{-\widetilde c_1p}$, where $\hU_j, \hV_{j, t}$'s are the outputs of Algorithm \ref{alg:transfer}. 


\end{thm}

Theorem \ref{thm:transfer} establishes the statistical error rate for the estimators obtained by Algorithm \ref{alg:transfer}. We note that the best rate that an estimator for the common component $U_j$ can attain without transfer is $O(\sqrt{pr}\sigma_t\Delta^{-1})$, compared to which our transfer learning approach improves $\sigma_t$ into $\sqrt{w_{s}^2\sigma_{s}^2+w_{t}^2\sigma_{t}^2}$ if $\sigma_{s} \lesssim \sigma_t$. At the same time, the individual component estimator $\hV_{j,t}$ matches the local rate $O(\sqrt{pr}\sigma_t\lambda_{\min}^{-1})$ under a mild condition that $\sqrt{r_V}\sqrt{w_s^2\sigma^2_s+w_t^2\sigma_t^2} \lesssim \sigma_t$.

Based on the error rates established in \eqref{eq:transfer-U} and \eqref{eq:transfer-V}, we further give the optimal choice for $w_s$ and $w_t$. Specifically, by minimizing $w_s^2\sigma^2_s+w_t^2\sigma_t^2$ under the constraint $w_s+w_t=1$, we obtain that the optimal weights are $w_s^* = \frac{\sigma_t^2}{\sigma_s^2+\sigma_t^2}$ and $w_t^* = \frac{\sigma_s^2}{\sigma_s^2+\sigma_t^2}$, leading to the error rates
\begin{equation*}\label{eq:transfer-rate}
\sup_j\norm{\hU_{j}\hU_{j}^{\top}-U_jU_j^{\top}}_{\rm F}=O_{\Prob}\left(\frac{\sqrt{pr}\overline{\sigma}}{\Delta}\right) \, \mathrm{and} \, 	\sup_{j}\norm{\hV_{j,t} \hV_{j,t}^{\top}-V_{j,t} V_{j,t}^{\top}}_{\rm F} 
=O_{\Prob}\left(\frac{\sqrt{pr_{V}}\sigma_t}{\lambda_{\min}}\Big(1+\frac{\sqrt{r}\overline{\sigma}}{\sigma_t}\Big)\right),
\end{equation*} 
where $\overline{\sigma}^2=1/(\sigma_s^{-2}+\sigma_t^{-2})$. 
In practice, we can estimate $\sigma_{\ell}$ for $\ell=s, t$ by 
\begin{equation}\label{eq:estimate-sigma}
\widehat\sigma_{\ell} = \Big\|\calT_{\ell}-\calT_{\ell} \times_1 \big[\hU_{1,\ell}\; \hV_{1,\ell}\big]\big[\hU_{1,\ell}\; \hV_{1,\ell}\big]^{\top} \times_2\cdots\times_J\big[\hU_{J,\ell}\; \hV_{J,\ell}\big]\big[\hU_{J,\ell}\; \hV_{J,\ell}\big]^{\top}\Big\|_{\rm F} / \sqrt{p_1p_2\cdots p_J}, 
\end{equation}
and estimate the optimal weights by $\widehat{w}_s = \frac{\widehat\sigma_t^2}{\widehat\sigma_s^2+\widehat\sigma_t^2}$ and $\widehat w_t = \frac{\widehat\sigma_s^2}{\widehat\sigma_s^2+\widehat\sigma_t^2}$.

In Theorem C.1 of the supplementary material, we generalize our knowledge transfer methodology and theoretical guarantee from the two-machine setting into an arbitrary number of machines $L$. Furthermore, we establish the statistical error rate for the estimators using data-adaptive weights $\widehat w_{\ell}$ in Theorem C.2.

\vspace{-.8em}
\section{Numerical Study}\label{sec:numerical}

In this section, we conduct numerical studies to verify the theoretical properties and evaluate the empirical performance of our proposed distributed tensor PCA algorithms. Section \ref{sec:simulations} presents simulation studies for both the homogeneous and heterogeneous settings. Section \ref{sec:real} illustrates the performance of our algorithms on two real datasets of molecule structure.

\vspace{-.6em}
\subsection{Simulations}\label{sec:simulations}

In this section, we use Monte Carlo simulations to verify the performance of our proposed distributed Tensor PCA methods under various settings. Throughout the simulation, we consider 3-mode tensors (i.e., $J=3$). All results are averaged over 1,000 independent runs.

\paragraph{Estimation for Homogeneous Tensors} We first simulate the distributed homogeneous tensor setting described in Section \ref{sec:homo-setup}, 
where $
\calT^* = \calG \times_1 U_1 \times_2 U_2  \times_3 U_3$, $U_j \in \mathbb{O}^{p_j \times r_j}$, $\calG \in \R^{r_1 \times r_2  \times r_3}.$
We set the dimensions $p_1=p_2=p_3=p$ and the ranks $r_1=r_2=r_3=r$. The core tensor $\calG$ is generated by first sampling a tensor $\widetilde{\calG} \in \R^{r_1 \times r_2 \times r_3} $ with i.i.d. $ \calN(0, 1)$ entries and rescaling it as $\calG = \lambda \cdot \widetilde{\calG} / \lambda_{\min}\big(\widetilde{\calG}\big)$, where $\lambda_{\min}\big(\widetilde\calG\big)$ denotes the minimum singular value of all matricizations $\calM_j\big(\widetilde\calG\big)$. The generation procedure of $\calG$ ensures that the minimum singular value of $\calG$ is $\lambda$, which is denoted as the signal strength $\lambda_{\min}$ defined in Section \ref{sec:theory-homo}. For $j=1,2,3$, the matrix $U_j$ is generated via QR decomposition on a matrix $\widetilde{U}_j \in \R^{p_j \times r_j}$ with i.i.d. $\calN(0, 1)$ entries. Then we independently generate $\calZ_{\ell}$ with i.i.d. $\calN(0, \sigma^2)$ entries and obtain $\calT_{\ell}=\calT^*+\calZ_{\ell}$ for $\ell=1,2,\dots, L$. 

In particular,  we fix $r=3$, $\sigma=1$, and let $p \in \left\{50, 100\right\}$, $L \in \{10, 20\}$, and  $\lambda = p^{\gamma}$ with $\gamma \in [0.45, 0.95]$.  We report the estimation error of our proposed Algorithm \ref{alg:homo} (referred to as ``distributed'') for $U_1$, i.e., $\rho\big(\hU_1, U_1\big)$, with the SNR $\lambda/\sigma$ ranging from $p^{0.45}$ to $p^{0.95}$.  For comparison, we also report the estimation error of the pooled estimator $\hU_{\mathrm{pooled}, 1}$  (referred to as ``pooled'') defined in Section \ref{sec:homo-setup} under the same settings. The estimation errors for $p=50$, $100$ and $L=10$, $20$ are displayed in Figure \ref{fig:homo}.

Figure \ref{fig:homo} shows that the estimation errors of both our proposed estimator $\hU_1$ and the pooled estimator $\hU_{\mathrm{pooled}, 1}$ decrease as the SNR $\lambda/\sigma$ increases, and moreover, when the SNR is sufficiently high (e.g., $\lambda/\sigma \geq p^{0.7}$ for $p=100$, $L=20$), the performance of $\hU_1$ becomes indistinguishable from that of the pooled estimator, verifying that our distributed algorithm achieves the optimal minimax rate as stated in Corollary \ref{cor:1}. Furthermore, we observe that increasing $L$ from 10 to 20 leads to a noticeable reduction in the estimation error of both $\hU_1$ and the pooled estimator, consistent with the $L^{-1/2}$ rate in Theorem \ref{thm:1} and Corollary \ref{cor:1}.

\begin{figure}[!t]
\centering
\begin{subfigure}[b]{0.45\textwidth}
\centering
\includegraphics[width=\textwidth]{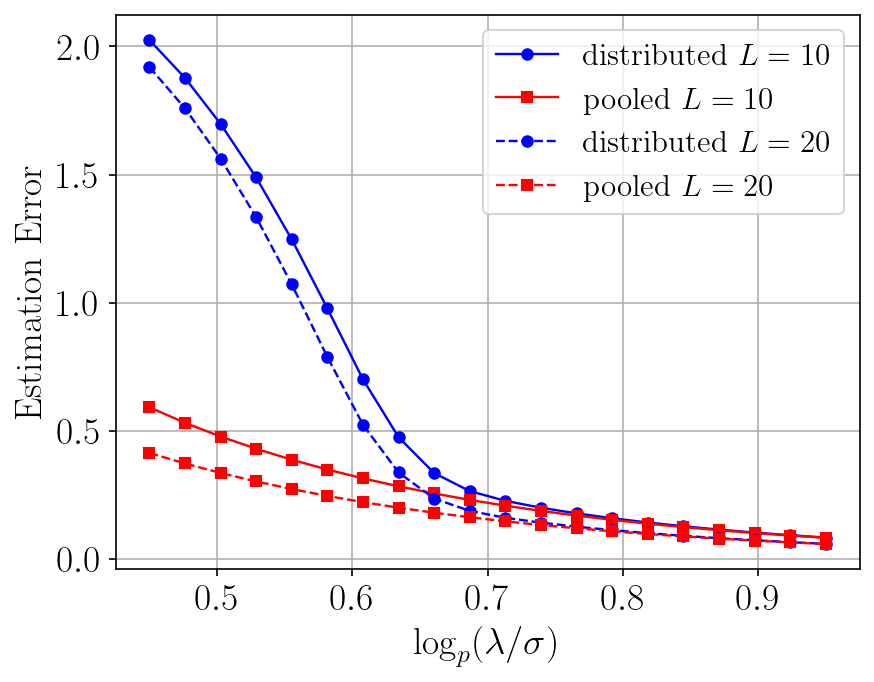}
\caption{$p=50$}
\label{fig:homo_p=1}
\end{subfigure}
\hfill
\begin{subfigure}[b]{0.45\textwidth}
\centering
\includegraphics[width=\textwidth]{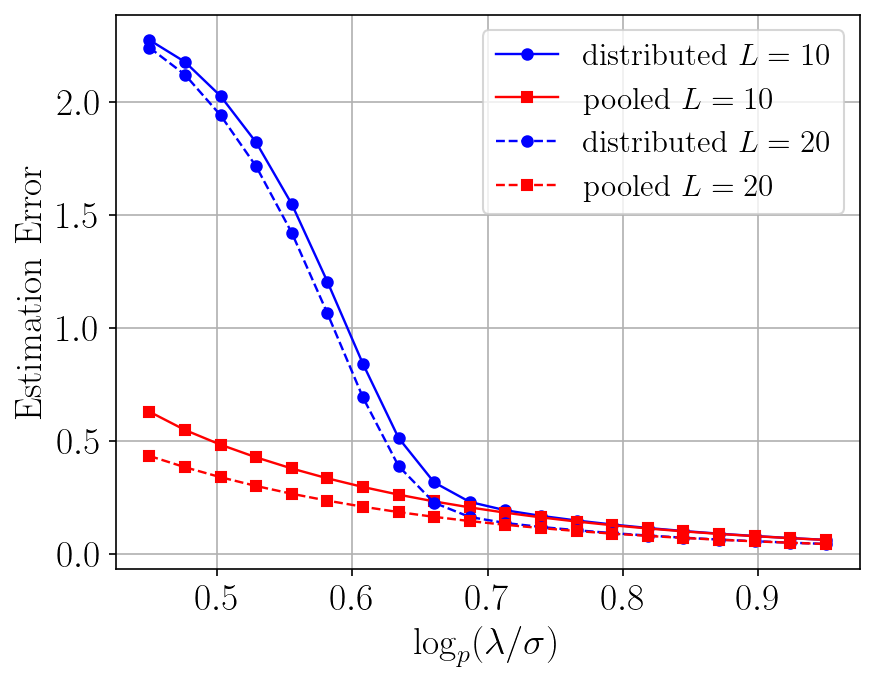}
\caption{$p=100$}
\label{fig:homo_p=10}
\end{subfigure}
\caption{The estimation errors of different methods under the homogeneous setting. }
\label{fig:homo}
\end{figure}

\begin{figure}[!t]
\centering
\begin{subfigure}[b]{0.45\textwidth}
\centering
\includegraphics[width=\textwidth]{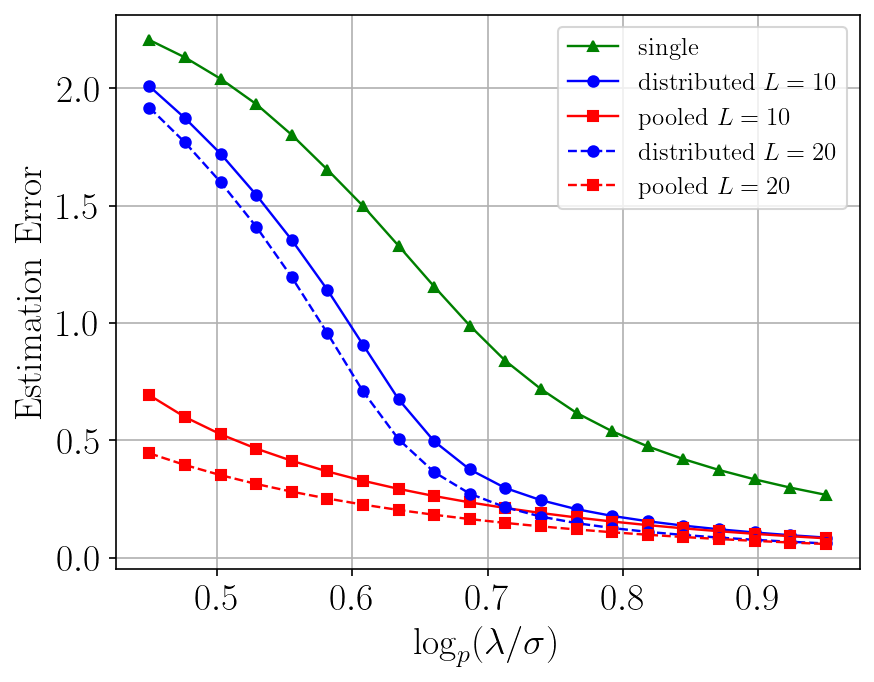}
\caption{$p=50$}
\label{fig:hetero_p=1}
\end{subfigure}
\hfill
\begin{subfigure}[b]{0.45\textwidth}
\centering
\includegraphics[width=\textwidth]{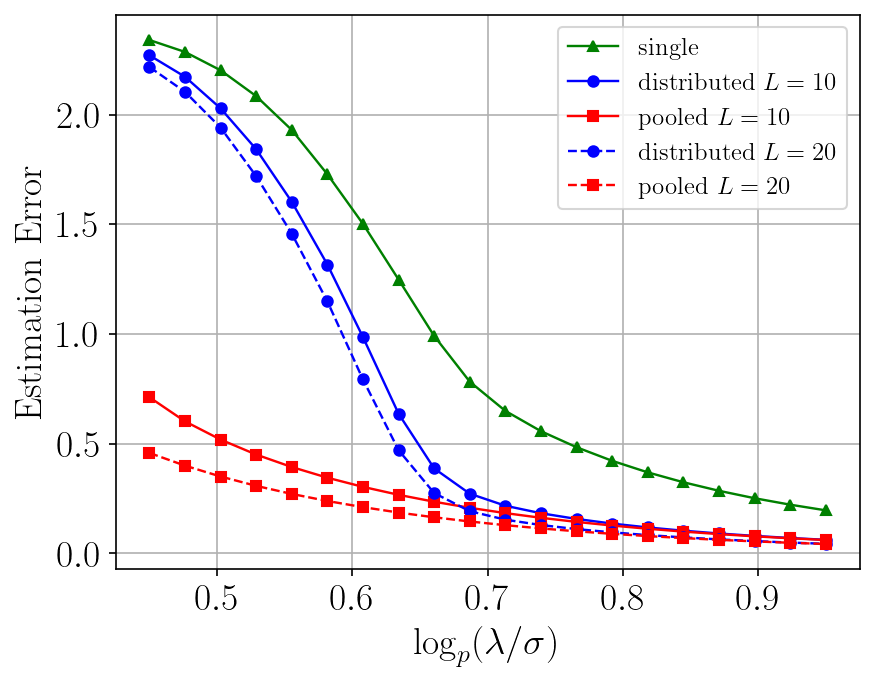}
\caption{$p=100$}
\label{fig:hetero_p=10}
\end{subfigure}
\caption{Estimation errors under heterogeneous settings where tensors share the same core. }
\label{fig:hetero-same-core}
\end{figure}

\paragraph{Estimation for Heterogeneous Tensors} We then conduct simulations for the distributed heterogeneous setting described in Section \ref{sec:hetero}, where 
\[
\calT_{\ell}^* = \calG_{\ell} \times_1 [U_1 \; V_{1,\ell}] \times_2 [U_2 \; V_{2, \ell}]  \times_3 [U_3 \; V_{3, \ell}], \quad U_j \in \R^{p_j \times r_{j,U}}, V_{j,\ell}\in\R^{p_j\times r_{j,V,\ell}}, \calG_{\ell} \in \R^{ r_{1,\ell} \times  r_{2,\ell}  \times  r_{3,\ell}}.
\]
We set the dimensions $p_j=p \in \{50, 100\}$ and the ranks $r_{j,U}=3$, $r_{j, V, \ell}=3$, and $r_{j,\ell}=r_{j, U}+r_{j, V, \ell}=6$,  for all $j=[3]$, $\ell \in [L]$. Similar to the homogeneous setting, the shared component $U_j$ is generated via QR decomposition on a matrix $\widetilde{U}_j \in \R^{p_j \times r_{j,U}}$ with i.i.d. $\calN(0, 1)$ entries. The individual component $V_{j,\ell}$ is generated via QR decomposition on $(I_p - U_jU_j^{\top})\widetilde{V}_{j,\ell}$, where $\widetilde{V}_{j,\ell} \in \R^{p_j \times r_{j,V,\ell}}$ has i.i.d. $\calN(0, 1)$ entries. The projection matrix $I_p - U_jU_j^{\top}$ ensures that $U_j^{\top}V_{j,\ell}=0$.  Moreover, the core tensors $\{\calG_{\ell}\}$ are generated as follows.

Given a pre-specified $\lambda \in \R_+$, we independently sample two tensors $\widetilde{\calG}_U \in \R^{3 \times 3 \times 3}$ and $\widetilde{\calG}_{V} \in \R^{3 \times 3 \times 3}$, both with i.i.d. $ \calN(0, 1)$ entries. Then we rescale them as $\calG_U = \lambda \cdot \widetilde{\calG}_U / \lambda_{\min}\big(\widetilde{\calG}_U\big)$ and $\calG_V = (\lambda/2) \cdot \widetilde{\calG}_V / \lambda_{\max}\big(\widetilde{\calG}_V\big)$, where $\lambda_{\min}(\lambda_{\max})(\mathcal{X})$ denotes the minimum (maximum) singular value over all matricizations $\calM_j(\mathcal{X})$. Finally, we generate $\calG$ as a ``block-diagonal'' tensor such that
\begin{equation}
\calG_{i_1, i_2, i_3} = \begin{cases}
(\calG_U)_{i_1, i_2, i_3} & \text{if } 1\leq i_1, i_2, i_3 \leq 3,\\
(\calG_V)_{i_1, i_2, i_3} & \text{if } 4\leq i_1, i_2, i_3 \leq 6,\\
0 & \text{otherwise}.
\end{cases}
\end{equation} 
In other words, only the top-left $3\times 3 \times 3$ block and the bottom-right $3 \times 3 \times 3$ block of $\calG$ are non-zero.  Furthermore, we consider two cases to generate the local core tensors $\{\calG_{\ell}\}$: 
\begin{itemize}
\item same core: After generating $\calG$, let $\calG_{\ell}=\calG$ for all $\ell \in [L]$;
\item different cores: For each $\ell$, independently generate $\calG_{\ell}$ using the same procedure of generating $\calG$ as described above.
\end{itemize}

The core tensors are constructed such that the minimum singular value gap $\Delta$ between the common and individual components equals $\lambda/2$, representing the signal strength for estimating $U_j$. 
We report the estimation errors of our proposed estimator $\hU_1$ in Algorithm \ref{alg:hetero}, along with the errors of the local estimator $\hU_{1, 1}$ in Algorithm \ref{alg:hetero} (referred to as ``single'')  and the pooled estimator $\hU_{\mathrm{pooled}, 1}$, for $p=50$, $100$ and $L=10$, $20$. 

The results are displayed in Figures \ref{fig:hetero-same-core} and \ref{fig:hetero-diff-core}. In Figure \ref{fig:hetero-same-core}, where all tensors share the same core tensor, our distributed estimator $\hU_1$ achieves a similar error rate as the pooled estimator when the SNR is sufficiently high, which verifies the theoretical results established in Theorem \ref{thm:hetero}. Meanwhile, the local estimator (``single'') exhibits a much higher error, highlighting the advantage of combining information across multiple tensors. In Figure \ref{fig:hetero-diff-core}, where the core tensors are different across machines, the performance of the pooled estimator deteriorates significantly, as the simple averaging of the tensors is invalidated due to the different cores. In contrast, our distributed estimator still achieves a decent error rate, outperforming both the pooled and local estimators. This further demonstrates the effectiveness of our method in learning the shared component in the presence of heterogeneity.

\begin{figure}[!t]
\centering
\begin{subfigure}[b]{0.45\textwidth}
\centering
\includegraphics[width=\textwidth]{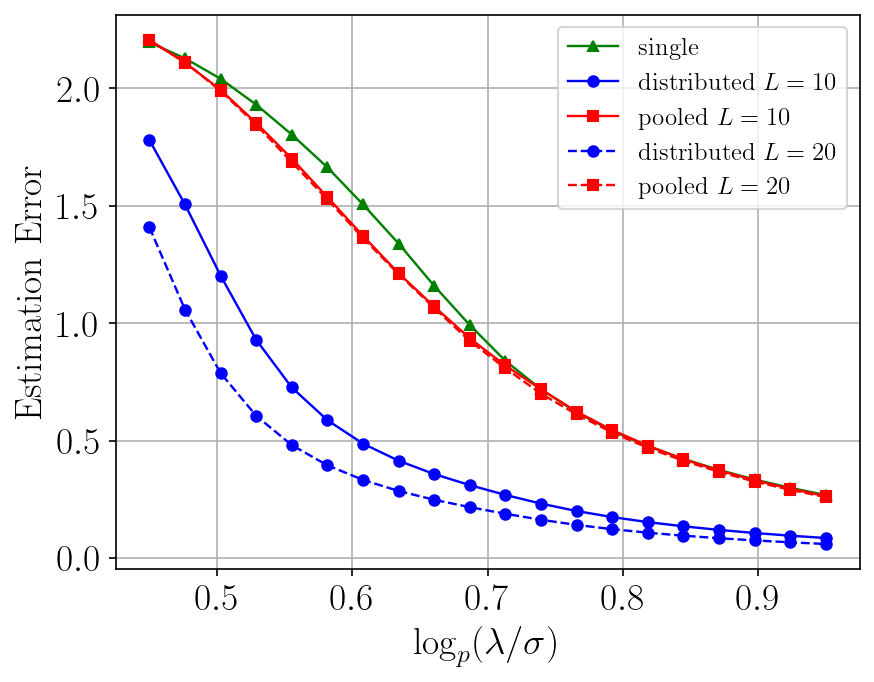}
\caption{$p=50$}
\label{fig:hetero_2_p=1}
\end{subfigure}
\hfill
\begin{subfigure}[b]{0.45\textwidth}
\centering
\includegraphics[width=\textwidth]{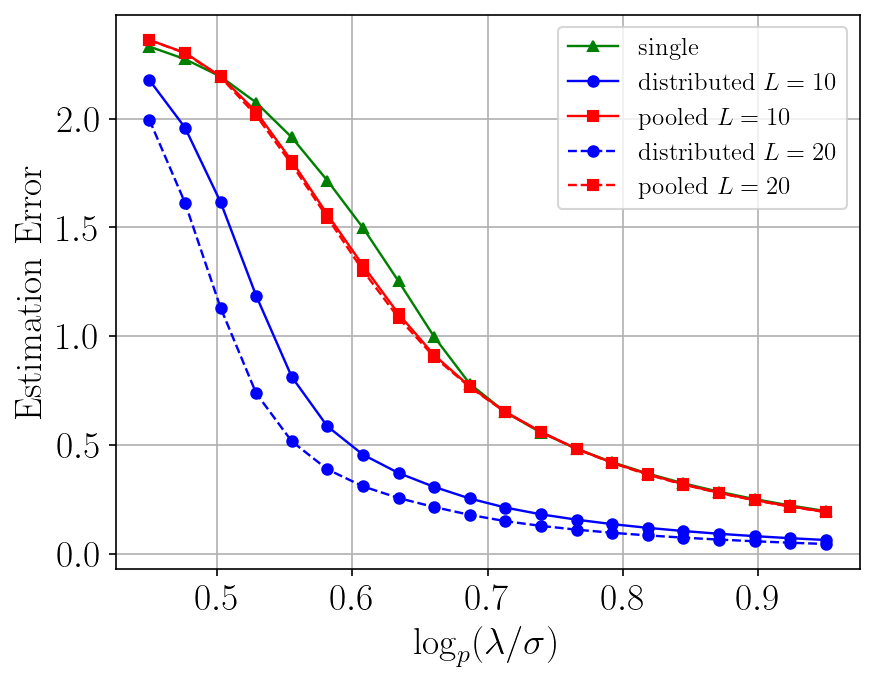}
\caption{$p=100$}
\label{fig:hetero_2_p=10}
\end{subfigure}

\caption{Estimation errors under heterogeneous settings where tensors have different cores. }
\label{fig:hetero-diff-core}
\end{figure}

\vspace{-.6em}
\subsection{Real Data Analysis}\label{sec:real}

\begin{figure}[!t]
\centering
\begin{subfigure}[b]{0.45\textwidth}
\centering
\includegraphics[width=\textwidth]{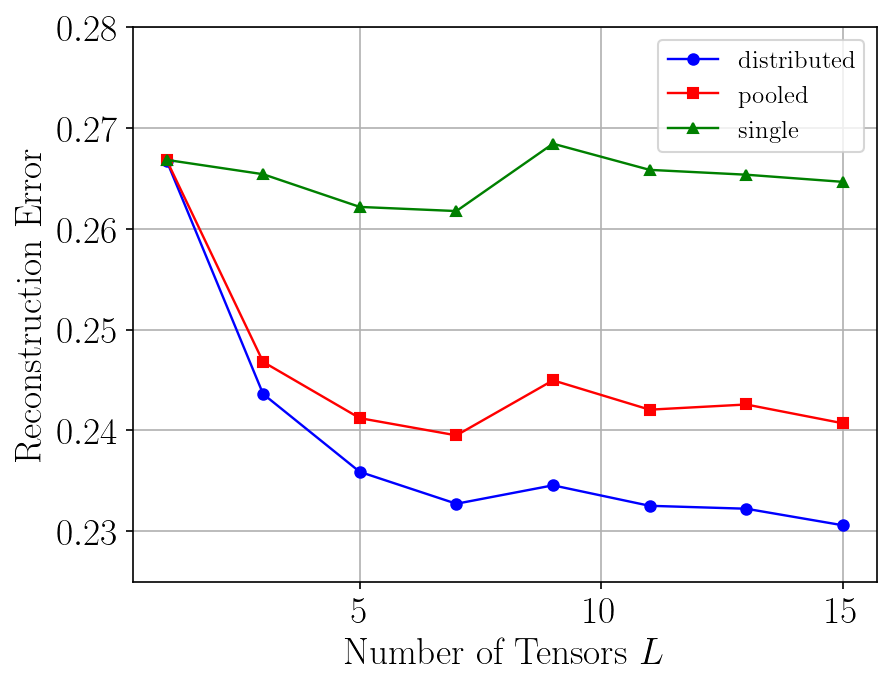}
\caption{PROTEINS}
\label{fig:real_small_r}
\end{subfigure}
\hfill
\begin{subfigure}[b]{0.45\textwidth}
\centering
\includegraphics[width=\textwidth]{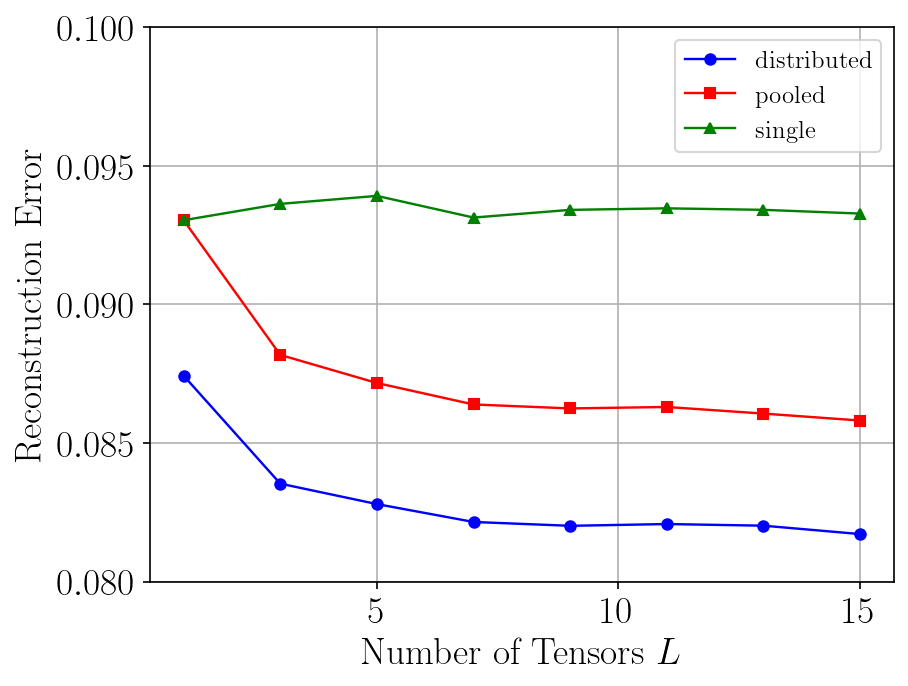}
\caption{PTC\_FM}
\label{fig:real_large_r}
\end{subfigure}
\caption{Comparison of the reconstruction errors within class 0 of the two datasets.}
\label{fig:real}
\end{figure}

{In this section, we illustrate the usefulness of our proposed methods on two real-world molecule datasets, PROTEINS \citep{borgwardt2005protein, morris2020tudataset} and PRC\_FM \citep{helma2001predictive, kriege2012subgraph}.  PROTEINS 
consists of 1,113 protein graphs labeled according to whether they are enzymes. 
Each graph represents the structure of a single protein, where the vertices represent the secondary structure elements (i.e., the helices, sheets, and turns) of the protein, and the edges connect nodes that are neighbors along the amino acid sequence or neighbors in space within the protein structure. PRC\_FM contains 349 compound graphs classified according to the carcinogenicity on female mice, where the vertices represent atoms, and the edges represent chemical bonds.  Following the procedure in \cite{wen2024tensor},  Topological Data Analysis (TDA) is employed to encode the topological and structural features of each graph into a three-mode tensor of dimensions $2\times 50 \times 50$ for PROTEINS and $5\times 50 \times 50$ for PTC\_FM, composed of $50 \times 50$ persistence images constructed by several filtration functions. }

Since obtaining the ground truth $U_j$ for real data is difficult, we evaluate the performance of different methods by the \textit{reconstruction error} of the estimators. Formally, given a tensor $\calT$, the reconstruction error of estimators $\big\{\hU_j\big\}_{j=1,2,3}$ on $\calT$ is defined as
\begin{equation*}
	\mathrm{RE}\big(\hU_1, \hU_2, \hU_3;\calT\big) = \frac{\norm{\calT-\calT \times_1\hU_1\hU_1^{\top}\times_2\hU_2\hU_2^{\top}\times_3\hU_3\hU_3^{\top}}_{\rm F}}{\norm{\calT}_{\rm F}},
\end{equation*}
which measures the difference between the original tensor and the tensor reconstructed from the estimated principal components, normalized by the Frobenius norm of the original tensor. 

Concretely, we randomly select $L$ tensors as the training samples, denoted by $\{\calT_{\ell}\}_{\ell \in \mathcal{I}}$, where $\mathcal{I}$ is an index set with cardinality $L$. The training samples $\{\calT_{\ell}\}_{\ell \in \mathcal{I}}$ are input into Algorithm \ref{alg:hetero} to obtain estimators $\big\{\hU_j\big\}_{j=1,2,3}$. Next, we randomly select $L'$ tensors other than the training samples as the test set, denoted by $\{\calT_{\ell}\}_{\ell \in \mathcal{I}'}$ ($\mathcal{I} \cap \mathcal{I}'=\emptyset$), and then compute the averaged reconstruction error of $\big\{\hU_j\big\}$ over the test samples, i.e., $\frac{1}{L'}\sum\limits_{\ell \in \mathcal{I}'}\mathrm{RE}\big(\hU_1, \hU_2, \hU_3;\calT_{\ell}\big)$. For comparison, we also record the reconstruction errors of two other methods: 
\begin{enumerate}
	\item[(1)] ``single'': the local estimators $\big\{\hU_{j, \ell}\big\}_{j=1,2,3}$ obtained using each training sample $\calT_{\ell}$, $\ell \in \mathcal{I}$;
	\item[(2)] ``pooled'':  the pooled estimators obtained by decomposing the averaged tensor $\frac{1}{L}\sum\limits_{\ell \in \mathcal{I}}\calT_{\ell}$. 
\end{enumerate}
Specifically, we fix $L'=100$ and let $L$ range from 1 to 100. 

{Figure \ref{fig:real} presents the reconstruction errors of the three methods for the two datasets, which are averaged over 200 independent repeats of random sample selection. For the ``single'' method, we report the average reconstruction error over all training samples $\calT_{\ell}$, that is, $\frac{1}{LL'}\sum\limits_{\ell \in \mathcal{I}}\sum\limits_{\ell' \in \mathcal{I}'}\mathrm{RE}\big(\hU_{1,\ell}, \hU_{2,\ell}, \hU_{3,\ell};\calT_{\ell'}\big)$. As shown in Figure \ref{fig:real}, our distributed method significantly reduces the reconstruction error compared to the ``single'' and ``pooled'' estimators, highlighting the advantage of our approach in correctly integrating different tensors in the presence of heterogeneity. 
This superior performance is attributed to the fundamental differences in how each method handles tensor heterogeneity. The ``pooled'' estimator performs direct averaging of the raw tensors, which can be inappropriate given the heterogeneity of core tensors across different tensors, as evidenced by the simulation results in Figure \ref{fig:hetero-diff-core}. 
In contrast, our distributed estimator identifies and aggregates the common principal components shared across tensors, which leads to better performance than direct tensor averaging. Notably, this finding suggests that even in scenarios where tensor pooling is technically feasible, a distributed estimator may still be preferred as it more appropriately extracts and integrates the common information. Additional real data experiments on transfer learning tasks are explored in Section F of the supplementary material.}

\vspace{-.8em}
\section{Conclusion}
\label{sec:conc}

This paper presents innovative distributed tensor PCA methods for both homogeneous and heterogeneous settings, addressing the challenges of analyzing high-dimensional tensor data stored across multiple locations. We develop and theoretically validate algorithms that efficiently aggregate shared low-rank subspaces and identify unique components, enhancing estimation accuracy under various data heterogeneity conditions. Our methods demonstrate significant improvements over traditional approaches in extensive numerical analyses. Future research will focus on expanding these techniques to more complex distributed frameworks, exploring adaptive algorithms that can dynamically adjust to varying data characteristics and network conditions. 

\spacingset{1.5}

\section*{Acknowledgment}
The authors gratefully thank the editor, the associate editor, and the anonymous referees for their helpful comments and suggestions. 

\section*{Disclosure Statement}
The authors report there are no competing interests to declare.

\section*{Funding}
Elynn Chen's research was supported in part by the National Science Foundation under Award ID 2412577. Xi Chen would like to thank the support from National Science Foundation via the grant IIS-1845444.

\bibliographystyle{chicago}
\bibliography{dist_tensor_pca}

\newpage
\spacingset{1.9}
\appendix

\begin{center}
{\large\bf SUPPLEMENTARY MATERIAL of \\
``\TITLE''}
\end{center}

\section{Notations} \label{append:notation}

For any integer $N$, let $[N]$ denote the set $\{1,2,\dots, N\}$.  For any positive sequences $\{a_n\}$ and $\{b_n\}$, we write $a_n \lesssim b_n$ if $a_n=O(b_n)$, and $a_n \asymp b_n$ if $a_n \lesssim b_n$ and $b_n \lesssim a_n$.
For a sequence of random variables $X_n$ and a sequence of real numbers $a_n$, we let $X_n=O_{\Prob}(a_n)$ denote that $\left\{X_n/a_n\right\}$ is bounded in probability and $X_n=o_{\Prob}(a_n)$ denote that $\left\{X_n/a_n\right\}$ converges to zero in probability. 

Define $\mathbb{O}^{p \times r}=\big\{U \in \mathbb{R}^{p \times r} \mid U^{\top}U=I_r\big\}$, where $I_r$ is the  $r \times r$ identity matrix.  For any matrix $M$, denote the top-$r$ left singular vectors of $M$ as $\mathrm{svd}_{r}(M)$, and let $\col(M)$ denote the linear space spanned by the columns of $M$.  Let $\norm{M}_2$ and $\norm{M}_{\rm F}$ be the spectral norm and Frobenius norm of $M$, respectively.

For any two matrices $M_1 \in \R^{p_1 \times q_1}$ and $M_2 \in \R^{p_2 \times q_2}$, define $M_1 \otimes M_2 \in \R^{p_1p_2 \times q_1q_2}$ to be their Kronecker product. For any tensor $\mathcal{X} \in \R^{p_1 \times p_2 \dots\times p_J}$, matrix $M \in \R^{q_j \times p_j}$, and $j \in [J]$, let $\mathcal X \times_j M$ denote the mode-$j$ matrix product of $\mathcal X$. That is, $(\mathcal X \times_j M) \in \R^{p_1 \times \dots \times q_j \times\dots \times p_J}$, and 
\[(\mathcal X \times_j M)_{i_1,i_2,\dots,k,\dots,i_J} = \sum_{i_j=1}^{p_{j}} \mathcal X_{i_1,i_2,\dots,i_j,\dots,i_J}M_{k,i_j}.\]
Moreover, let $\calM_j(\mathcal X) \in \R^{p_j \times (p_1p_2\cdots p_J/p_j)}$ denote the mode-$j$ matricization (unfolding) of $\mathcal X$, which reorders the mode-$j$ fibers of the tensor $\mathcal{X}$ to be the columns of the matrix $\calM_j(\mathcal X)$. Let $\lambda_{\min}(\mathcal X)$ denote the minimum singular value over all the matricizations $\calM_j(\mathcal{X})$, $j \in [J]$. Additionally, define $\norm{\mathcal X}_{\rm F} = \norm{\calM_1(\mathcal X)}_{\rm F}$.

\section{Discussions on the Second-Order Term}\label{sec:supp-lower}

In this section, we provide an example to show that the second-order term $pr\sigma^2\lambda_{\min}^{-2}$ in Theorem 2.1 is not improvable, which implies the necessity of the condition $\sqrt{prL}=o(\lambda_{\min}/\sigma)$ in Corollary 2.1 to achieve the optimal rate $\sqrt{pr/L}\sigma\lambda_{\min}^{-1}$.  For simplicity, we consider the case $J=3$, $p_1=p_2=p_3=p$,  and $r_1=r_2=r_3=1$. We first define $\mathcal{F}_{p, 1, \lambda_{\min}}:=\left\{\calT: \calT=\lambda_{\min}\times_1 U_1 \times_2 U_2 \times_3 U_3, U_{j} \in \mathbb{O}^{p \times 1}, j=1, 2, 3\right\}$ to be the class of target tensors $\calT^*$ satisfying model (1). Moreover, define the estimator class $\mathfrak{U}_{C_1, c_1, C_2}:=\Big\{\big(\hU_1, \hU_2, \hU_3\big): \Prob\big[\big\|\hU_j\hU_j^{\top}-U_jU_j^{\top}\big\|_2 \leq C_2\sqrt{p}\sigma\lambda_{\min}^{-1}\big] \geq 1- C_1e^{-c_1p}, \hU_j\in \mathbb{O}^{p\times 1}, j=1,2,3\Big\}$, which contains all the possible initial estimators satisfying the assumption in Theorem 2.1 for constants $C_1$, $c_1$, and $C_2$. Then we have the following result.
\begin{thm}\label{thm:second-lower}
	As assumed in Theorem 2.1, suppose that $\sqrt{p}\sigma/\lambda_{\min}=o(1)$ and $L =O (p^{c_3})$ for some $c_3>0$. There exist constants $C_1$, $c_1$, and $C_2>0$, such that, if $\lambda_{\min}/\sigma=o(\sqrt{pL})$, then
	\begin{equation}
		\sup_{\substack{\calT^* \in \mathcal{F}_{p, 1, \lambda_{\min}}\\ \big(\hU_{1, \ell}^{(0)}, \hU_{2, \ell}^{(0)}, \hU_{3, \ell}^{(0)}\big)\in\mathfrak{U}_{C_1, c_1, C_2} \\ \text{for all } \ell=1,2,\dots,L}} \norm{\hU_1\hU_1^{\top}-U_1U_1^{\top}}_{\rm F} \geq \widetilde{c}_2p\sigma^2\lambda_{\min}^{-2}, 
	\end{equation}
	with probability at least $1-\widetilde{C}_1 e^{-\widetilde{c}_1p}-\widetilde{C}_1p^{-L}$ for some constants $\widetilde{c}_1, \widetilde{C}_1, \widetilde{c}_2>0$, where $\hU_1$ is the estimator for $U_1$ obtained from Algorithm 1.
\end{thm}

The proof of Theorem \ref{thm:second-lower} is provided in Section \ref{sec:supp_proof_homo}. Recall that Theorem 2.1 establishes that the estimation error $\big\|\hU_1\hU_1^{\top}-U_1U_1^{\top}\big\|_{\rm F}$ is bounded above by  $O\big( \sigma\lambda_{\min}^{-1}(pr)^{1/2}L^{-1/2}+ pr\sigma^2\lambda_{\min}^{-2}\big)$.  Under the condition $\sqrt{prL} \lesssim \lambda_{\min} / \sigma$, the first term $\sigma\lambda_{\min}^{-1}(pr)^{1/2}L^{-1/2}$ dominates the second term $pr\sigma^2\lambda_{\min}^{-2}$, leading to the minimax optimal error rate in Corollary 2.1. Theorem \ref{thm:second-lower} further provides a lower bound result: if the condition  $\sqrt{prL} \lesssim \lambda_{\min} / \sigma$ is violated, there exists a tensor $\calT^*$ and a set of initial estimators that satisfy all the assumptions in Theorem 2.1 but result in a sub-optimal rate $pr\sigma^2\lambda_{\min}^{-2}$. This result demonstrates that the error bound in Theorem 2.1 is sharp and thus the condition $\sqrt{prL} \lesssim \lambda_{\min} / \sigma$ is necessary for achieving the optimal convergence rate.

\section{Further Results for Knowledge Transfer}\label{sec:supp-transfer}

In this section, we generalize our knowledge transfer methodology and theoretical analysis in Section 4 from the two-machine setting into an arbitrary number of machines $L$. Furthermore, we present results not only for fixed weights  but also for data-adaptive weights. 

Suppose there are $L$ tensors $\calT_{\ell}=\calT_{\ell}^*+\calZ_{\ell}$, where $\ell=1$ denotes the target tensor and $\ell=2, \dots, L$ represent the $L-1$ source tensors. Analogous to Section 4, we assume that
\begin{equation*}
	\begin{aligned}
		\calT_{\ell}^* &= \calG_{\ell} \times_1 [U_1 \; V_{1,\ell}] \times_2 [U_2 \; V_{2, \ell}] \cdots \times_J [U_J \; V_{J, \ell}], \quad U_j \in \mathbb{O}^{p_j \times r_{j,U}}, V_{j,\ell}\in\mathbb{O}^{p_j\times r_{j,V,\ell}}, \calG_{\ell} \in \R^{r_{1,\ell} \times  r_{2,\ell} \cdots \times r_{J,\ell}},
	\end{aligned}
\end{equation*}
with $U_j^{\top}V_{j,\ell}=0$ and $ r_{j,\ell}=r_{j, U}+r_{j, V, \ell}$ for $j \in [J]$ and $\ell\in[L]$. Our objective is to estimate the singular subspaces spanned by $U_j$ and $V_{j, 1}$ of the target tensor $\calT_1$. The noise tensors $\calZ_{\ell}$ are assumed to have i.i.d. $\calN(0, \sigma_{\ell}^2)$ entries for $\ell\in [L]$, thus accommodating different noise levels across the tensors. 

\begin{algorithm}[!t] 
	\spacingset{1.4}
	\caption{Transferred Tensor PCA for $L$ Machines} 
	\label{alg:transfer-L} 
	\vspace*{0.08in} {\bf Input:}
	Target tensor $\calT_{1}$, source tensors $\calT_{2},\dots, \calT_{L}$, initial estimators $\big\{[\hU_{1,\ell}^{(0)} \; \hV^{(0)}_{1,\ell}],\dots,[\hU_{J,\ell}^{(0)} \; \hV^{(0)}_{J,\ell}]\big\}$, and weights $w_{\ell}$ for $\ell=1,2,\dots, L$ that satisfy $\sum_{\ell}w_\ell=1$;
	
	{\bf Output:}  Estimators $\big\{\hU_1, \hV_{1, 1}, \hU_2, \hV_{2, 1}, \dots, \hU_J, \hV_{J, 1}\big\}$.
	\begin{algorithmic}[1]  
		\For{$\ell=1,2,\dots, L$} \do \\
		\For{$j=1,2,\dots,J$}
		\State Compute a local estimator 
		$\hU_{j, \ell}=\mathrm{svd}_{r_{j, U}}\big(\widetilde M_{j,\ell}\big)$, where $\widetilde M_{j,\ell}$ is defined in \eqref{eq:local_hetero}.
		\State Send  $\hU_{j,\ell}$ to the target machine; 
		\EndFor
		\EndFor
		\For{$j=1,2,\dots, J$}
		\State On the target machine, compute  $\hU_j = \mathrm{svd}_{r_{j, U}}\Big[\sum_{\ell=1}^{L}w_\ell\hU_{j, \ell}\hU^{\top}_{j, \ell}\Big]$;
		\State Compute 
		$\hV_{j, 1}=\mathrm{svd}_{r_{j, V,1}}\Big[\big(I_{p_j}-\hU_{j}\hU_{j}^{\top}\big)\widetilde M_{j,1}\Big]$;
		\EndFor
	\end{algorithmic} 
\end{algorithm} 


To transfer knowledge from $\calT_{2},\dots, \calT_{L}$ to improve the learning performance on $\calT_1$, we first generalize Algorithm 3 to Algorithm \ref{alg:transfer-L}  and then provide the theoretical guarantee for the estimators. Similar to Section 4, for $j \in [J]$ and $\ell\in[L]$, let $\Lambda_{j,\ell}$ be the
$r_{j,\ell} \times  r_{j,\ell} $ singular value matrix of $\calM_j(\calG_{\ell})$.  Define $\lambda_{\max}$, $\lambda_{\min}$ to be the maximum and minimum singular value over all $\Lambda_{j,\ell}$, respectively, and let $\kappa_0 = \lambda_{\max} \lambda^{-1}_{\min}$.  Moreover, define $\Delta=\min_{j \in [J], \ell\in [L]}\{\lambda_{r_{j,U},j,\ell}-\lambda_{r_{j,U}+1,j,\ell}\}$ and $\kappa=\lambda_{\max}/\Delta$, where $\lambda_{r,j,\ell}$ denotes the $r$-th largest singular value of $\Lambda_{j,\ell}$. Additionally, let $r=\max_{j,\ell}  r_{j, \ell}$ and $r_V=\max_{j, \ell}r_{j, V, \ell}$.
\begin{thm}\label{thm:transfer-L}
	Assume that there exist constants $C_1, c_1, C_2$ such that, with probability at least $1-C_1e^{-c_1p}$, $\big\|\hU_{j,\ell}^{(0)}\hU_{j,\ell}^{(0)\top}+\hV_{j,\ell}^{(0)}\hV_{j,\ell}^{(0)\top}- U_jU_j^{\top}-V_{j, \ell}V_{j, \ell}^{\top}\big\|_{2}\leq C_2\sqrt{p}\sigma_{\ell}\lambda_{\min}^{-1}$ for all $j \in [J], \ell \in  [L]$. Assume $p_j \asymp p$ for all $j$,  $r^{J-1}=O(p)$, $\kappa_0=O(1)$, $\kappa=O(1)$, and $\sqrt{prL}=o(\min(\Delta, \lambda_{\min})/\max_{\ell} \sigma_\ell)$. Then there exist constants $\widetilde{C}_1, \widetilde{c}_1, \widetilde{C}_2$ such that
	\begin{align}\label{eq:transfer-U-L}
		\sup_j\norm{\hU_{j}\hU_{j}^{\top}-U_jU_j^{\top}}_{\rm F}&\leq \widetilde{C}_2 \frac{\sqrt{pr}\sqrt{\sum_{\ell=1}^{L}w_{\ell}^2\sigma_{\ell}^2}}{\Delta},
		\\
		\label{eq:transfer-V-L}
		\sup_{j}\norm{\hV_{j,1} \hV_{j,1}^{\top}-V_{j,1} V_{j,1}^{\top}}_{\rm F} 
		&\leq \widetilde{C}_2\left(\frac{\sqrt{pr_{V}}}{\lambda_{\min}}\Big(\sigma_{1}+\sqrt{r}\sqrt{\sum_{\ell=1}^{L}w_{\ell}^2\sigma_{\ell}^2}\Big)\right),
	\end{align}
	with probability at least $1-\widetilde{C}_1e^{-\widetilde c_1p}$, where $\hU_j, \hV_{j, 1}$'s are the outputs of Algorithm \ref{alg:transfer-L}. 
\end{thm}

Theorem \ref{thm:transfer-L} is a generalization of Theorem 4.1 from $L=2$ to arbitrary $L$. Furthermore, it also broadens the results for the heterogeneous setting established in Section 3  to accommodate varying noise levels across different tensors $\calT_{\ell}$.  If $\sigma_{\ell}=\sigma$ for all $\ell$ and $w_{\ell}=\frac{1}{L}$, the results of Theorem \ref{thm:transfer-L} reduce to those of Theorem 3.1.  The proof of Theorem \ref{thm:transfer-L} is provided in Section \ref{sec:supp_proof_transfer}. 

Based on the error rates established in \eqref{eq:transfer-U-L} and \eqref{eq:transfer-V-L}, we further give the optimal choice for $w_{\ell}$ by minimizing $\sum_{\ell=1}^{L}w_{\ell}^2\sigma^2_{\ell}$ under the constraint $\sum_{\ell=1}^L w_{\ell}=1$, which yields the optimal weights $w_{\ell}^* = \frac{\sigma_{\ell}^{-2}}{\sum_{\ell=1}^{L}\sigma_{\ell}^{-2}}$. The corresponding optimal error rates are
\begin{equation}\label{eq:transfer-rate-L}
	\sup_j\norm{\hU_{j}\hU_{j}^{\top}-U_jU_j^{\top}}_{\rm F}=O_{\Prob}\left(\frac{\sqrt{pr}\overline{\sigma}}{\Delta}\right) \, \mathrm{and} \, 	\sup_{j}\norm{\hV_{j,1} \hV_{j,1}^{\top}-V_{j,1} V_{j,1}^{\top}}_{\rm F} 
	=O_{\Prob}\left(\frac{\sqrt{pr_{V}}}{\lambda_{\min}}\big(\sigma_{1}+\sqrt{r}\overline{\sigma}\big)\right),
\end{equation} 
where $\overline{\sigma}^2=1/\sum_{\ell=1}^L\sigma_{\ell}^{-2}$. 

In practice, we can estimate $\sigma_{\ell}$ by 
\begin{equation}\label{eq:estimate-sigma-L}
	\widehat\sigma_{\ell} =\Big\|\calT_{\ell}-\calT_{\ell}\times_1\big[\hU_{1,\ell}\; \hV_{1,\ell}\big]\big[\hU_{1,\ell}\; \hV_{1,\ell}\big]^{\top} \times_2\cdots\times_J\big[\hU_{J,\ell}\; \hV_{J,\ell}\big]\big[\hU_{J,\ell}\; \hV_{J,\ell}\big]^{\top}\Big\|_{\rm F} / \sqrt{p_1p_2\cdots p_J}, \quad \ell\in [L],
\end{equation}
and choose $\widehat{w}_{\ell}=\frac{\widehat{\sigma}_{\ell}^{-2}}{\sum_{\ell=1}^{L}\widehat{\sigma}_{\ell}^{-2}}$. We further provide the theoretical guarantee for the estimators obtained in Algorithm \ref{alg:transfer-L} with data-adaptive weights $\widehat w_{\ell}$.

\begin{thm}\label{thm:transfer-adaptive}
	Let $\hU_j, \hV_{j, 1}$'s be the outputs of Algorithm \ref{alg:transfer-L} with weights $\widehat{w}_{\ell}$ defined above. Under the conditions of Theorem \ref{thm:transfer-L}, we have that
	\begin{equation}\label{eq:error-transfer-U-adaptive}
	\sup_j\norm{\hU_{j}\hU_{j}^{\top}-U_jU_j^{\top}}_{\rm F}\leq \widetilde{C}_2\bigg( \frac{\sqrt{pr}\overline{\sigma}}{\Delta}+\frac{r\sqrt{L}\overline{\sigma}}{ p^{J/2-1}\Delta}\bigg),
\end{equation}
	and
	\begin{equation}\label{eq:error-transfer-V-adaptive}
	\sup_{j}\norm{\hV_{j,1} \hV_{j,1}^{\top}-V_{j,1} V_{j,1}^{\top}}_{\rm F} \leq \widetilde{C}_2 \frac{\sqrt{pr_{V}}}{\lambda_{\min}}\bigg(\sigma_{1}+\sqrt{r}\overline{\sigma}+\frac{r\sqrt{L}\overline{\sigma}}{p^{(J-1)/2}}\bigg),
\end{equation}
	for some constants $ \widetilde{C}_1$ and $ \widetilde{C}_2$, with probability at least $1-\widetilde {C}_1p^{-1}$.  If we further assume that $rL/p^{J-1}=o(1)$, then, with probability at least $1-\widetilde {C}_1p^{-1}$,
		\[\sup_j\norm{\hU_{j}\hU_{j}^{\top}-U_jU_j^{\top}}_{\rm F}\leq \widetilde{C}_2 \frac{\sqrt{pr}\overline{\sigma}}{\Delta}, \quad 	\sup_{j}\norm{\hV_{j,1} \hV_{j,1}^{\top}-V_{j,1} V_{j,1}^{\top}}_{\rm F} \leq \widetilde{C}_2 \frac{\sqrt{pr_{V}}}{\lambda_{\min}}\big(\sigma_{1}+\sqrt{r}\overline{\sigma}\big),\]
		which matches the error rates established in \eqref{eq:transfer-rate-L}.
\end{thm} 
The proof of Theorem \ref{thm:transfer-adaptive} is provided in Section \ref{sec:supp_proof_transfer}. Compared to \eqref{eq:transfer-rate-L}, the error rate of $\hU_j$ in \eqref{eq:error-transfer-U-adaptive} contains an additional term $r\sqrt{L}\overline{\sigma}/(p^{J/2-1}\Delta)$ due to the estimation error between $\widehat{w}_{\ell}$ and $w_{\ell}^*$, which is dominated by the first term $\sqrt{pr}\overline{\sigma}/\Delta$ under a mild condition that $rL/p^{J-1}=o(1)$. The same is true for the error rate of $\hV_{j, 1}$ in \eqref{eq:error-transfer-V-adaptive}.

\section{Technical Proof of the Theoretical Results}\label{sec:supp_theory}

\subsection{Proof of the Results for the Homogeneous Setting}
\label{sec:supp_proof_homo}
\begin{proof}[Proof for Theorem 2.1]
For $j \in [J]$, $\ell \in [L]$, define $Z_{j, \ell}=\calM_j(\calZ_{\ell})$, $T_{j}=\calM_j(\calT^*)$, and $G_{j}=\calM_j(\calG)$.
Let $U_{j\perp} \in \R^{p_j \times (p_j-r_j)}$ be the orthogonal complement of $U_j$, i.e., $[U_j \,\,  U_{j\perp}]$ is an orthogonal matrix in $\R^{p_j \times p_j}$. Denote the compact singular value decomposition of $G_j$ by $U_{G_j}\Lambda_jV^{\top}_{G_j}$, where $U_{G_j} \in \mathbb{O}^{r_j \times r_j}$, $\Lambda_j \in \R^{r_j \times r_j}$, and $V^{\top}_{G_j}V_{G_j}=I_{r_j}$. Note the model (1) is equivalent to the model with $\widetilde{\calG}=\calG \times_1 O^{\top}_{1} \times_2 O^{\top}_{2}\cdots \times_J O^{\top}_{J}$ and $\widetilde U_j = U_j O_j$, for any $O_j \in \mathbb{O}^{r_j \times r_j}$, $j\in[J]$. Taking $O_j=U_{G_j}$ leads to $\widetilde{G}_j=\Lambda_jV_{G_j}^{\top}$ and thus $\widetilde{G}_j\widetilde{G}_j^{\top}=\Lambda_j^2$, while the projection matrix $U_jO_jO_j^{\top}U_j^{\top}=U_jU_j^{\top}$ remains invariant. Therefore, without loss of generality, we assume $G_jG_j^{\top}=\Lambda_j^2$ in the sequel. 
Moreover, by the definition of $\lambda_{\min}$ and $\kappa_0$, it holds that
\begin{equation}
	\norm{\Lambda^{-1}_j}_2 \leq \lambda_{\min}^{-1}, \quad 	\norm{G_j}_2=\norm{\Lambda_j}_2 \leq \kappa_0\lambda_{\min}.
\end{equation}

Hereafter, for clear presentation, we define $p_{-j} = \prod_{k \in [J]}p_k / p_j$, $r_{-j} = \prod_{k \in [J]}r_k / r_j$, and 
\[U_{-j} = U_1 \otimes U_2 \otimes \cdots \otimes U_{j-1} \otimes U_{j+1} \otimes \cdots \otimes U_{J} \in \mathbb{O}^{p_{-j} \times r_{-j}}.\] 
For each $j\in [J]$, $\ell\in [L]$, define a ``locally-good" event
\begin{equation}\label{eq:good}
	\begin{aligned}
		&\quad E_{j,\ell}(C)\\
		&:=\bigg\{\norm{\hU_{j,\ell}^{(0)}\hU_{j,\ell}^{(0)\top}-U_{j}U_{j}^{\top}}_2 \leq C \sqrt{p}\sigma\lambda_{\min}^{-1}, \, 
		\norm{Z_{j,\ell}U_{-j}}_2 \leq C\sigma\sqrt{p}, \\ 
		&\quad \sup_{\substack{X_{j'} \in \R^{p_{j'} \times r_{j'}}\\ \norm{X_{j'}}_2 \leq 1, \forall j'\in [J]\setminus \{j\}}} \norm{Z_{j,\ell}(X_1 \otimes X_2 \otimes \cdots \otimes X_{j-1} \otimes X_{j+1}\otimes \cdots \otimes X_J)}_2 \leq C\sigma\sqrt{pr} \bigg\}.
	\end{aligned}
\end{equation}
Then define a ``globally-good'' event
\[	E(C):=\bigcap_{j=1}^J \bigg\{\Big\{\bigcap_{\ell=1}^LE_{j,\ell}(C) \Big\} \bigcap \Big\{	\norm{\overline{Z}_{j}U_{-j}}_2 \leq C\sigma\sqrt{p/L}\Big\}\bigg\},\]
where $\overline{Z}_{j}:=\frac{1}{L}\sum_{\ell=1}^{L}Z_{j, \ell}$.
We have the following lemma for the probability of $E(C)$.
\begin{lemma}\label{lem:1}
	Under the assumptions in Theorem 2.1, there exist constants $C_1^{\prime}$, $c_1^{\prime}$, $C_2^{\prime}$ such that  $\Prob[E(C_2^{\prime})] \geq 1-C_1^{\prime}Le^{-c_1^{\prime}p}$.
\end{lemma}
The proof of Lemma \ref{lem:1} is provided in Section \ref{sec:supp-theory-lemmas}. Note that we assume that $L \lesssim p^{c_3}$ for some $c_3>0$, which implies that $E(C_2')$ has probability approaching one. For the homogeneous setting, we follow \cite{xia2022inference} to decompose $\hU_j\hU_j^{\top}-U_jU_j^{\top}$. By definition, the local estimator $\hU_{j,\ell}$ is composed of the first $r_j$ left singular vectors of \[ \calM_j\big(\calT_{\ell} \times_{1} \hU_{1,\ell}^{(0)\top} \times_{2} \hU_{2,\ell}^{(0)\top} \cdots \times_{j-1}\hU_{j-1,\ell}^{(0)\top} \times_{j+1}\hU_{j+1,\ell}^{(0)\top}\cdots \times_{J}\hU_{J,\ell}^{(0)\top}\big)=	(T_j + Z_{j, \ell})\hU^{(0)}_{-j, \ell},\]  
which are also the eigenvectors of the symmetric matrix
\begin{equation}
	\quad (T_j + Z_{j, \ell}) \hU^{(0)}_{-j, \ell}\hU_{-j, \ell}^{(0)\top}(T^{\top}_j + Z^{\top}_{j, \ell})= T_jU_{-j}U_{-j}^{\top} T_j^{\top} +  \frakE_{j, \ell}= U_j\Lambda_j^2U_j^{\top} +  \frakE_{j, \ell},
	\label{eq:decomp_before_SVD-1}
\end{equation}
where 
\[\hU^{(0)}_{-j, \ell} = \hU^{(0)}_{1, \ell} \otimes \hU^{(0)}_{2,\ell} \otimes \cdots \otimes \hU^{(0)}_{j-1, \ell} \otimes \hU^{(0)}_{j+1, \ell} \otimes \cdots \otimes \hU^{(0)}_{J, \ell},\]
and 
$\frakE_{j, \ell}$ is a remainder term defined by
\begin{equation}
	\begin{aligned}
		\frakE_{j, \ell}&:=\zeta_{j,\ell,1}+\zeta^{\top}_{j,\ell,1}+\zeta_{j,\ell,2}+\zeta^{\top}_{j,\ell,2}+\zeta_{j,\ell,3}+\zeta_{j,\ell,4}+\zeta_{j,\ell,5},\\
		\zeta_{j,\ell,1}&:=T_{j} U_{-j}U_{-j}^{\top}Z^{\top}_{j,\ell}, \\
		\zeta_{j,\ell,2}&:= T_{j} \left[\hU^{(0)}_{-j,\ell}\hU_{-j,\ell}^{(0)\top}-U_{-j}U_{-j}^{\top}\right]Z^{\top}_{j,\ell},\\
		\zeta_{j,\ell,3}&:= Z_{j,\ell}U_{-j}U_{-j}^{\top}Z_{j,\ell}^{\top},\\
		\zeta_{j,\ell,4}&:= Z_{j,\ell}\left[\hU^{(0)}_{-j,\ell}\hU_{-j,\ell}^{(0)\top}-U_{-j}U_{-j}^{\top}\right]Z_{j,\ell}^{\top},\\
		\zeta_{j,\ell,5}&:=T_j \left[\hU^{(0)}_{-j,\ell}\hU_{-j,\ell}^{(0)\top}-U_{-j}U_{-j}^{\top}\right] T_j^{\top}.
	\end{aligned}
	\label{eq:frakE_decomp-1}
\end{equation} 
The last equality of \eqref{eq:decomp_before_SVD-1} follows from the facts that $T_j=U_jG_jU_{-j}^{\top}$ and $G_jG_j^{\top}=\Lambda_j^2$. We use the following lemma to provide upper bounds for each term in \eqref{eq:frakE_decomp-1}. 

\begin{lemma}\label{lem:2} Under the event $E(C_2')$ and the assumptions in Theorem 2.1, 	there exists some absolute constant $C_3>0$ such that 
	\begin{equation}
		\begin{aligned}
			&\norm{\zeta_{j,\ell,1}}_2 \leq C_3 \kappa_0\lambda_{\min}\sigma\sqrt{p}, \quad  \norm{\zeta_{j,\ell,2}}_2\leq C_3 \kappa_0p\sigma^2\sqrt{r},\\ &\norm{\zeta_{j,\ell,3}}_2 \leq C_3p\sigma^2, \quad \norm{\zeta_{j,\ell, 4}}_2 \leq C_3 p^{3/2}\sqrt{r}\sigma^3\lambda_{\min}^{-1}, \quad  \norm{\zeta_{j,\ell, 5}}_2 \leq C_3 \kappa_0^2p\sigma^2,
		\end{aligned}
		\label{eq:frakE_bound-1} 
	\end{equation}
\end{lemma}
By Lemma \ref{lem:2},  under the event $E(C_2')$, $\norm{\frakE_{j, \ell}}_2 \lesssim \kappa_0\lambda_{\min}\sigma\sqrt{p} < \lambda^2_{\min}/2$.  Applying Theorem 1 in \cite{xia2021normal} yields
\begin{equation}
	\begin{aligned}
		\hU_{j,\ell}\hU^{\top}_{j,\ell} - U_jU_j^{\top} &= U_j\Lambda_{j}^{-2}U_j^{\top}\frakE_{j, \ell}U_{j\perp}U_{j\perp}^{\top} + U_{j\perp}U_{j\perp}^{\top}\frakE_{j, \ell}U_j\Lambda_{j}^{-2}U_j^{\top} \\
		&\quad-U_j\Lambda_{j}^{-2}U_j^{\top}\frakE_{j,\ell}U_{j\perp}U_{j\perp}^{\top}\frakE_{j, \ell}U_j\Lambda_{j}^{-2}U_j^{\top} 
		-U_j\Lambda_{j}^{-2}U_j^{\top}\frakE_{j,\ell}U_j\Lambda_{j}^{-2}U_j^{\top}\frakE_{j, \ell}U_{j\perp}U_{j\perp}^{\top}\\
		&\quad-U_{j\perp}U_{j\perp}^{\top}\frakE_{j,\ell}U_j\Lambda_{j}^{-2}U_j^{\top}\frakE_{j, \ell}U_j\Lambda_{j}^{-2}U_j^{\top}+U_j\Lambda_{j}^{-4}U_j^{\top}\frakE_{j,\ell}U_{j\perp}U_{j\perp}^{\top}\frakE_{j,\ell}U_{j\perp}U_{j\perp}^{\top}\\
		&\quad+U_{j\perp}U_{j\perp}^{\top}\frakE_{j,\ell}U_{j\perp}U_{j\perp}^{\top}\frakE_{j,\ell}U_j\Lambda_{j}^{-4}U_j^{\top}+U_{j\perp}U_{j\perp}^{\top}\frakE_{j,\ell}U_j\Lambda_{j}^{-4}U_j^{\top}\frakE_{j,\ell}U_{j\perp}U_{j\perp}^{\top}+\frakR_{j, \ell},
	\end{aligned}
	\label{eq:rho_decomp_thm-1}
\end{equation}
where $\norm{\frakR_{j, \ell}}_2 \lesssim \kappa_0^3\sigma^3p^{3/2}/\lambda_{\min}^3$.
Then plugging \eqref{eq:frakE_decomp-1} into \eqref{eq:rho_decomp_thm-1} and using the fact that $U^{\top}_{j\perp}T_j=0$, we obtain
\begin{equation}
	\hU_{j,\ell}\hU^{\top}_{j,\ell} - U_jU_j^{\top}
	=\frakS_{j, \ell, 1}+\frakS_{j, \ell, 2}+\frakS_{j, \ell, 3},
	\label{eq:rho_decomp_plugin-1}
\end{equation}
where for $j\in [J]$ and $\ell \in [L]$,
\begin{equation}
	\label{eq:def-Sjlt12-1}
	\begin{aligned}
		\frakS_{j, \ell, 1}&:=U_j\Lambda_{j}^{-2}U_j^{\top}\zeta_{j, \ell,1}U_{j\perp}U_{j\perp}^{\top} +U_{j\perp}U_{j\perp}^{\top}\zeta^{\top}_{j, \ell,1}U_j\Lambda_{j}^{-2}U_j^{\top}\\
		&= U_j\Lambda_{j}^{-2}G_j U_{-j}^{\top}Z^{\top}_{j, \ell}U_{j\perp}U_{j\perp}^{\top} +  U_{j\perp}U_{j\perp}^{\top}Z_{j, \ell} U_{-j}G_j^{\top}\Lambda_{j}^{-2}U_j^{\top},\\
		\frakS_{j, \ell, 2}&:= U_j\Lambda_{j}^{-2}U_j^{\top}\zeta_{j, \ell, 2}U_{j\perp}U_{j\perp}^{\top}+U_{j\perp}U_{j\perp}^{\top}\zeta^{\top}_{j, \ell, 2}U_j\Lambda_{j}^{-2}U_j^{\top}\\
		&\quad+U_j\Lambda_{j}^{-2}U_j^{\top}\zeta_{j, \ell, 3}U_{j\perp}U_{j\perp}^{\top}+U_{j\perp}U_{j\perp}^{\top}\zeta_{j, \ell, 3}U_j\Lambda_{j}^{-2}U_j^{\top}\\
		&\quad-U_j\Lambda_{j}^{-2}U_j^{\top}\zeta_{j, \ell, 1}U_{j\perp}U_{j\perp}^{\top}\zeta^{\top}_{j, \ell, 1}U_j\Lambda_{j}^{-2}U_j^{\top}\\
		&\quad-U_j\Lambda_{j}^{-2}U_j^{\top}(\zeta_{j, \ell, 1}+\zeta^{\top}_{j, \ell, 1})U_j\Lambda_{j}^{-2}U_j^{\top}\zeta_{j,\ell, 1}U_{j\perp}U_{j\perp}^{\top}\\
		&\quad-U_{j\perp}U_{j\perp}^{\top}\zeta^{\top}_{j, \ell, 1}U_j\Lambda_{j}^{-2}U_j^{\top}(\zeta_{j, \ell, 1}+\zeta^{\top}_{j, \ell, 1})U_j\Lambda_{j}^{-2}U_j^{\top}\\
		&\quad +U_{j\perp}U_{j\perp}^{\top}\zeta^{\top}_{j, \ell, 1}U_j\Lambda_{j}^{-4}U_j^{\top}\zeta_{j, \ell, 1}U_{j\perp}U_{j\perp}^{\top},
	\end{aligned}
\end{equation}
and $\frakS_{j, \ell,3}$ is a remainder term determined by the above equations.

Under $E(C'_2)$, it holds that $\norm{Z_{j,\ell}U_{-j}}_2 \leq C'_2 \sigma\sqrt{p}$. Using the fact that $\norm{\Lambda_{j}^{-1}G_j}_2=1$ and $\norm{\Lambda_{j}^{-1}}_2 \leq \lambda_{\min}^{-1}$, we have
\begin{equation}
	\norm{\frakS_{j,\ell,1}}_2=\norm{U_j\Lambda_{j}^{-2}G_j U_{-j}^{\top}Z_{j,\ell}^{\top}U_{j\perp}U_{j\perp}^{\top} +  U_{j\perp}U_{j\perp}^{\top}Z_{j,\ell} U_{-j}G_j^{\top}\Lambda_{j}^{-2}U_j^{\top}}_2 \lesssim \lambda^{-1}_{\min}\sigma\sqrt{p}.
	\label{eq:bound-1-S1}
\end{equation} 
By \eqref{eq:frakE_bound-1}, it holds that
\begin{equation}
	\label{eq:bound-1-S23-1}
	\norm{\frakS_{j, \ell, 2}}_2 \leq C_3 \kappa_0^2pr^{1/2}\sigma^2\lambda_{\min}^{-2}, \quad \norm{\frakS_{j, \ell, 3}}_2 \leq C_3 \kappa_0^3p^{3/2}r^{1/2}\sigma^3\lambda_{\min}^{-3},
\end{equation}
for some absolute constant $C_3>0$. 
Then
\begin{equation}
	\begin{aligned}
		&\quad \frac{1}{L}\sum_{\ell=1}^L\hU_{j,\ell}\hU^{\top}_{j,\ell} - U_jU_j^{\top}\\ &=\overline{\frakS_{j, 1}}+ \overline{\frakS_{j, 2}} + \overline{\frakS_{j, 3}}\\
		&= U_j\Lambda_{j}^{-2}G_j U_{-j}^{\top}\overline{Z}^{\top}_{j}U_{j\perp}U_{j\perp}^{\top} +  U_{j\perp}U_{j\perp}^{\top}\overline{Z}_{j} U_{-j}G_j^{\top}\Lambda_{j}^{-2}U_j^{\top}+ \overline{\frakS_{j, 2}} + \overline{\frakS_{j, 3}},
	\end{aligned}
\end{equation}
where $\overline{Z}_j:=(1/L)\sum_{\ell=1}^L Z_{j, \ell}$ and $\overline{\frakS_{j,k}}:=(1/L)\sum_{\ell=1}^L \frakS_{j, \ell, k}$, $k \in \{1,2,3\}$. Note that \[	\norm{\overline{\frakS_{j,2}}}_2 \lesssim \kappa_0^2pr^{1/2}\sigma^2\lambda_{\min}^{-2}, \quad \norm{\overline{\frakS_{j,3}}}_2 \lesssim \kappa_0^3p^{3/2}r^{1/2}\sigma^3\lambda_{\min}^{-3}.\] 

Under $E(C_2')$, it holds that $\norm{\overline{Z}_{j}U_{-j}}_2 \lesssim \sigma\sqrt{p/L}$ and thus
\begin{equation}\label{eq:Sj1} \norm{\overline{\frakS_{j,1}}}_2=\norm{U_j\Lambda_{j}^{-2}G_j U_{-j}^{\top}\overline{Z}^{\top}_jU_{j\perp}U_{j\perp}^{\top} +  U_{j\perp}U_{j\perp}^{\top}\overline{Z}_{j} U_{-j}G_j^{\top}\Lambda_{j}^{-2}U_j^{\top}}_2 \lesssim \lambda^{-1}_{\min}\sigma\sqrt{p/L}.
\end{equation}

Therefore, we obtain that
\begin{equation}
	\label{eq:rate-hU2-1}
	\frac{1}{L}\sum_{\ell=1}^L\hU_{j,\ell}\hU^{\top}_{j,\ell} - U_jU_j^{\top}=\overline{\frakS_{j, 1}}+ \overline{\frakS_{j, 2}} + \overline{\frakS_{j, 3}},
\end{equation}
with
\begin{equation}
	\label{eq:bound-2-S123-1}
	\norm{\overline{\frakS_{j,1}}}_2 \lesssim \lambda^{-1}_{\min}\sigma\sqrt{p/L}, \quad
	\norm{\overline{\frakS_{j,2}}}_2 \lesssim \kappa_0^2pr^{1/2}\sigma^2\lambda_{\min}^{-2}, \quad \norm{\overline{\frakS_{j,3}}}_2 \lesssim \kappa_0^3p^{3/2}r^{1/2}\sigma^3\lambda_{\min}^{-3},
\end{equation}
with probability at least $1-C'_1Le^{-c'_1p}$ for some absolute constant $C_1', c_1'>0$.

Since the columns of $\hU_j$ are the first $r_j$ eigenvectors of $\frac{1}{L}\sum_{\ell=1}^L\hU_{j,\ell}\hU^{\top}_{j,\ell}$, similar to  \eqref{eq:rho_decomp_thm-1},
\begin{equation}
	\label{eq:decomp-after-SVD-thm-1}
	\begin{aligned}
		\hU_j\hU_j^{\top}-U_jU_j^{\top}&=U_jU_j^{\top}\left(\overline{\frakS_{j, 1}}+ \overline{\frakS_{j, 2}}\right)U_{j\perp}U_{j\perp}^{\top}+ 
		U_{j\perp}U_{j\perp}^{\top}\left(\overline{\frakS_{j, 1}}+ \overline{\frakS_{j, 2}}\right)U_jU_j^{\top} \\ 
		&\quad -U_jU_j^{\top}\overline{\frakS_{j, 1}}U_{j\perp}U_{j\perp}^{\top} \overline{\frakS_{j, 1}}U_jU_j^{\top} + U_{j\perp}U_{j\perp}^{\top}\overline{\frakS_{j, 1}}U_jU_j^{\top}\overline{\frakS_{j, 1}}U_{j\perp}U_{j\perp}^{\top}
		+ \widetilde{\frakR}_j,
	\end{aligned}
\end{equation}
where $\norm{\widetilde{\frakR}_j}_2 \lesssim \kappa_0^3p^{3/2}r^{1/2}\sigma^3\lambda_{\min}^{-3}$. By \eqref{eq:bound-2-S123-1}, we obtain that
\begin{equation}
	\label{eq:S1-terms-rate}
	\begin{aligned}
			&\norm{U_jU_j^{\top}\overline{\frakS_{j, 1}}U_{j\perp}U_{j\perp}^{\top}+ 
			U_{j\perp}U_{j\perp}^{\top}\overline{\frakS_{j, 1}}U_jU_j^{\top}}_2 \lesssim \lambda^{-1}_{\min}\sigma\sqrt{p/L},\\
			&\norm{U_jU_j^{\top}\overline{\frakS_{j, 2}}U_{j\perp}U_{j\perp}^{\top}+ 
			U_{j\perp}U_{j\perp}^{\top}\overline{\frakS_{j, 2}}U_jU_j^{\top}}_2 \lesssim \kappa_0^2pr^{1/2}\sigma^2\lambda_{\min}^{-2},\\
			& \norm{U_jU_j^{\top}\overline{\frakS_{j, 1}}U_{j\perp}U_{j\perp}^{\top} \overline{\frakS_{j, 1}}U_jU_j^{\top} - U_{j\perp}U_{j\perp}^{\top}\overline{\frakS_{j, 1}}U_jU_j^{\top}\overline{\frakS_{j, 1}}U_{j\perp}U_{j\perp}^{\top}}_2 \lesssim \lambda^{-2}_{\min}\sigma^2p/L.
	\end{aligned}
\end{equation}
Therefore, with probability at least $1-C_1'Le^{-c'_1p}$, 
\[\norm{\hU_j\hU_j^{\top}-U_jU_j^{\top}}_{\rm F} \leq \sqrt{2r}\norm{\hU_j\hU_j^{\top}-U_jU_j^{\top}}_2 \lesssim \lambda^{-1}_{\min}\sigma\sqrt{pr/L}+ \kappa_0^2pr\sigma^2\lambda_{\min}^{-2}.\]
Since $L \lesssim p^{c_3}$, there exist constants $\widetilde C_1$, $\widetilde c_1$ such that $1-C_1'Le^{-c'_1p} \geq 1-\widetilde C_1e^{-\widetilde c_1p}$, 
which proves (3).
\end{proof}

\begin{proof}[Proof of Theorem \ref{thm:second-lower}]

By \eqref{eq:decomp-after-SVD-thm-1} and \eqref{eq:S1-terms-rate}, it suffices to show that, for some constants $C_1$, $c_1$, $C_2$, there exist $\calT^* \in \mathcal{F}_{p, 1, \lambda_{\min}}$ and $\big(\hU_{1, \ell}^{(0)}, \hU_{2, \ell}^{(0)}, \hU_{3, \ell}^{(0)}\big)\in\mathfrak{U}_{C_1, c_1, C_2}$ for all $\ell=1,2,\dots, L$, such that 
\[\norm{U_jU_j^{\top}\overline{\frakS_{j, 2}}U_{j\perp}U_{j\perp}^{\top}+ 
	U_{j\perp}U_{j\perp}^{\top}\overline{\frakS_{j, 2}}U_jU_j^{\top}}_2 \gtrsim p\sigma^2\lambda_{\min}^{-2},\]
with probability at least $1-CLe^{-cp}-Cp^{-L}$ for some constants $c, C$.
	
By definition, we have
\begin{equation}
\label{eq:S2-terms}
\begin{aligned}
	&\quad U_jU_j^{\top}\overline{\frakS_{j, 2}}U_{j\perp}U_{j\perp}^{\top}+ 
	U_{j\perp}U_{j\perp}^{\top}\overline{\frakS_{j, 2}}U_jU_j^{\top}\\
	&=\frac{1}{L}\sum_{\ell=1}^{L}\bigg[U_j\Lambda_{j}^{-2}U_j^{\top}\zeta_{j, \ell, 2}U_{j\perp}U_{j\perp}^{\top}+U_{j\perp}U_{j\perp}^{\top}\zeta_{j, \ell, 2}^{\top}U_j\Lambda_{j}^{-2}U_j^{\top}\\
	&\quad +U_j\Lambda_{j}^{-2}U_j^{\top}\zeta_{j, \ell, 3}U_{j\perp}U_{j\perp}^{\top}+U_{j\perp}U_{j\perp}^{\top}\zeta_{j, \ell, 3}^{\top}U_j\Lambda_{j}^{-2}U_j^{\top}\\
	&\quad -U_j\Lambda_{j}^{-2}U_j^{\top}(\zeta_{j, \ell, 1}+\zeta^{\top}_{j, \ell, 1})U_j\Lambda_{j}^{-2}U_j^{\top}\zeta_{j,\ell, 1}U_{j\perp}U_{j\perp}^{\top}\\
	&\quad-U_{j\perp}U_{j\perp}^{\top}\zeta^{\top}_{j, \ell, 1}U_j\Lambda_{j}^{-2}U_j^{\top}(\zeta_{j, \ell, 1}+\zeta^{\top}_{j, \ell, 1})U_j\Lambda_{j}^{-2}U_j^{\top}\bigg]\\
	&=\frac{1}{L}\sum_{\ell=1}^{L}[\mathfrak{A}_{j, \ell, 1}+\mathfrak{A}_{j, \ell, 2}+\mathfrak{A}_{j, \ell, 3}],
\end{aligned}
\end{equation}
where 
\begin{align*}
	\mathfrak{A}_{j, \ell, 1}&:=U_j\Lambda_{j}^{-2}G_jU_{-j}^{\top}\left[\hU^{(0)}_{-j,\ell}\hU_{-j,\ell}^{(0)\top}-U_{-j}U_{-j}^{\top}\right]Z_{j,\ell}^{\top}U_{j\perp}U_{j\perp}^{\top}\\
	&\quad+U_{j\perp}U_{j\perp}^{\top}{Z_{j, \ell}\left[\hU^{(0)}_{-j,\ell}\hU_{-j,\ell}^{(0)\top}-U_{-j}U_{-j}^{\top}\right]}G^{\top}_j\Lambda_{j}^{-2}U_j^{\top},\\
	\mathfrak{A}_{j, \ell, 2}&:=U_j\Lambda_{j}^{-2}U_j^{\top}Z_{j, \ell}U_{-j}U_{-j}^{\top}Z_{j, \ell}^{\top}U_{j\perp}U_{j\perp}^{\top}+U_{j\perp}U_{j\perp}^{\top}{Z_{j, \ell}U_{-j}U_{-j}^{\top}Z_{j, \ell}^{\top}}U_j\Lambda_{j}^{-2}U_j^{\top}\\
	\mathfrak{A}_{j, \ell, 3}&:=-U_j\Lambda_{j}^{-2}U_j^{\top}(\zeta_{j, \ell, 1}+\zeta^{\top}_{j, \ell, 1})U_j\Lambda_{j}^{-2}U_j^{\top}\zeta_{j,\ell, 1}U_{j\perp}U_{j\perp}^{\top}\\
	&\quad-U_{j\perp}U_{j\perp}^{\top}\zeta^{\top}_{j, \ell, 1}U_j\Lambda_{j}^{-2}U_j^{\top}(\zeta_{j, \ell, 1}+\zeta^{\top}_{j, \ell, 1})U_j\Lambda_{j}^{-2}U_j^{\top}.
\end{align*}

Now consider the following setting: Let $J=3$, $r_1=r_2=r_3=1$, $p_1=p_2=p_3=p$, $\calG=\lambda_{\min}$, and $U_1=U_2=U_3=(1, \bm{0}^{\top}_{p-1, 1})^{\top}$, where $\bm{0}_{n_1, n_2}$ denotes a zero matrix of size $n_1 \times n_2$. For machine $\ell$, let the initial $\hU^{(0)}_{1,\ell}=U_1$, $\hU^{(0)}_{3,\ell}=U_3$, and $\hU^{(0)}_{2,\ell}=\left(Q_{1, \ell}, \frac{1}{\lambda_{\min}Q_{1, \ell}}\bm{z}_{\ell}^{\top} \right)^{\top}$, where $\bm{z}_{\ell}^{\top}:=(Z_{p, p+1, \ell}, Z_{p, 2p+1, \ell}, \dots, Z_{p, p(p-1)+1, \ell})$, $Q_{1,\ell}:=\sqrt{\frac{1+\sqrt{1-4\norm{\bm{z}_{\ell}}_2^2\lambda_{\min}^{-2}}}{2}}$, and $Z_{i_1, i_2,\ell}$ denotes the $(i_1, i_2)$-th entry in $Z_{1, \ell} \in \R^{p \times p^2}$. Using inequality that, for all $0<t<1$ and $Z_k \sim_{i.i.d.} \calN(0, 1)$, 
\begin{equation}\label{eq:chisq-concentration}
	\Prob\left[\abs{\frac{1}{n}\sum_{k=1}^{n}Z_{k}^2-1}\geq t\right] \leq 2e^{-nt^2/8},
\end{equation}
we have, with probability at least $1-2e^{-p/32}$,
\begin{equation}\label{eq:sum-Z2}
\abs{\sum_{k=1}^{p-1}Z_{p, kp+1, \ell}^2-(p-1)\sigma^2}\geq \frac{1}{2}(p-1)\sigma^2,
\end{equation}
which ensures that $1-4\norm{\bm{z}_{\ell}}_2^2\lambda_{\min}^{-2}>1-6p\sigma^2\lambda_{\min}^{-2}>\frac{1}{2}$, using the assumption that $\sqrt{p}\sigma/\lambda_{\min}=o(1)$. Then the definition of $Q_{1,\ell}$ forces that $\hU^{(0)\top}_{2,\ell}\hU^{(0)}_{2,\ell}=Q_{1,\ell}^2+\frac{\norm{\bm{z}_{\ell}}_2^2}{\lambda_{\min}^2Q^2_{1,\ell}}=1$. Now we show that this initial satisfies the assumption in Theorem 2.1. Note that
\begin{align*}
	\norm{\hU^{(0)}_{2}\hU^{(0)\top}_{2}-U_2U_2^{\top}}_2 &=\norm{\begin{pmatrix}
			Q_{1,\ell}^2-1 & \bm{z}_{\ell}^{\top} /\lambda_{\min}\\\
			\bm{z}_{\ell} /\lambda_{\min}\ & \bm{z}_{\ell}\bm{z}_{\ell}^{\top} /(Q_{1, \ell}^2\lambda^2_{\min})
	\end{pmatrix}}_2 \\
&\leq \norm{\begin{pmatrix}
		Q_{1,\ell}^2-1 & \bm{z}_{\ell}^{\top} /\lambda_{\min}\\\
		\bm{z}_{\ell} /\lambda_{\min}\ & \bm{z}_{\ell}\bm{z}_{\ell}^{\top} /(Q_{1, \ell}^2\lambda^2_{\min})
		\end{pmatrix}}_{\rm F} \\
&=\sqrt{(Q_{1,\ell}^2-1)^2+2\norm{\bm{z}_{\ell}}_2^2\lambda_{\min}^{-2}+\norm{\bm{z}_{\ell}}_2^4\lambda_{\min}^{-4}Q_{1,\ell}^{-4}}\\
& \leq \sqrt{21}\sqrt{p}\sigma/\lambda_{\min},
\end{align*}
with probability at least $1-2e^{-p/32}$, where we use \eqref{eq:sum-Z2} and that \begin{equation}\label{eq:V2-1-bound}
	\abs{Q_{1,\ell}^{2}-1}=\frac{2\norm{\bm{z}_{\ell}}_2^2\lambda_{\min}^{-2}}{1+\sqrt{1-4\norm{\bm{z}_{\ell}}_2^2\lambda_{\min}^{-2}}} \leq 3 p\sigma^2\lambda_{\min}^{-2}.
\end{equation}

Now we calculate the terms in \eqref{eq:S2-terms} for $j=1$ under the above setting. Note that, $U_{-1}=U_2\otimes U_3=(1, \bm{0}_{p^2-1, 1}^{\top})^{\top}$, $U_{1\perp}U_{1\perp}^{\top}=\begin{pmatrix}0& \bm{0}^{\top}_{p-1, 1}\\\bm{0}_{p-1, 1}&I_{p-1}\end{pmatrix}$, and $G_1=\Lambda_1=\lambda_{\min}$ under this specific setting. Moreover, 
\begin{align*}
	\hU^{(0)}_{-1,\ell}\hU_{-1,\ell}^{(0)\top}-U_{-1}U_{-1}^{\top}&=\left(\hU^{(0)}_{2}\hU^{(0)\top}_{2}-U_2U_2^{\top}\right)\otimes(U_3U_3^{\top})\\
	&= \begin{pmatrix}
	Q_{1,\ell}^2-1 & \bm{z}_{\ell}^{\top} /\lambda_{\min}\\\
	\bm{z}_{\ell} /\lambda_{\min}\ & \bm{z}_{\ell}\bm{z}_{\ell}^{\top} /(Q_{1, \ell}^2\lambda^2_{\min})
	\end{pmatrix}\otimes \begin{pmatrix}
	1 & \bm{0}_{p-1, 1}^{\top}\\
	\bm{0}_{p-1, 1} & \bm{0}_{p-1, p-1}
	\end{pmatrix}
\end{align*}

Therefore
\begin{align*}
	&\quad U_1\Lambda_{1}^{-2}G_1U_{-1}^{\top}\left[\hU^{(0)}_{-1,\ell}\hU_{-1,\ell}^{(0)\top}-U_{-1}U_{-1}^{\top}\right]Z_{1,\ell}^{\top}U_{1\perp}U_{1\perp}^{\top}\\
	&= \lambda_{\min}^{-1}\begin{pmatrix}
		1 & \bm{0}_{1, p^2-1}\\
		\bm{0}_{p-1, 1} & \bm{0}_{p-1, p^2-1}
	\end{pmatrix} \left[\begin{pmatrix}
	Q_{1,\ell}^2-1 & \bm{z}_{\ell}^{\top} /\lambda_{\min}\\\
	\bm{z}_{\ell} /\lambda_{\min}\ & \bm{z}_{\ell}\bm{z}_{\ell}^{\top} /(Q_{1, \ell}^2\lambda^2_{\min})
	\end{pmatrix}\otimes \begin{pmatrix}
		1 & \bm{0}_{p-1, 1}^{\top}\\
		\bm{0}_{p-1, 1} & \bm{0}_{p-1, p-1}
	\end{pmatrix}\right]Z_{1,\ell}^{\top}\begin{pmatrix}0& \bm{0}^{\top}_{p-1, 1}\\\bm{0}_{p-1, 1}&I_{p-1}\end{pmatrix}\\
	&= \lambda_{\min}^{-1}\begin{pmatrix}
		0& \widetilde{\bm{z}}_{\ell}^{\top}\\
		\bm{0}_{p-1, 1} & \bm{0}_{p-1, p-1}
		\end{pmatrix},
\end{align*}
where the $i$-th entry of $\widetilde{\bm{z}}_{\ell} \in \R^{p-1}$ is $(Q_{1,\ell}^2-1)Z_{i+1, 1, \ell}+\lambda_{\min}^{-1}\sum_{k=1}^{p-1}Z_{i+1,kp+1,\ell}Z_{p,kp+1,\ell}$.
In particular, the $(p-1)$-th entry of $\widetilde{\bm{z}}_{\ell}$ is $(Q_{1,\ell}^2-1)Z_{p, 1, \ell}+\lambda_{\min}^{-1}\sum_{k=1}^{p-1}Z^2_{p,kp+1,\ell}$.
Hence,
\begin{equation}\label{eq:eigenvalue}
	\begin{aligned}
			&\quad \frac{1}{L}\sum_{\ell=1}^{L}\mathfrak{A}_{j, \ell, 1}=\frac{1}{L} \lambda_{\min}^{-1}\begin{pmatrix}
			0& \sum_{\ell=1}^{L}\widetilde{\bm{z}}_{\ell}^{\top}\\
			\sum_{\ell=1}^{L}\widetilde{\bm{z}}_{\ell} & \bm{0}_{p-1, p-1}
		\end{pmatrix}\\
		&= \lambda_{\min}^{-1}\norm{ \frac{1}{L}\sum_{\ell=1}^{L}\widetilde{\bm{z}}_{\ell}}_2\begin{pmatrix}\frac{\sqrt{2}}{2}\\\frac{\sqrt{2}}{2} \frac{ \frac{1}{L}\sum_{\ell=1}^{L}\widetilde{\bm{z}}_{\ell}}{\norm{ \frac{1}{L}\sum_{\ell=1}^{L}\widetilde{\bm{z}}_{\ell}}_2}\end{pmatrix}\begin{pmatrix}\frac{\sqrt{2}}{2}\\\frac{\sqrt{2}}{2} \frac{ \frac{1}{L}\sum_{\ell=1}^{L}\widetilde{\bm{z}}_{\ell}}{\norm{ \frac{1}{L}\sum_{\ell=1}^{L}\widetilde{\bm{z}}_{\ell}}_2}\end{pmatrix}^{\top}\\
		&\quad- \lambda_{\min}^{-1}\norm{ \frac{1}{L}\sum_{\ell=1}^{L}\widetilde{\bm{z}}_{\ell}}_2\begin{pmatrix}\frac{\sqrt{2}}{2}\\-\frac{\sqrt{2}}{2} \frac{ \frac{1}{L}\sum_{\ell=1}^{L}\widetilde{\bm{z}}_{\ell}}{\norm{ \frac{1}{L}\sum_{\ell=1}^{L}\widetilde{\bm{z}}_{\ell}}_2}\end{pmatrix}\begin{pmatrix}\frac{\sqrt{2}}{2}\\-\frac{\sqrt{2}}{2} \frac{ \frac{1}{L}\sum_{\ell=1}^{L}\widetilde{\bm{z}}_{\ell}}{\norm{ \frac{1}{L}\sum_{\ell=1}^{L}\widetilde{\bm{z}}_{\ell}}_2}\end{pmatrix}^{\top},
	\end{aligned}
\end{equation}
which yields that 
\begin{align}
	&\quad\norm{ \frac{1}{L}\sum_{\ell=1}^{L}\mathfrak{A}_{j, \ell, 1}}_2= \lambda_{\min}^{-1}\norm{ \frac{1}{L}\sum_{\ell=1}^{L}\widetilde{\bm{z}}_{\ell}}_2 \geq  \frac{1}{L} \lambda_{\min}^{-2}\sum_{\ell=1}^{L}\sum_{k=1}^{p-1}Z^2_{p,kp+1,\ell}-\frac{1}{L}\lambda_{\min}^{-1}\abs{\sum_{\ell=1}^{L}(Q_{1, \ell}^2-1)Z_{p, 1, \ell}}.
\end{align}
By \eqref{eq:sum-Z2} and \eqref{eq:V2-1-bound}, with probability at least $1-2Le^{-p/32}$, $\frac{1}{2}(p-1)\sigma^2<\abs{\sum_{k=1}^{p-1}Z_{p, kp+1, \ell}^2}<\frac{3}{2}(p-1)\sigma^2$ and $\abs{Q_{1,\ell}^2-1} \leq 3p\sigma^2\lambda_{\min}^{-2}$ for all $\ell$, which further leads to
\[ \frac{1}{L} \lambda_{\min}^{-2}\sum_{\ell=1}^{L}\sum_{k=1}^{p-1}Z^2_{p,kp+1,\ell} \geq \frac{1}{2}(p-1)\sigma^2\lambda_{\min}^{-2}.\]
Also, since $\abs{Z_{p, 1, \ell}}$ is $\sigma^2$-sub-Gaussian, and $\E[\abs{Z_{p, 1, \ell}}]=\sigma\sqrt{\frac{2}{\pi}}$,
we have
\[\Prob\left(\abs{\frac{Z_{p, 1, \ell}}{\sigma}-\sqrt{\frac{2}{\pi}}} > t \right)\leq 2e^{-t^2/2}.\]
Hence, with probability at least $1-2Le^{-p/2}$,
\[\frac{1}{L}\sum_{\ell=1}^{L}\abs{Z_{p, 1, \ell}} \leq 2\sigma\sqrt{p}.\]
Therefore,
\[\frac{1}{L}\lambda_{\min}^{-1}\abs{\sum_{\ell=1}^{L}(Q_{1, \ell}^2-1)Z_{p, 1, \ell}} \leq 6p\sigma^2\lambda_{\min}^{-2}(\sigma\lambda_{\min}^{-1}\sqrt{p})=o(p\sigma^2\lambda_{\min}^{-2}),\] 
and thus
\begin{align}\label{eq:frakA-1-bound}
	&\quad\norm{ \frac{1}{L}\sum_{\ell=1}^{L}\mathfrak{A}_{j, \ell, 1}}_2 \gtrsim p\sigma^2\lambda_{\min}^{-2},
\end{align}
with probability at least $1-4Le^{-p/32}$.

For $\mathfrak{A}_{j, \ell, 2}$ and $\mathfrak{A}_{j, \ell, 3}$, it is straightforward to calculate that
\[\mathfrak{A}_{j, \ell, 2}+\mathfrak{A}_{j, \ell, 3}=-\lambda_{\min}^{-2}\begin{pmatrix}
	0 & Z_{1, 1, \ell}Z_{2, 1, \ell} & \cdots & Z_{1, 1, \ell}Z_{p, 1, \ell}\\
	Z_{1, 1, \ell}Z_{2, 1, \ell} &&&\\
	\vdots &&\bm{0}_{p-1,p-1}& \\\
	Z_{1,1,\ell}Z_{p,1,\ell}&&&
\end{pmatrix}.\]
Similar to  \eqref{eq:eigenvalue}, it holds that 
\[\norm{ \frac{1}{L}\sum_{\ell=1}^{L}(\mathfrak{A}_{j, \ell, 2}+\mathfrak{A}_{j, \ell, 3})}_2=\lambda_{\min}^{-2} \sqrt{\sum_{k=2}^{p}\left(\frac{1}{L}\sum_{\ell=1}^{L}Z_{1,1,\ell}Z_{k,1,\ell}\right)^2}.
\]
Since $\norm{Z_{1,1,\ell}Z_{k,1,\ell}}_{\psi_1} \leq \norm{Z_{1,1,\ell}}_{\psi_2}\norm{Z_{k,1,\ell}}_{\psi_2} \leq C_{\sigma}\sigma^2$ for some constant $C_{\sigma}$, applying Bernstein's inequality leads to 
\begin{equation}\label{eq:z1z2-tail}
	\Prob\left(\abs{\sum_{\ell=1}^{L}Z_{1,1,\ell}Z_{k,1,\ell}} \geq CL\sigma^2\log p\right) \leq 2p^{-2L},
\end{equation}
for some constant $C>0$. Therefore,
\begin{equation}\label{eq:frakA-23-bound}
\norm{  \frac{1}{L}\sum_{\ell=1}^{L}(\mathfrak{A}_{j, \ell, 2}+\mathfrak{A}_{j, \ell, 3})}_2 \leq C\sigma^2\lambda_{\min}^{-2}\sqrt{p}\log p \ll p\sigma^2\lambda_{\min}^{-2}, 
\end{equation}
with probability at least $1-2p^{-2L+1} \geq 1-2p^{-L}$.

Combining \eqref{eq:S2-terms}, \eqref{eq:frakA-1-bound}, and \eqref{eq:frakA-23-bound}, we have that, with probability at least $1-4Le^{-p/32}-2p^{-L}$, 
\[\norm{U_jU_j^{\top}\overline{\frakS_{j, 2}}U_{j\perp}U_{j\perp}^{\top}+ 
	U_{j\perp}U_{j\perp}^{\top}\overline{\frakS_{j, 2}}U_jU_j^{\top}}_2 \geq \norm{  \frac{1}{L}\sum_{\ell=1}^{L}\mathfrak{A}_{j, \ell, 1}}_2-\norm{  \frac{1}{L}\sum_{\ell=1}^{L}(\mathfrak{A}_{j, \ell, 2}+\mathfrak{A}_{j, \ell, 3})}_2 \gtrsim p\lambda_{\min}^{-2}\sigma^2.\]
\end{proof}

\subsection{Proof of the Results for the Heterogeneous Setting}\label{sec:supp_proof_hetero}
\begin{proof}[Proof of Theorem 3.1]  Similar to the proof of Theorem 2.1, define $T^*_{j, \ell}=\calM_j(\calT^*_{\ell})$, and
\[[U \, V]_{-j, \ell} = [U_1  \, V_{1, \ell}] \otimes [U_2  \, V_{2, \ell}] \otimes \cdots \otimes [U_{j-1}  \, V_{j-1, \ell}] \otimes [U_{j+1}  \, V_{j+1, \ell}] \otimes \cdots \otimes [U_J  \, V_{J, \ell}].\] 

Recall that $\Lambda_{j, \ell}$ denotes the singular value matrix of $\calM_j(\calG_{\ell})=:G_{j,\ell}$. By (8), $G_{j,\ell}G_{j,\ell}^{\top}=\Lambda_{j, \ell}^2$. Let $\Lambda_{j,U}$ denote the top-left $r_{j,U} \times r_{j,U}$ diagonal block of $\Lambda_{j,\ell}$ and  $\Lambda_{j,V_{\ell}}$ be the remaining block, i.e., $\Lambda_{j,\ell}=\begin{pmatrix}
	\Lambda_{j,U} & 0\\0& 	\Lambda_{j,V_{\ell}}
\end{pmatrix}$. By (8) and \eqref{eq:decomp_before_SVD-1}, 
the columns of $\hU_{j,\ell}$ are the first $r_{j,U}$ eigenvectors  of 
\begin{equation}
	U_j\Lambda_{j, U}^2U_j^{\top} +  V_{j,\ell}\Lambda_{j, V_{\ell}}^2V_{j,\ell}^{\top} + \frakE_{j, \ell},
	\label{eq:decomp_before_SVD_hetero}
\end{equation}
where $ \frakE_{j, \ell}$ is a remainder term satisfying
\begin{equation}
	\sup_{j,\ell}\norm{\frakE_{j, \ell}}_2\leq C'_2\lambda_{\min}\kappa_0\sigma\sqrt{p},
	\label{eq:frakE_1_hetero}
\end{equation} 
\begin{equation}
	\sup_{j,\ell}\norm{\frakE_{j, \ell}-\zeta_{j,\ell,1}+\zeta^{\top}_{j,\ell,1}}_2\leq C'_2\kappa_0^2p\sigma^2\sqrt{r},
	\label{eq:frakE_2_hetero}
\end{equation} 
for
\[\zeta_{j,\ell,1}=T^*_{j,\ell} [U \, V]_{-j, \ell} [U \, V]_{-j, \ell} ^{\top}Z^{\top}_{j,\ell}=[U_j \, V_{j,\ell}]G_{j, \ell}[U \, V]_{-j, \ell}^{\top}Z_{j,\ell}^{\top},\] and some absolute constant $C'_2$, with probability at least $1-C_1'Le^{-c_1'p}$.  

Let the $(\bu_{1, j, \ell}, \dots,\bu_{p_j, j, \ell})$ be the eigenvectors of $U_j\Lambda_{j, U}^2U_j^{\top} +  V_{j,\ell}\Lambda_{j, V_{\ell}}^2V_{j,\ell}^{\top}$ and the corresponding eigenvalues be $\lambda_{1,j,\ell}^2>\dots>\lambda_{p_j,j,\ell}^2$. For $k=1,\dots, r_{j, U}$,  let $O_{k,j,\ell}=\sum_{k'=r_{j, U}+1}^{p_j}(\lambda^2_{k,j,\ell}-\lambda^2_{k',j,\ell})^{-1}\bu_{k'}\bu_{k'}^{\top}$, and define
\[
f_{j,\ell}:\R^{p_j \times r_{j, U}} \to \R^{p_j \times r_{j, U}}, (\bm w_1,\dots,\bm w_{r_{j, U}})\to(O_{1,j,\ell}\bm w_1,\dots, O_{r_{j, U},j,\ell}\bm w_{r_{j,U}}).
\] 
Since, with probability at least $1-C_1'Le^{-c_1'p}$, 
\[\sup_{j, \ell}\norm{\frakE_{j, \ell}}_2 \lesssim \kappa_0\lambda_{\min}\sigma\sqrt{p} \ll \lambda_{\min}\Delta \leq \lambda_{r_{j, U},j,\ell}^2-\lambda_{r_{j, U}+1,j,\ell}^2,\] (using $(\Delta/\sigma) \gg \sqrt{p}$ and $\kappa_0=O(1)$), we can apply Lemma 2 in \cite{fan2019distributed}  and obtain that
\begin{equation}\label{eq:decomp-hetero}
\hU_{j, \ell}\hU_{j, \ell}^{\top}=U_jU_j^{\top}+f_{j,\ell}(\frakE_{j, \ell}U_j)U_j^{\top}+U_jf_{j, \ell}(\frakE_{j, \ell}U_j)^{\top}+\frakR_{j, \ell},
\end{equation}
where the remainder term satisfies 
\begin{equation}\label{eq:second-1-hetero}
	\sup_{j, \ell}\norm{\frakR_{j, \ell}}_{\rm F} \lesssim \frac{\sqrt{r_{j, U}}\norm{\frakE_{j, \ell}}_2^2}{\Delta^2\lambda_{\min}^2}\lesssim \frac{\sqrt{r}p\kappa_0^2\sigma^2}{\Delta^2}.\
\end{equation}

By the linearity of $f_{j,\ell}$, we further have
\[f_{j,\ell}(\frakE_{j, \ell}U_j)=f_{j,\ell}\big[\zeta_{j, \ell, 1}U_j\big]+f_{j,\ell}\big[\zeta_{j, \ell, 1}^{\top}U_j\big]+f_{j,\ell}\big[(\frakE_{j, \ell}-\zeta_{j, \ell, 1}-\zeta_{j, \ell, 1}^{\top})U_j\big].\]
Note that $\norm{f_{j, \ell}(X)}_{\rm F} \leq \frac{1}{\inf_{j, \ell} (\lambda^2_{r_{j, U}, j, \ell}-\lambda^2_{r_{j, U}+1, j, \ell})} \norm{X}_{\rm F}$ and $\lambda^2_{r_{U}, j, \ell}-\lambda^2_{r_{U}+1, j, \ell} \geq \Delta^2$. Hence
\begin{equation}\label{eq:second-2-hetero}
	\sup_{j, \ell}\norm{f_{j,\ell}\big[(\frakE_{j, \ell}-\zeta_{j, \ell, 1}-\zeta_{j, \ell, 1}^{\top})U_j\big]}_{\rm F}\leq \Delta^{-2}\sup_{j, \ell} \norm{(\frakE_{j, \ell}-\zeta_{j, \ell, 1}-\zeta_{j, \ell, 1}^{\top})U_j}_{\rm F} \lesssim \kappa_0^2pr\sigma^2\Delta^{-2}. 
\end{equation}

Now we bound $\frac{1}{L}\sum_{\ell=1}^{L}f_{j,\ell}\big[\zeta_{j, \ell, 1}U_j\big]$ and $\frac{1}{L}\sum_{\ell=1}^{L}f_{j,\ell}\big[\zeta_{j, \ell, 1}^{\top}U_j\big]$.  Since $\ve(G_{j, \ell}[U\, V]_{-j, \ell}^{\top}Z_{j,\ell}^{\top}U_j)=(U_j^{\top} \otimes G_{j, \ell}[U \, V]_{-j, \ell}^{\top})\ve(Z_{j,\ell}^{\top})$, we have that \[\ve(G_{j, \ell}[U \, V]_{-j, \ell}^{\top}Z_{j,\ell}^{\top}U_j) \sim \calN(0, \sigma^2(I_{r_{j, U}} \otimes \Lambda_{j, \ell}^2)),\]
which shows that $G_{j, \ell}[U \, V]_{-j, \ell}^{\top}Z_{j,\ell}^{\top}U_j$ has independent normal entries with variance at most $\lambda_{\max}^2\sigma^2$. Let 
$G_{j, \ell}[U \, V]_{-j, \ell}^{\top}Z_{j,\ell}^{\top}U_j=(\bm w_{1,j,\ell},\dots,\bm w_{r_{j, U}, j, \ell})$. Since \[\frac{1}{L}\sum_{\ell=1}^{L}f_{j,\ell}\big[\zeta_{j, \ell, 1}U_j\big]=\frac{1}{L}\sum_{\ell=1}^{L}f_{j, \ell}([U_j \, V_{j,\ell}]G_{j, \ell}[U \, V]_{-j, \ell}^{\top}Z_{j,\ell}^{\top}U_j),\] we can write
 \[\frac{1}{L}\sum_{\ell=1}^{L}f_{j,\ell}\big[\zeta_{j, \ell, 1}U_j\big]=\left(\frac{1}{L}\sum_{\ell=1}^{L}O_{1,j,\ell}[U_{j} \, V_{j,\ell}]\bm w_{1,j,\ell},\dots, \frac{1}{L}\sum_{\ell=1}^{L}O_{r_{j, U},j,\ell}[U_{j} \, V_{j,\ell}]\bm w_{r_{j, U},j,\ell}\right).\]
 
 In the above matrix, each column $\frac{1}{L}\sum_{\ell=1}^{L}O_{k,j,\ell}[U_{j} V_{j,\ell}]\bm w_{k,j,\ell}$ is independent among different $k$ and has norm distribution $\calN\left(0, \Sigma_{j, k}\right)$, where 
 \[\Sigma_{j, k} =\frac{\sigma^2}{L^2}\sum_{\ell=1}^{L}O_{1,j,\ell}[U_j V_{j,\ell}]\Lambda_{j,\ell}^2[U_j V_{j,\ell}]^{\top}O_{1,j,\ell}^{\top}.\]
 Therefore, $\ve\left(\frac{1}{L}\sum_{\ell=1}^{L}f_{j,\ell}\big[\zeta_{j, \ell, 1}U_j\big]\right) \sim \calN(0, \Sigma_j)$, where $\Sigma_j=\mathrm{diag}(\Sigma_{j, 1}, \Sigma_{j, 2},\dots, \Sigma_{j, r_{j, U}})$. Since
 \[\norm{\Sigma_j}_2 \leq \max_k \norm{\Sigma_{j, k}}_2 \leq L^{-1}\sigma^2 \lambda^2_{\max}\lambda_{\min}^{-2}\Delta^{-2},\]
 we obtain that 
 \begin{equation} \label{eq:first-1-hetero}
 	\begin{aligned}
 		\norm{\frac{1}{L}\sum_{\ell=1}^{L}f_{j,\ell}\big[\zeta_{j, \ell, 1}U_j\big]U_j^{\top}}_{\rm F}
 		&=\norm{\frac{1}{L}\sum_{\ell=1}^{L}f_{j,\ell}\big[\zeta_{j, \ell, 1}U_j\big]}_{\rm F}\\
 		&=\norm{\ve\left(\frac{1}{L}\sum_{\ell=1}^{L}f_{j,\ell}\big[\zeta_{j, \ell, 1}U_j\big]\right)}_{2}\\
 		&\leq \norm{\Sigma_j^{1/2}}_2\norm{\Sigma_j^{-1/2}\ve\left(\frac{1}{L}\sum_{\ell=1}^{L}f_{j,\ell}\big[\zeta_{j, \ell, 1}U_j\big]\right)}_2\\
 		&\leq  \sqrt{2pr/L}\kappa_0\sigma\Delta^{-1},
 	\end{aligned}
 \end{equation}
 with probability at least $1-2e^{-pr/8}$, where we apply \eqref{eq:sum-Z2} to bound the norm of the vector $\Sigma_j^{-1/2}\ve\left(\frac{1}{L}\sum_{\ell=1}^{L}f_{j,\ell}\big[\zeta_{j, \ell, 1}U_j\big]\right) \sim \calN(0, I_{p_jr_{j, U}})$.
 
 A similar argument can be used to bound \[\frac{1}{L}\sum_{\ell=1}^{L}f_{j,\ell}\big[\zeta_{j, \ell, 1}^{\top}U_j\big]=\frac{1}{L}\sum_{\ell=1}^{L}f_{j, \ell}(Z_{j,\ell}[UV]_{-j, \ell}G_{j, \ell}^{\top}[U_j V_{j,\ell}]^{\top}U_j),\]
 since $Z_{j,\ell}[UV]_{-j, \ell}G_{j, \ell}^{\top}[U_j V_{j,\ell}]^{\top}U_j = Z_{j,\ell}[UV]_{-j, \ell}G_{j, \ell}^{\top} \begin{pmatrix}
 	I_{r_{j, U}}\\ \bm{0}_{r_{j, U}, r_{j,V,\ell}}
 \end{pmatrix}$, and $Z_{j,\ell}[UV]_{-j, \ell}G_{j, \ell}^{\top} $ has independent normal entries with variance at most $\lambda_{\max}^2\sigma^2$. Then repeating the previous procedure leads to 	
 \begin{equation}\label{eq:first-2-hetero}
 	\norm{\frac{1}{L}\sum_{\ell=1}^{L}f_{j,\ell}\big[\zeta_{j, \ell, 1}^{\top}U_j\big]U_j^{\top}}_{\rm F} \leq  \sqrt{2pr/L}\kappa_0\sigma\Delta^{-1},
 \end{equation}
  with probability at least $1-2e^{-pr/8}$.

Combing \eqref{eq:decomp-hetero}, \eqref{eq:second-1-hetero}, \eqref{eq:second-2-hetero}, \eqref{eq:first-1-hetero}, and \eqref{eq:first-2-hetero}, we obtain that 
\[\norm{\frac{1}{L}\sum_{\ell=1}^L\hU_{j, \ell}\hU_{j, \ell}-U_jU_j^{\top}}_{\rm F} \lesssim  \sqrt{\frac{pr}{L}}\frac{\kappa_0\sigma}{\Delta} +  \frac{\kappa_0^2pr\sigma^2}{\Delta^{2}},\]
with probability at least $1-C_1''Le^{c_1''p}$ for some constants $c_1''$ and $C_1''$.

In the following proof, the Davis-Kahan Theorem refers to the variant provided in \cite{yu2015useful}. By Davis-Kahan Theorem, it holds that 
\[\norm{\hU_j\hU_j^{\top}-U_jU_j^{\top}}_{\rm F} \leq 2\norm{\frac{1}{L}\sum_{\ell=1}^L\hU_{j, \ell}\hU_{j, \ell}-U_jU_j^{\top}}_{\rm F}.\]
Therefore, we conclude that 
\[
\norm{\hU_j\hU_j^{\top}-U_jU_j^{\top}}_{\rm F} \leq \widetilde C_2 \left(\sqrt{\frac{pr}{L}}\frac{\kappa_0\sigma}{\Delta}+\frac{pr\kappa_0^2\sigma^2}{\Delta^2}\right),
\]
with probability $1- C_1''Le^{-c_1''p}\geq 1-\widetilde C_1e^{-\widetilde c_1p}$, for some constants $\widetilde C_1, \widetilde c_1, \widetilde C_2$.

For $\hV_{j, \ell}$, note that its columns are the eigenvectors of 
\begin{align*}
	&\quad \left(I_{p_j}-\hU_j\hU_j^{\top}\right)\left(U_j\Lambda_{j, U}^2U_j^{\top} +  V_{j,\ell}\Lambda_{j, V_{\ell}}^2V_{j,\ell}^{\top} + \frakE_{j, \ell}\right)\left(I_{p_j}-\hU_j\hU_j^{\top}\right)\\
	&=\left(I_{p_j}-U_jU_j^{\top}+U_jU_j^{\top}-\hU_j\hU_j^{\top}\right)\left(U_j\Lambda_{j, U}^2U_j^{\top} +  V_{j,\ell}\Lambda_{j, V_{\ell}}^2V_{j,\ell}^{\top} + \frakE_{j, \ell}\right)\left(I_{p_j}-U_jU_j^{\top}+U_jU_j^{\top}-\hU_j\hU_j^{\top}\right)\\
	&=V_{j,\ell}\Lambda_{j, V_{\ell}}^2V_{j,\ell}^{\top} + V_{j,\ell}\Lambda_{j, V_{\ell}}^2V_{j,\ell}^{\top}\left(U_jU_j^{\top}-\hU_j\hU_j^{\top}\right)+\left(U_jU_j^{\top}-\hU_j\hU_j^{\top}\right)V_{j,\ell}\Lambda_{j, V_{\ell}}^2V_{j,\ell}^{\top}\\
	&\quad +  \left(U_jU_j^{\top}-\hU_j\hU_j^{\top}\right)\left(U_j\Lambda_{j, U}^2U_j^{\top} +  V_{j,\ell}\Lambda_{j, V_{\ell}}^2V_{j,\ell}^{\top}\right)\left(U_jU_j^{\top}-\hU_j\hU_j^{\top}\right)+ \left(I_{p_j}-\hU_j\hU_j^{\top}\right)\frakE_{j, \ell}\left(I_{p_j}-\hU_j\hU_j^{\top}\right).
\end{align*}
Since, with probability at least $1-\widetilde C_1e^{-\widetilde c_1p}$,  
\[
\sup_{j, \ell}\norm{V_{j,\ell}\Lambda_{j, V_{\ell}}^2V_{j,\ell}^{\top}\left(U_jU_j^{\top}-\hU_j\hU_j^{\top}\right)}_2 \leq \widetilde C_2\lambda_{\max}^2 \left(\sqrt{\frac{pr}{L}}\frac{\kappa_0\sigma}{\Delta}+\frac{pr\kappa_0^2\sigma^2}{\Delta^2}\right),
\]
\begin{align*}
	&\quad \sup_{j, \ell}\norm{\left(U_jU_j^{\top}-\hU_j\hU_j^{\top}\right)\left(U_j\Lambda_{j, U}^2U_j^{\top} +  V_{j,\ell}\Lambda_{j, V_{\ell}}^2V_{j,\ell}^{\top}\right)\left(U_jU_j^{\top}-\hU_j\hU_j^{\top}\right)}_2\\ 
	&\lesssim \frac{pr}{L}\frac{\lambda^2_{\max}\kappa_0^2\sigma^2}{\Delta^2}+\frac{p^2r^2\lambda_{\max}^2\kappa_0^4\sigma^4}{\Delta^4}\\
	&\ll \sqrt{\frac{pr}{L}}\frac{\lambda^2_{\max}\kappa_0\sigma}{\Delta}+\frac{pr\lambda_{\max}^2\kappa_0^2\sigma^2}{\Delta^2},
\end{align*}
and
\[
\sup_{j, \ell}\norm{\left(I_{p_j}-\hU_j\hU_j^{\top}\right)\frakE_{j, \ell}\left(I_{p_j}-\hU_j\hU_j^{\top}\right)}_2 \lesssim \sqrt{p}\lambda_{\max}\sigma.
\]
By Davis-Kahan Theorem, we obtain that 
\begin{align*}
	\sup_{j, \ell}\norm{\hV_{j,\ell} \hV_{j,\ell}^{\top}-V_{j,\ell} V_{j,\ell}^{\top}}_{\rm F} &\lesssim \frac{\sqrt{pr_{V}}\lambda_{\max}\sigma+\sqrt{prr_{V}/L}\lambda_{\max}^2\kappa_0\sigma/\Delta+\frac{\sqrt{r_{V}}pr\lambda_{\max}^2\kappa_0^2\sigma^2}{\Delta^2}}{\lambda_{\min}^2}\\
	&\lesssim \frac{\sqrt{pr_{V}}\kappa_0\sigma(1+\kappa\kappa_0\sqrt{r/L})}{\lambda_{\min}}+\frac{\sqrt{r_V}pr\kappa^2\kappa_0^2\sigma^2}{\lambda_{\min}^2},
\end{align*}
with probability at least $1-\widetilde C_1e^{-\widetilde c_1p}$.
\end{proof}

	\begin{proof}[Proof for Theorem 3.2]
	By the proof of Theorem 2.1, with probability at least $1-C_1'Le^{-c_1'p} \geq 1-\widetilde{C}_1e^{-\widetilde c_1p}$, we have that 
	\begin{equation}\label{eq:rank-W-error}
		\norm{\overline{W}_{j}-U_jU_j^{\top}-\frac{1}{L}\sum_{\ell=1}^{L}V_{j, \ell}V_{j, \ell}^{\top}}_{\rm F} \leq 2C_2'\lambda_{\min}^{-1}\sigma\sqrt{pr}.
	\end{equation}
	Let $\widecheck{V}_{j}\Sigma_{j, V}\widecheck{V}_{j}^{\top}$ be the eigen-decomposition of $\overline{V}_{j}:= \frac{1}{L}\sum_{\ell=1}^{L}V_{j, \ell}V_{j, \ell}^{\top}$. Since $U_{j}^{\top}\overline{V}_{j}=0$, we have $U_{j}^{\top}\widecheck{V}_{j}=0$. We will show that 
	\[\lambda_{\max}(\overline{V}_{j})=\lambda_{\max}(\Sigma_{j, V}) \leq 1- \frac{\eta_{V}}{3}.\]
	
	First, rewrite $\overline{V}_{j}=\frac{1}{L}\widetilde{V}_{j}\widetilde{V}_{j}^{\top}$, where $\widetilde{V}_{j}:=[V_{j,1} \, V_{j,2} \dots V_{j, L}] \in \R^{p_j \times \widetilde{r}_{j, V}}$ and $\widetilde{r}_{j, V}:=\sum_{\ell=1}^{L}r_{j,V,\ell}$.
	Note that 
	\[I_{\widetilde{r}_{j, V}}-\frac{1}{L}\widetilde{V}_{j}^{\top}\widetilde{V}_{j}=\begin{pmatrix}
		\frac{L-1}{L}I_{r_{j, V, 1}}& -\frac{1}{L}V_{j,1}^{\top}V_{j, 2} & \cdots & -\frac{1}{L}V_{j,1}^{\top}V_{j, L}\\
	-\frac{1}{L}V_{j,2}^{\top}V_{j, 1}	&\frac{L-1}{L}I_{r_{j, V, 2}} & \cdots & -\frac{1}{L}V_{j,2}^{\top}V_{j, L}\\
	&\cdots&\cdots&\\
		-\frac{1}{L}V_{j,L}^{\top}V_{j, 1}	&	-\frac{1}{L}V_{j,L}^{\top}V_{j, 2}& \cdots & \frac{L-1}{L}I_{r_{j, V, L}} \\
	\end{pmatrix}.\]
Assume that $\lambda$ is an eigenvalue of $I_{\widetilde{r}_{j, V}}-\frac{1}{L}\widetilde{V}_{j}^{\top}\widetilde{V}_{j}$, and corresponding eigenvector is $\bm v$. Write $\bm v = (\bm v_1^{\top}, \bm v_2^{\top}, \dots, \bm v_L^{\top})^{\top}$, where $\bm v_\ell \in \R^{r_{j, V, \ell}}$. Let $\widetilde\ell=\argmax_{\ell} \norm{\bm v_{\ell}}_2$. Since $(I_{\widetilde{r}_{j, V}}-\frac{1}{L}\widetilde{V}_{j}^{\top}\widetilde{V}_{j})\bm v=\lambda\bm v$, we have
\[\norm{\left(\lambda-\frac{L-1}{L}\right)\bm v_{\widetilde\ell}}_2 =\frac{1}{L}\norm{ \sum_{\ell' \neq \widetilde\ell} V_{j, \widetilde\ell}^{\top}V_{j, \ell'}\bm 	v_{\ell'}}_2 \leq  \frac{1}{L}\sum_{\ell' \neq \widetilde\ell} \norm{V_{j, \widetilde\ell}^{\top}V_{j, \ell'}}_2 \norm{\bm 	v_{\ell'}}_2 \leq  \frac{(L-1)(1-\eta_{V})}{L}\norm{\bm 	v_{\widetilde\ell}}_2,\]
where the last inequality follows from the definition that $\eta_{V}=\min_{j, \ell}\sum_{\ell'\neq \ell}(1-\norm{V_{j, \ell}^{\top} V_{j, \ell'}}_2)/ (L-1)$ and thus $\sum_{\ell'\neq \widetilde\ell}\norm{V_{j, \widetilde\ell}^{\top} V_{j, \ell'} }_2\leq (L-1)(1-\eta_{V})$. Therefore, 
we obtain that $\lambda \geq \frac{\eta_V(L-1)}{L}$, and this is true for all of the eigenvalues of $I_{\widetilde{r}_{j, V}}-\frac{1}{L}\widetilde{V}_{j}^{\top}\widetilde{V}_{j}$.
	
By Woodbury matrix identity, it holds that
	\[(I_{p_j}-\overline{V}_{j})^{-1}=I_{p_j}+\frac{1}{L}\widetilde{V}_{j}\left(I_{\widetilde{r}_{j, V}}-\frac{1}{L}\widetilde{V}_{j}^{\top}\widetilde{V}_{j}\right)^{-1}\widetilde{V}_{j}^{\top}.\]
Since $\norm{\widetilde{V}_{j}}_2 \leq \sqrt{\sum_{\ell=1}^{L}\norm{V_{j, \ell}}_2^2}\leq\sqrt{L}$, we obtain
	\[\norm{(I_{p_j}-\overline{V}_{j})^{-1}}_2 \leq 1+\frac{L}{\eta_V(L-1)},\]
	yielding that 
	\[\norm{\overline{V}_{j}}_2 =1- \lambda_{\min}(I_{p_j}-\overline{V}_{j}) = 1-\norm{(I_{p_j}-\overline{V}_{j})^{-1}}_2^{-1} \leq 1-\frac{\eta_{V}}{3}.\]
	
As a result, we establish that $U_jU_j^{\top}+\frac{1}{L}\sum_{\ell=1}^{L}V_{j, \ell}V_{j, \ell}^{\top}=[U_j \, \widecheck{V}_{j}]\mathrm{diag}(I_{r_{j, U}}, \Sigma_{j, V})[U_j \, \widecheck{V}_{j}]^{\top}$. By Weyl's theorem and \eqref{eq:rank-W-error}, we obtain 
\[\sigma_{k}(\overline{W}) \geq 1-2C_2'\frac{\sigma\sqrt{pr}}{\lambda_{\min}}, \forall k \leq r_{j, U},\]
and 
\[\sigma_{k}(\overline{W}) \leq 1-\frac{\eta_{V}}{3}+2C_2'\frac{\sigma\sqrt{pr}}{\lambda_{\min}}, \forall k  \geq r_{j, U}+1.\]
Therefore, since $\frac{\sigma\sqrt{pr}}{\lambda_{\min}}=o(1)$and $\eta_V\geq \eta_0$,  for any $\delta_0 \in (0, \eta_{0}/3)$, we have $\Prob(\widehat{r}_{j, U}=r_{j, U}) \geq 1-\widetilde C_1e^{-\widetilde c_1p}$.
	
\end{proof}

\subsection{Proof of the Results for the Knowledge Transfer Setting} \label{sec:supp_proof_transfer}

\begin{proof}[Proof of Theorems 4.1 and \ref{thm:transfer-L}]
	It suffices to show Theorem \ref{thm:transfer-L} since Theorem 4.1 is a special case with $L=2$. Analogous to \eqref{eq:decomp-hetero}, \eqref{eq:second-1-hetero}, and  \eqref{eq:second-2-hetero}, we have that
	\begin{equation}\label{eq:decomp-weighted}
		\begin{aligned}
				&\sup_{j,\ell} \frac{1}{\sigma_{\ell}^2}\norm{\hU_{j, \ell}\hU_{j, \ell}^{\top}-U_jU_j^{\top}-f_{j,\ell}\big[\zeta_{j, \ell, 1}U_j\big]U_j^{\top}-f_{j,\ell}\big[\zeta_{j, \ell, 1}^{\top}U_j\big]U_j^{\top}-U_jf_{j,\ell}\big[\zeta_{j, \ell, 1}U_j\big]^{\top}-U_jf_{j,\ell}\big[\zeta_{j, \ell, 1}^{\top}U_j\big]^{\top}}_{\rm F} \\
				&\lesssim \kappa_0^2pr\Delta^{-2},
		\end{aligned}
	\end{equation} 
	with probability at least $1-C_1'Le^{-c_1'p}$. 
	
	Now we bound $\sum_{\ell=1}^{L}w_{\ell}f_{j,\ell}\big[\zeta_{j, \ell, 1}U_j\big]U_j^{\top}$. The other three terms can be bounded similarly.  Since $\ve(G_{j, \ell}[UV]_{-j, \ell}^{\top}Z_{j,\ell}^{\top}U_j)=(U_j^{\top} \otimes G_{j, \ell}[UV]_{-j, \ell}^{\top})\ve(Z_{j,\ell}^{\top})$, we have that \[\ve(G_{j, \ell}[UV]_{-j, \ell}^{\top}Z_{j,\ell}^{\top}U_j) \sim \calN(0, \sigma_{\ell}^2(I_{r_{j, U}} \otimes \Lambda_{j, \ell}^2)),\]
	which shows that $G_{j, \ell}[UV]_{-j, \ell}^{\top}Z_{j,\ell}^{\top}U_j$ has independent normal entries with variance at most $\lambda_{\max}^2\sigma_{\ell}^2$. Let 
	$G_{j, \ell}[UV]_{-j, \ell}^{\top}Z_{j,\ell}^{\top}U_j=(\bm w_{1,j,\ell},\dots,\bm w_{r_{j, U}, j, \ell})$. Since \[\sum_{\ell=1}^{L}w_{\ell}f_{j,\ell}\big[\zeta_{j, \ell, 1}U_j\big]=\sum_{\ell=1}^{L}w_{\ell}f_{j, \ell}([U_j V_{j,\ell}]G_{j, \ell}[UV]_{-j, \ell}^{\top}Z_{j,\ell}^{\top}U_j),\] we can write
	\[\sum_{\ell=1}^{L}w_{\ell}f_{j,\ell}\big[\zeta_{j, \ell, 1}U_j\big]=\left(\sum_{\ell=1}^{L}w_{\ell}O_{1,j,\ell}[U_{j} V_{j,\ell}]\bm w_{1,j,\ell},\dots, \sum_{\ell=1}^{L}w_{\ell}O_{r_{j, U},j,\ell}[U_{j} V_{j,\ell}]\bm w_{r_{j, U},j,\ell}\right).\]
	
	In the above matrix, each column $\sum_{\ell=1}^{L}w_{\ell}O_{k,j,\ell}[U_{j} V_{j,\ell}]\bm w_{k,j,\ell}$ is independent among different $k$ and has norm distribution $\calN\left(0, \Sigma_{j, k}\right)$, where 
	\[\Sigma_{j, k} =\sum_{\ell=1}^{L}w_{\ell}^2\sigma_{\ell}^2O_{1,j,\ell}[U_j V_{j,\ell}]\Lambda_{j,\ell}^2[U_j V_{j,\ell}]^{\top}O_{1,j,\ell}^{\top}.\]
	Therefore, $\ve\left(\sum_{\ell=1}^{L}w_{\ell}f_{j,\ell}\big[\zeta_{j, \ell, 1}U_j\big]\right) \sim \calN(0, \Sigma_j)$, where $\Sigma_j=\mathrm{diag}(\Sigma_{j, 1}, \Sigma_{j, 2},\dots, \Sigma_{j, r_{j, U}})$. Since
	\[\norm{\Sigma_j}_2 \leq \max_k \norm{\Sigma_{j, k}}_2 \leq \sum_{\ell=1}^{L}w^2_{\ell}\sigma_{\ell}^2 \lambda^2_{\max}\lambda_{\min}^{-2}\Delta^{-2},\]
	we obtain that 
	\begin{equation} \label{eq:first-1-weighted}
		\begin{aligned}
			\norm{\sum_{\ell=1}^{L}w_{\ell}f_{j,\ell}\big[\zeta_{j, \ell, 1}U_j\big]U_j^{\top}}_{\rm F}
			&=\norm{\sum_{\ell=1}^{L}w_{\ell}f_{j,\ell}\big[\zeta_{j, \ell, 1}U_j\big]}_{\rm F}\\
			&=\norm{\ve\left(\sum_{\ell=1}^{L}w_{\ell}f_{j,\ell}\big[\zeta_{j, \ell, 1}U_j\big]\right)}_{2}\\
			&\leq \norm{\Sigma_j^{1/2}}_2\norm{\Sigma_j^{-1/2}\ve\left(\sum_{\ell=1}^{L}w_{\ell}f_{j,\ell}\big[\zeta_{j, \ell, 1}U_j\big]\right)}_2\\
			&\leq  \sqrt{2pr}\kappa_0\sqrt{\sum_{\ell=1}^{L}w_{\ell}^2\sigma_{\ell}^2}\Delta^{-1},
		\end{aligned}
	\end{equation}
	with probability at least $1-2e^{-pr/8}$, where we apply \eqref{eq:sum-Z2} to bound the norm of the vector $\Sigma_j^{-1/2}\ve\left(\sum_{\ell=1}^{L}w_{\ell}f_{j,\ell}\big[\zeta_{j, \ell, 1}U_j\big]\right) \sim \calN(0, I_{p_jr_{j, U}})$.

Therefore, by Davis-Kahan Theorem, we obtain
\begin{align*}
	&\quad \sup_j\norm{\hU_{j}\hU_{j}^{\top}-U_jU_j^{\top}}_{\rm F}\\
	& \leq 2\norm{\sum_{\ell=1}^Lw_{\ell}\hU_{j,\ell}\hU_{j,\ell}^{\top}-U_jU_j^{\top}}_{\rm F}\\
	&\leq \widetilde C_2\left(\frac{\sqrt{pr}\kappa_0\sqrt{\sum_{\ell=1}^{L}w_{\ell}^2\sigma_{\ell}^2}}{\Delta}+\frac{pr\kappa_0^2\sum_{\ell=1}^{L}w_{\ell}\sigma_{\ell}^2}{\Delta^2}\right),
\end{align*}
with probability at least $1-\widetilde{C}_1e^{-\widetilde c_1p}$. 

Using the same argument in the proof of Theorem 3.1, for any $\ell$, we have
\begin{align*}
	&\quad \sup_{j}\norm{\hV_{j,\ell} \hV_{j,\ell}^{\top}-V_{j,\ell} V_{j,\ell}^{\top}}_{\rm F} \\
	&\lesssim \frac{\sqrt{pr_{V}}\kappa_0}{\lambda_{\min}}\left(\sigma_{\ell}+\kappa\kappa_0\sqrt{r}\sqrt{\sum_{\ell=1}^{L}w_{\ell}^2\sigma_{\ell}^2}\right)+\frac{\sqrt{r_V}pr\kappa^2\kappa_0^2\sum_{\ell=1}^{L}w_{\ell}\sigma_{\ell}^2}{\lambda_{\min}^2}.
\end{align*}
with probability at least $1-\widetilde{C}_1e^{-\widetilde c_1p}$. 

Since $\kappa_0, \kappa=O(1)$, 
$\Delta / \max_{\ell}\sigma_{\ell} \gg \sqrt{prL}$ 
and \[\sum_{\ell=1}^{L}w_{\ell}\sigma_{\ell}^2 \leq \sqrt{L}\max_{\ell}\sigma_{\ell} \sqrt{\sum_{\ell=1}^{L}w_{\ell}^2\sigma_{\ell}^2},\]
we have 
\[\frac{pr\kappa_0^2\sum_{\ell=1}^{L}w_{\ell}\sigma_{\ell}^2}{\Delta^2} \lesssim  \frac{\sqrt{pr}\kappa_0\sqrt{\sum_{\ell=1}^{L}w_{\ell}^2\sigma_{\ell}^2}}{\Delta},\]
and thus, with probability at least $1-\widetilde{C}_1e^{-\widetilde c_1p}$,
\[\sup_j\norm{\hU_{j}\hU_{j}^{\top}-U_jU_j^{\top}}_{\rm F}\lesssim \frac{\sqrt{pr}\kappa_0\sqrt{\sum_{\ell=1}^{L}w_{\ell}^2\sigma_{\ell}^2}}{\Delta}.\]
Similarly, it holds that
\[\frac{\sqrt{r_V}pr\kappa^2\kappa_0^2\sum_{\ell=1}^{L}w_{\ell}\sigma_{\ell}^2}{\lambda_{\min}^2}\lesssim \frac{\sqrt{pr_{V}}\kappa_0}{\lambda_{\min}}\kappa\kappa_0\sqrt{r}\sqrt{\sum_{\ell=1}^{L}w_{\ell}^2\sigma_{\ell}^2},\]
which yields that, with probability at least $1-\widetilde{C}_1e^{-\widetilde c_1p}$,
\[
\sup_{j}\norm{\hV_{j,\ell} \hV_{j,\ell}^{\top}-V_{j,\ell} V_{j,\ell}^{\top}}_{\rm F} 
\lesssim \frac{\sqrt{pr_{V}}\kappa_0}{\lambda_{\min}}\Big(\sigma_{\ell}+\kappa\kappa_0\sqrt{r}\sqrt{\sum_{\ell=1}^{L}w_{\ell}^2\sigma_{\ell}^2}\Big).
\]

Plugging in $w_{\ell}=\frac{\sigma_{\ell}^{-2}}{\sum_{\ell=1}^{L}\sigma_{\ell}^{-2}}$ leads to 
\[\sup_j\norm{\hU_{j}\hU_{j}^{\top}-U_jU_j^{\top}}_{\rm F}\lesssim \frac{\sqrt{pr}\kappa_0}{\Delta\sqrt{\sum_{\ell=1}^L\sigma_{\ell}^{-2}}},\]
and
\[
\sup_{j}\norm{\hV_{j,\ell} \hV_{j,\ell}^{\top}-V_{j,\ell} V_{j,\ell}^{\top}}_{\rm F} \lesssim \frac{\sqrt{pr_{V}}\kappa_0}{\lambda_{\min}}\Big(\sigma_{\ell}+\frac{\kappa\kappa_0\sqrt{r}}{\sqrt{\sum_{\ell=1}^L\sigma_{\ell}^{-2}}}\Big).
\]

\end{proof}
\begin{proof}[Proof of Theorem \ref{thm:transfer-adaptive}]
	With adaptive weights $\widehat{w}_{\ell}$, we have, by Davis-Khan Theorem,
	\begin{align*}
		&\quad \sup_j\norm{\hU_{j}\hU_{j}^{\top}-U_jU_j^{\top}}_{\rm F}\\
		& \leq  2\sup_j\norm{\sum_{\ell=1}^L\widehat{w}_{\ell}\hU_{j,\ell}\hU_{j,\ell}^{\top}-U_jU_j^{\top}}_{\rm F}\\
		&\leq 2 \sup_j\sum_{\ell=1}^L\abs{\widehat{w}_{\ell}-w_{\ell}^*}\norm{\hU_{j,\ell}\hU_{j,\ell}^{\top}-U_jU_j^{\top}}_{\rm F}\\
		&\quad+2 \sup_j\norm{\sum_{\ell=1}^Lw_{\ell}^*\hU_{j,\ell}\hU_{j,\ell}^{\top}-U_jU_j^{\top}}_{\rm F}.
	\end{align*}
	By the proof of Theorem \ref{thm:transfer-L}, with probability at least $1-\widetilde{C}_1e^{-\widetilde c_1p}$, 
	\[\sup_j\norm{\sum_{\ell=1}^{L}w_{\ell}^*\hU_{j,\ell}\hU_{j,\ell}^{\top}-U_jU_j^{\top}}_{\rm F}\lesssim \frac{\sqrt{pr}\kappa_0}{\Delta\sqrt{\sum_{\ell=1}^L\sigma_{\ell}^{-2}}}.\]
Also, by \eqref{eq:frakE_1_hetero} and Davis-Khan Theorem, with probability at least $1-C_1'Le^{-c_1'p}$,
	\[\sup_{j,\ell}\frac{1}{\sigma_{\ell}}\norm{\hU_{j,\ell}\hU_{j,\ell}^{\top}-U_jU_j^{\top}}_{\rm F} \lesssim \frac{\sqrt{pr}\kappa_0}{\Delta}.\]
	
	Now we bound $\abs{\widehat{w}_{\ell}-w_{\ell}^*}$. Recall that
	\[   \widehat\sigma_{\ell} = \Big\|\calT_{\ell}-\calT_{\ell} \times_1 \big[\hU_{1,\ell}\; \hV_{1,\ell}\big]\big[\hU_{1,\ell}\; \hV_{1,\ell}\big]^{\top} \times_2\cdots\times_J\big[\hU_{J,\ell}\; \hV_{J,\ell}\big]\big[\hU_{J,\ell}\; \hV_{J,\ell}\big]^{\top}\Big\|_{\rm F} / \sqrt{p_1p_2\cdots p_J}.\]
	By the proof of Lemma 1 in \cite{xia2022inference}, we have that,
	\[\abs{\widehat\sigma_{\ell}\sqrt{p_1p_2\cdots p_J}-\norm{\calZ_{\ell}}_{\rm F}} \leq C'_2 \sqrt{pr}\kappa_0\sigma_{\ell}, \quad \forall \ell,\]
	with probability $1-C_1'Le^{-c_1'p}$.
	
	By Lemma 1 of \cite{laurent2000adaptive}, it holds that
	\[\Prob\left(\abs{\frac{\norm{\calZ_\ell}_{\rm F}^2}{\sigma_{\ell}^2}-p_1p_2\cdots p_J} \geq 2\sqrt{p_1p_2\cdots p_J}\sqrt{c_3+1}\sqrt{\log p}+2(c_3+1)\log p\right) \leq 2p^{-c_3-1},\]
	and thus, since $L=O(p^{c_3})$,
	\[\sup_{\ell}\abs{\frac{\norm{\calZ_\ell}_{\rm F}^2}{\sigma_{\ell}^2}-p_1p_2\cdots p_J} \lesssim p^{J/2}\sqrt{\log p}, \]
	with probability at least $1-\widetilde{C}_1p^{-1}$ for some constant $\widetilde{C}_1$. Hence,
	\[\abs{\norm{\calZ_\ell}_{\rm F}-\sigma_{\ell}\sqrt{p_1p_2\cdots p_J}}\leq \frac{\abs{\norm{\calZ_\ell}_{\rm F}^2-\sigma_{\ell}^2p_1p_2\cdots p_J}}{\norm{\calZ_\ell}_{\rm F}+\sigma_{\ell}\sqrt{p_1p_2\cdots p_J}}\lesssim \sigma_{\ell}\sqrt{\log p},\]
	and thus
	\[\abs{\widehat{\sigma}_{\ell}-\sigma_{\ell}} \leq \widetilde{C}_2 p^{-(J-1)/2}\sqrt{r}\kappa_0\sigma_{\ell},\]
	with probability at least $1-\widetilde {C}_1p^{-1}$ for some constants $\widetilde {C}_1$ and $\widetilde {C}_2$. This further leads to 
	\[\frac{(1-\widetilde{C}_2p^{-(J-1)/2}\kappa_0\sqrt{r})^2}{(1+\widetilde{C}_2p^{-(J-1)/2}\kappa_0\sqrt{r})^2}\frac{\sigma_{\ell}^{-2}}{\sum_{\ell=1}^{L}\sigma_{\ell}^{-2}}\leq \widehat{w}_{\ell}=\frac{\widehat{\sigma}_{\ell}^{-2}}{\sum_{\ell=1}^{L}\widehat{\sigma}_{\ell}^{-2}} \leq \frac{(1+\widetilde{C}_2p^{-(J-1)/2}\kappa_0\sqrt{r})^2}{(1-\widetilde{C}_2p^{-(J-1)/2}\kappa_0\sqrt{r})^2}\frac{\sigma_{\ell}^{-2}}{\sum_{\ell=1}^{L}\sigma_{\ell}^{-2}},\]
	and thus
	\[\abs{\widehat{w}_{\ell}-w^*_{\ell}} \leq C_Wp^{-(J-1)/2}\kappa_0\sqrt{r}\frac{\sigma_{\ell}^{-2}}{\sum_{\ell=1}^{L}\sigma_{\ell}^{-2}}, \forall \ell, \]
	for some constant $C_W$, with probability at least $1-\widetilde {C}_1p^{-1}$. Therefore,
	\[\sup_j\sum_{\ell=1}^L\abs{\widehat{w}_{\ell}-w_{\ell}^*}\norm{\hU_{j,\ell}\hU_{j,\ell}^{\top}-U_jU_j^{\top}}_{\rm F} \lesssim p^{-(J-1)/2}\kappa_0\sqrt{r}\frac{\sum_{\ell=1}^{L}\sigma_{\ell}^{-1}}{\sum_{\ell=1}^{L}\sigma_{\ell}^{-2}} \frac{\sqrt{pr}\kappa_0}{\Delta} \leq \sqrt{\frac{rL}{p^{J-1}}} \frac{\sqrt{pr}\kappa_0}{\Delta\sqrt{\sum_{\ell=1}^{L}\sigma_{\ell}^{-2}}},\]
	 with probability at least $1-\widetilde {C}_1p^{-1}$. Then we obtain that, with probability at least $1-\widetilde {C}_1p^{-1}$,
	 \[\sup_j\norm{\hU_{j}\hU_{j}^{\top}-U_jU_j^{\top}}_{\rm F}\lesssim \frac{\sqrt{pr}\kappa_0}{\Delta\sqrt{\sum_{\ell=1}^L\sigma_{\ell}^{-2}}}+\frac{r\kappa_0\sqrt{L}}{\Delta\sqrt{p^{J-2}\sum_{\ell=1}^L\sigma_{\ell}^{-2}}},\]
	 and
	 \[
	 \sup_{j}\norm{\hV_{j,\ell} \hV_{j,\ell}^{\top}-V_{j,\ell} V_{j,\ell}^{\top}}_{\rm F} \lesssim \frac{\sqrt{pr_{V}}\kappa_0}{\lambda_{\min}}\Big(\sigma_{\ell}+\frac{\kappa\kappa_0\sqrt{r}}{\sqrt{\sum_{\ell=1}^L\sigma_{\ell}^{-2}}}+\frac{\kappa\kappa_0r\sqrt{L}}{\sqrt{p^{J-1}\sum_{\ell=1}^L\sigma_{\ell}^{-2}}}\Big).
	 \]
	\end{proof}

\subsection{Proof of Technical Lemmas}\label{sec:supp-theory-lemmas}
\begin{proof}[Proof of Lemma \ref{lem:1}] 
We first show the high-probability bound for a fix $(j, \ell)$. It is assumed in Theorem 2.1 that there exist $C_1, c_1, C_2$ such that
\begin{equation}\label{eq:good-1}
	\Prob\left[\norm{\hU_{j,\ell}^{(0)}\hU_{j,\ell}^{(0)\top}-U_{j}U_{j}^{\top}}_2 \leq C_2 \sqrt{p}\sigma\lambda_{\min}^{-1}\right] \geq 1-C_1e^{-c_1p}.
\end{equation} 

Since $Z_{j, \ell}$ has i.i.d. $\calN(0, \sigma^2)$ entries and $U_{-j}^{\top}U_{-j}=I_{r_{-j}}$, the matrix  $Z_{j,\ell}U_{-j}$ also has i.i.d. $\calN(0, \sigma^2)$ entries. Then we need the following lemma:
\begin{lemma}\label{lem:3}[Theorem 5.39 in \cite{vershynin2010introduction}]
	Let $M \in \R^{p_1 \times p_2}$ whose rows $M_1, \dots ,M_{p_1}$ are  independent sub-Gaussian random vectors with $\E[M_iM_i^{\top}] = I_{p_2}$. Then for every $t \geq 0$, there exist constants $c, C>0$ such that 
	\[\Prob\left[\norm{M}_2 \leq \sqrt{p_1}+C\sqrt{p_2}+t\right]\geq 1-2e^{-ct^2},\]
	for any $t\geq 0$. 
\end{lemma}
By Lemma \ref{lem:3}, for any $t>0$, there exist $c, C$ such that \[\norm{\frac{1}{\sigma}Z_{j, \ell}U_{-j}}_2 \leq \sqrt{p_j}+C\sqrt{r_{-j}}+t, \] 
with probability at least $1-2e^{-ct^2}$. Using the assumptions $p \asymp p_j$ and $p \gtrsim r^{J-1} \geq r_{-j}$ and choosing $t = \sqrt{p}$, we obtain that 
\begin{equation}\label{eq:good-2}
	\Prob\left[\norm{Z_{j, \ell}U_{-j}}_2 \leq C_{Z}\sigma\sqrt{p}\right] \geq 1-2e^{-c_{Z}p},
\end{equation} 
for some $c_{Z}, C_{Z}>0$. The same argument leads to 
\begin{equation}\label{eq:good-2-ave}
	\Prob\left[\norm{\overline{Z}_{j}U_{-j}}_2 \leq C_{Z}\sigma\sqrt{p/L}\right] \geq 1-2e^{-c_{Z}p},
\end{equation} 
since $\overline{Z}_{j}=\frac{1}{L}\sum_{\ell=1}^{L}Z_{j,\ell} \in \calN(0, \sigma^2/L)$.

To show the third inequality in \eqref{eq:good}, we first define the ball  $\mathcal{B}^{p \times r}(X_0, \varepsilon):=\{X \,|\, X \in \R^{p\times r}, \norm{X-X_0}_2 \leq \varepsilon\}$. By Lemma 7 in \cite{zhang2018tensor}, for any $\varepsilon>0$, there exist $\mathcal{C}^{p, r}:=\big\{\overline{X}_{1},\dots,\overline{X}_{N}\big\}$ such that $\norm{\overline{X}_i}_2 \leq 1$ and $\mathcal{B}^{p \times r}(\bm{0}, 1) \subset \bigcup_{i=1}^N \mathcal{B}^{p \times r}\big(\overline{X}_i, \varepsilon\big)$, with $N \leq (1+2/\varepsilon)^{pr}$. In particular, let $\varepsilon=1/(2J)$, and  we have that, for any $(X_1, \dots, X_J)$ satisfying $X_j \in \mathcal{B}^{p_j \times r_j}(\bm{0}, 1)$, there exists $\big(\widetilde{X}_{1}, \dots, \widetilde{X}_{J}\big) \in \mathcal{C}^{p_1, r_1} \times \mathcal{C}^{p_2, r_2} \times \dots \times \mathcal{C}^{p_J, r_J}$ such that $\norm{X_j-\widetilde{X}_j}_2 \leq 1/(2J)$ for all $j$. Since  
\begin{align*}
	&\norm{Z_{j,\ell}(X_1 \otimes X_2 \otimes \cdots \otimes X_{j-1} \otimes X_{j+1}\otimes \cdots \otimes X_J)}_2 -	\norm{Z_{j,\ell}(\widetilde{X}_1 \otimes \widetilde{X}_2 \otimes \cdots \otimes \widetilde{X}_{j-1} \otimes \widetilde{X}_{j+1}\otimes \cdots \otimes \widetilde{X}_J)}_2\\
	&\leq \sum_{j'\in [J]\setminus \{j\}}\norm{Z_{j,\ell}\big(\widetilde{X}_1 \otimes \widetilde{X}_2 \otimes \cdots \otimes(X_{j'} -\widetilde{X}_{j'} )\otimes\cdots \otimes X_{j-1} \otimes X_{j+1}\otimes \cdots \otimes X_J\big)}_2\\
	&\leq \sum_{j'\in [J]\setminus \{j\}}\norm{X_{j'} -\widetilde{X}_{j'} }_2\sup_{\substack{X_{j'} \in \R^{p_{j'} \times r_{j'}}\\ \norm{X_{j'}}_2 \leq 1, \forall j' \in [J] \setminus\{j\}}} \norm{Z_{j,\ell}\big(X_1 \otimes X_2 \otimes \cdots \otimes X_{j-1} \otimes X_{j+1}\otimes \cdots \otimes X_J\big)}_2\\
	&\leq \frac{1}{2} \sup_{\substack{X_{j'} \in \R^{p_{j'} \times r_{j'}}\\ \norm{X_{j'}}_2 \leq 1, \forall j' \in [J] \setminus\{j\}}} \norm{Z_{j,\ell}\big(X_1 \otimes X_2 \otimes \cdots \otimes X_{j-1} \otimes X_{j+1}\otimes \cdots \otimes X_J\big)}_2,
\end{align*}
which implies
\begin{equation}\label{eq:cover}
	\begin{aligned}
		&\sup_{\substack{X_{j'} \in \R^{p_{j'} \times r_{j'}}\\ \norm{X_{j'}}_2 \leq 1, \forall j' \in [J] \setminus\{j\}}} \norm{Z_{j,\ell}\big(X_1 \otimes X_2 \otimes \cdots \otimes X_{j-1} \otimes X_{j+1}\otimes \cdots \otimes X_J\big)}_2 \\
		&\leq 2\sup_{\big(\widetilde{X}_{1}, \dots, \widetilde{X}_{J}\big) \in  \mathcal{C}^{p_1, r_1}  \times \dots \times \mathcal{C}^{p_J, r_J}} \norm{Z_{j,\ell}\big(\widetilde{X}_1 \otimes \widetilde{X}_2 \otimes \cdots \otimes \widetilde{X}_{j-1} \otimes \widetilde{X}_{j+1}\otimes \cdots \otimes \widetilde{X}_J\big)}_2.
	\end{aligned}
\end{equation}

Therefore, it suffices to show that there exist $\widetilde C_1, \widetilde c_1, \widetilde C_2$, s.t.
\[\sup_{\big(\widetilde{X}_{1}, \dots, \widetilde{X}_{J}\big) \in  \mathcal{C}^{p_1, r_1} \times  \dots \times \mathcal{C}^{p_J, r_J}} \norm{Z_{j,\ell}\big(\widetilde{X}_1 \otimes \widetilde{X}_2 \otimes \cdots \otimes \widetilde{X}_{j-1} \otimes \widetilde{X}_{j+1}\otimes \cdots \otimes \widetilde{X}_J\big)}_2 \leq \widetilde C_2\sigma\sqrt{pr},\]
with probability at least $1-\widetilde C_1e^{-\widetilde c_1p}$. 

For a fixed tuple $\big(\widetilde{X}_{1}, \dots, \widetilde{X}_{J}\big)$, let $\Psi_{j, \ell} = Z_{j,\ell}\big(\widetilde{X}_1 \otimes \widetilde{X}_2 \otimes \cdots \otimes \widetilde{X}_{j-1} \otimes \widetilde{X}_{j+1}\otimes \cdots \otimes \widetilde{X}_J\big)$. The rows of $\Psi_{j, \ell}$ are independent normal vectors $\calN(0, \Sigma_{-j})$, where $\Sigma_{-j} =  \sigma^2\big(\widetilde{X}_1^{\top}\widetilde{X}_1 \otimes \widetilde{X}_2^{\top}\widetilde{X}_2 \otimes \cdots \otimes \widetilde{X}_{j-1}^{\top}\widetilde{X}_{j-1} \otimes \widetilde{X}_{j+1}^{\top}\widetilde{X}_{j+1}\otimes \cdots \otimes \widetilde{X}_J^{\top}\widetilde{X}_{J}\big)$, and  hence $\Psi_{j, \ell}\Sigma_{-j}^{-1/2}$ has i.i.d. standard normal entries. Applying Lemma \ref{lem:3} to $\Psi_{j, \ell}\Sigma_{-j}^{-1/2}$ and using that $\norm{\Psi_{j, \ell}}_2 \leq \norm{\Psi_{j, \ell}\Sigma_{-j}^{-1/2}}_2\norm{\Sigma_{-j}^{-1/2}}_2 \leq \sigma\norm{\Psi_{j, \ell}\Sigma_{-j}^{-1/2}}_2$ leads to 
\[\norm{Z_{j,\ell}\big(\widetilde{X}_1 \otimes \widetilde{X}_2 \otimes \cdots \otimes \widetilde{X}_{j-1} \otimes \widetilde{X}_{j+1}\otimes \cdots \otimes \widetilde{X}_J\big)}_2 \leq \sigma(\sqrt{p_{j}}+C\sqrt{r_{-j}}+t),\]
with probability at least $1-2e^{-ct^2}$, for some absolute constants $c, C>0$. Using the assumptions $p \asymp p_j$ and $p \gtrsim r^{J-1} \geq r_{-j}$ and letting $t=\gamma'\sqrt{pr}$, we obtain 
\[	\norm{Z_{j,\ell}\big(\widetilde{X}_1 \otimes \widetilde{X}_2 \otimes \cdots \otimes \widetilde{X}_{j-1} \otimes \widetilde{X}_{j+1}\otimes \cdots \otimes \widetilde{X}_J\big)}_2 \leq C_{\gamma}'\sigma\sqrt{pr},\]
with probability at least $1-2e^{-c\gamma'^2pr}$, for any $\gamma'>0$ and some $C_{\gamma}'>0$ depending on $\gamma'$. This further leads to that 
\[	\sup_{\big(\widetilde{X}_{1}, \dots, \widetilde{X}_{J}\big) \in  \mathcal{C}^{p_1, r_1}  \times \dots \times \mathcal{C}^{p_J, r_J}}\norm{Z_{j,\ell}\big(\widetilde{X}_1 \otimes \widetilde{X}_2 \otimes \cdots \otimes \widetilde{X}_{j-1} \otimes \widetilde{X}_{j+1}\otimes \cdots \otimes \widetilde{X}_J\big)}_2 \leq C_{\gamma}'\sigma\sqrt{pr},\]
with probability at least
$1-2(4J+1)^{Jpr}e^{-c\gamma'^2pr}$. By choosing $\gamma'$ sufficiently large and combining the above inequality with \eqref{eq:cover}, we finally obtain that there exist $\widetilde C_1, \widetilde c_1, \widetilde C_2>0$, such that 
\begin{equation}\label{eq:good-3}
	\sup_{\substack{X_{j'} \in \R^{p_{j'} \times r_{j'}}\\ \norm{X_{j'}}_2 \leq 1, j' \in [J] \setminus\{j\}}} \norm{Z_{j,\ell}\big(X_1 \otimes X_2 \otimes \cdots \otimes X_{j-1} \otimes X_{j+1}\otimes \cdots \otimes X_J\big)}_2 \leq 2\widetilde C_2\sigma\sqrt{pr},
\end{equation}
with probability at least $1-\widetilde C_1e^{-\widetilde c_1pr}$. Combining \eqref{eq:good-1}, \eqref{eq:good-2}, \eqref{eq:good-2-ave}, and \eqref{eq:good-3} leads to 
that $\Prob[E(C'_2)] \geq 1-C'_1Le^{-c'_1p}$ for some constants $C'_1, c'_1, C'_2>0$.
\end{proof}
\begin{proof}[Proof of Lemma \ref{lem:2}] 
Under event $E(C_2)$, it holds that $\sup_{j, \ell}\norm{Z_{j, \ell}U_{-j}}_2 \leq C_2\sigma\sqrt{p}$. Combining with $\norm{T_j}_2 = \norm{G_j}_2 \leq \kappa_0\lambda_{\min}$, it is straightforward to show that $\norm{\zeta_{j, \ell, 1}}_2 \leq C_2\kappa_0\lambda_{\min}\sigma\sqrt{p}$ and $\norm{\zeta_{j, \ell, 3}}_2 \leq\norm{Z_{j, \ell}U_{-j}}^2_2\leq C_2^2\sigma^2p$.

We have the following decomposition
\begin{equation}\label{eq:projection_decomp_-j}
	\begin{aligned}
		&\quad \hU^{(0)}_{-j, \ell}\hU_{-j, \ell}^{(0)\top} - U_{-j}U_{-j}^{\top} \\
		&=\sum_{j' \neq j}\bigg[ \left(U_{1}U_{1}^{\top}\right)\otimes \cdots \otimes \left(U_{j'-1}U_{j'-1}^{\top}\right) \otimes  \left(\hU^{(0)}_{j', \ell}\hU_{j', \ell}^{(0)\top} - U_{j'}U_{j'}^{\top}\right)\otimes \left(\hU^{(0)}_{j'+1, \ell}\hU_{j'+1, \ell}^{(0)\top}\right)\\
		&\quad\quad\quad\quad  \otimes \cdots \otimes\left(\hU^{(0)}_{j-1, \ell}\hU_{j-1, \ell}^{(0)\top}\right) \otimes \left(\hU^{(0)}_{j+1, \ell}\hU_{j+1, \ell}^{(0)\top}\right)\otimes \cdots \otimes \left(\hU^{(0)}_{J, \ell}\hU_{J, \ell}^{(0)\top}\right)\bigg]\\
		&=\sum_{j' \neq j}\bigg[ \left(U_{1}U_{1}^{\top}\right)\otimes \cdots \otimes \left(U_{j'-1}U_{j'-1}^{\top}\right) \otimes  \left(\hU^{(0)}_{j', \ell}\hU_{j', \ell}^{(0)\top} - \hU^{(0)}_{j', \ell}B_{j', \ell}A_{j', \ell}^{\top}U_{j'}^{\top}\right)\otimes \left(\hU^{(0)}_{j'+1, \ell}\hU_{j'+1, \ell}^{(0)\top}\right)\\
		&\quad\quad\quad\quad  \otimes \cdots \otimes\left(\hU^{(0)}_{j-1, \ell}\hU_{j-1, \ell}^{(0)\top}\right) \otimes \left(\hU^{(0)}_{j+1, \ell}\hU_{j+1, \ell}^{(0)\top}\right)\otimes \cdots \otimes \left(\hU^{(0)}_{J, \ell}\hU_{J, \ell}^{(0)\top}\right)\\
		&\qquad \qquad+ \left(U_{1}U_{1}^{\top}\right)\otimes \cdots \otimes \left(U_{j'-1}U_{j'-1}^{\top}\right) \otimes  \left(\hU^{(0)}_{j', \ell}B_{j', \ell}A_{j', \ell}^{\top}U_{j'}^{\top}-U_{j'}U_{j'}^{\top}\right)\otimes \left(\hU^{(0)}_{j'+1, \ell}\hU_{j'+1, \ell}^{(0)\top}\right)\\
		&\quad\quad\quad\quad  \otimes \cdots \otimes\left(\hU^{(0)}_{j-1, \ell}\hU_{j-1, \ell}^{(0)\top}\right) \otimes \left(\hU^{(0)}_{j+1, \ell}\hU_{j+1, \ell}^{(0)\top}\right)\otimes \cdots \otimes \left(\hU^{(0)}_{J, \ell}\hU_{J, \ell}^{(0)\top}\right)\bigg]\\
		&=\sum_{j' \neq j}\bigg[ \left(U_{1}\otimes \cdots \otimes U_{j'-1} \otimes \hU^{(0)}_{j', \ell} \otimes  \hU^{(0)}_{j'+1} \otimes \cdots \otimes \hU^{(0)}_{J} \right)\\
		&\qquad\cdot\left(U_{1}^{\top}\otimes \cdots \otimes U_{j'-1}^{\top} \otimes \left(\hU_{j', \ell}^{(0)\top} - B_{j', \ell}A_{j', \ell}^{\top}U_{j'}^{\top}\right) \otimes  \hU^{(0)\top}_{j'+1} \otimes \cdots \otimes \hU^{(0)\top}_{J} \right)\\		
		&\qquad\qquad+\left(U_{1}\otimes \cdots \otimes U_{j'-1} \otimes\left(\hU^{(0)}_{j', \ell}B_{j', \ell}A_{j', \ell}^{\top}-U_{j'}\right) \otimes  \hU^{(0)}_{j'+1} \otimes \cdots \otimes \hU^{(0)}_{J} \right)\\
		&\qquad\cdot\left(U_{1}^{\top}\otimes \cdots \otimes U_{j'-1}^{\top} \otimes U_{j'}^{\top} \otimes  \hU^{(0)\top}_{j'+1} \otimes \cdots \otimes \hU^{(0)\top}_{J} \right)\bigg]k,		
	\end{aligned}
\end{equation}
where $A_{j, \ell} \in \mathbb{O}^{r_j \times r_j}$ and $B_{j, \ell} \in \mathbb{O}^{r_j \times r_j}$ are defined by an SVD for $U_{j}^{\top}\hU^{(0)}_{j, \ell}$, that is,  $U_{j}^{\top}\hU^{(0)}_{j, \ell}=A_{j, \ell}S_{j, \ell}B_{j, \ell}^{\top}$. Note that
\begin{align*}
	&\quad \norm{\hU^{(0)}_{j', \ell}B_{j', \ell}A_{j', \ell}^{\top}-U_{j'}}_2 \\
	&= \norm{\begin{pmatrix}
			U_{j'}^{\top}\\ U_{j'\perp}^{\top}
		\end{pmatrix}\left(\hU^{(0)}_{j', \ell}B_{j', \ell}A_{j', \ell}^{\top}-U_{j'}\right)}_2\\
	&=\norm{\begin{pmatrix}
			U_{j'}^{\top}\hU^{(0)}_{j', \ell}B_{j', \ell}A_{j', \ell}^{\top}-I_{r_{j'}}\\ U_{j'\perp}^{\top}\hU^{(0)}_{j', \ell}B_{j', \ell}A_{j', \ell}^{\top}
	\end{pmatrix}}_2\\
	&\leq \sqrt{\norm{	U_{j'}^{\top}\hU^{(0)}_{j', \ell}B_{j', \ell}A_{j', \ell}^{\top}-I_{r_{j'}}}^2_2 +\norm{U_{j'\perp}^{\top}\hU^{(0)}_{j', \ell}B_{j', \ell}A_{j', \ell}^{\top}}^2_2 } \\
	&\leq \sqrt{\left[1-\lambda_{\min}\left(U_{j}^{\top}\hU^{(0)}_{j, \ell}\right)\right]^2+ 1-\lambda^2_{\min}\left(U_{j}^{\top}\hU^{(0)}_{j, \ell}\right)}\\
	& \leq \sqrt{2} \sqrt{1-\lambda^2_{\min}\left(U_{j}^{\top}\hU^{(0)}_{j, \ell}\right)}\leq\sqrt{2}\norm{\hU_{j, \ell}^{(0)}\hU_{j, \ell}^{(0)\top}-U_{j}U_{j}^{\top}}_2,
\end{align*}
where we sequentially use the following properties:

(1) For all block matrix $[A \; B]$, we have
\[\norm{[A \; B]}_2 = \sup_{\norm{\bm{v}}_2=1} \norm{[A \; B] \bm{v}}_2 = \sup_{\norm{\bm{v}_1}^2_2+\norm{\bm{v}_2}^2_2=1} \norm{A\bm{v}_1}_2 + \norm{B\bm{v}_2}_2 \leq \sqrt{\norm{A}_2^2 + \norm{B}_2^2}.\]

(2) For all $U, \hU \in \mathbb{O}^{p \times r}$, we have
\[\norm{U_{\perp}^{\top}\hU}_2 = \sup_{\norm{\bm{v}}_2=1} \norm{U_{\perp}^{\top}\hU \bm{v}}_2 = \sup_{\norm{\bm{v}}_2=1} \sqrt{\norm{\hU \bm{v}}^2_2-\norm{U^{\top}\hU \bm{v}}^2_2}=\sqrt{1-\lambda^2_{\min}\left(U^{\top}\hU\right)}.\]

(3) As $\norm{U_{j}^{\top}\hU^{(0)}_{j, \ell}}_2 \leq 1$, we have that $\lambda_{\min}\left(U_{j}^{\top}\hU^{(0)}_{j, \ell}\right) \leq 1$, and hence 
\[\sqrt{\left[1-\lambda_{\min}\left(U_{j}^{\top}\hU^{(0)}_{j, \ell}\right)\right]^2+ 1-\lambda^2_{\min}\left(U_{j}^{\top}\hU^{(0)}_{j, \ell}\right)} \leq \sqrt{2} \sqrt{1-\lambda^2_{\min}\left(U_{j}^{\top}\hU^{(0)}_{j, \ell}\right)}.\]

(4) For all $U, \hU \in \mathbb{O}^{p \times r}$,
\begin{equation}\label{eq:UperphU}
	\norm{U_{\perp}^{\top}\hU}_2 = \norm{U_{\perp}^{\top}\hU\hU^{\top}}_2 = \norm{U_{\perp}^{\top}UU^{\top}-U_{\perp}^{\top}\hU\hU^{\top}}_2 \leq \norm{UU^{\top}-\hU\hU^{\top}}_2.
\end{equation}

Since under $E(C_2)$, it holds that $\sup_{j, \ell}\norm{\hU_{j, \ell}^{(0)}\hU_{j, \ell}^{(0)\top}-U_{j, \ell}U_{j, \ell}^{\top}}_2 \leq C_2\sqrt{p}\sigma\lambda_{\min}^{-1}$, then we obtain that 
\[\sup_{j, \ell} \norm{\hU^{(0)}_{j', \ell}B_{j', \ell}A_{j', \ell}^{\top}-U_{j'}}_2 \leq C_2\sqrt{2}\sqrt{p}\sigma\lambda_{\min}^{-1}.\]

Combining the above inequality with \eqref{eq:projection_decomp_-j} and using that $\norm{T_j}_2 \leq \kappa_0\lambda_{\min}$ and 
\[\sup_{\substack{X_{j'} \in \R^{p_{j'} \times r_{j'}}\\ \norm{X_{j'}}_2 \leq 1, j'\in [J]\setminus \{j\}}} \norm{Z_{j,\ell}(X_1 \otimes X_2 \otimes \cdots \otimes X_{j-1} \otimes X_{j+1}\otimes \cdots \otimes X_J)}_2 \leq C_2\sigma\sqrt{pr}, \] 
we obtain that
\[
\sup_{j, \ell}\norm{\zeta_{j, \ell, 2}}_2 \leq \sup_{j, \ell} \norm{T_{j} \left[\hU^{(0)}_{-j,\ell}\hU_{-j,\ell}^{(0)\top}-U_{-j}U_{-j}^{\top}\right]Z^{\top}_{j,\ell}}_2 \lesssim  \kappa_0p\sigma^2 \sqrt{r},
\]
and
\[
\sup_{j, \ell}\norm{\zeta_{j, \ell, 4}}_2 \leq \sup_{j, \ell} \norm{Z_{j, \ell} \left[\hU^{(0)}_{-j,\ell}\hU_{-j,\ell}^{(0)\top}-U_{-j}U_{-j}^{\top}\right]Z^{\top}_{j,\ell}}_2 \lesssim  p^{3/2}r\sigma^3\lambda_{\min}^{-1}.
\]

Moreover, since	
\begin{align*}
	&\quad \hU^{(0)}_{-j, \ell}\hU_{-j, \ell}^{(0)\top} - U_{-j}U_{-j}^{\top} \\
	&=\sum_{j' \neq j}\bigg[ \left(U_{1}U_{1}^{\top}\right)\otimes \cdots \otimes \left(U_{j'-1}U_{j'-1}^{\top}\right) \otimes  \left(\hU^{(0)}_{j', \ell}\hU_{j', \ell}^{(0)\top} - U_{j'}U_{j'}^{\top}\right)\otimes \left(\hU^{(0)}_{j'+1, \ell}\hU_{j'+1, \ell}^{(0)\top}\right)\\
	&\quad\quad\quad\quad  \otimes \cdots \otimes\left(\hU^{(0)}_{j-1, \ell}\hU_{j-1, \ell}^{(0)\top}\right) \otimes \left(\hU^{(0)}_{j+1, \ell}\hU_{j+1, \ell}^{(0)\top}\right)\otimes \cdots \otimes \left(\hU^{(0)}_{J, \ell}\hU_{J, \ell}^{(0)\top}\right)\bigg],
\end{align*}
and 
$T_j = U_j G_j U_{-j}^{\top}$, we have that
\begin{align*}
	&\quad T_j \left[\hU^{(0)}_{-j, \ell}\hU_{-j, \ell}^{(0)\top} - U_{-j}U_{-j}^{\top}\right]T_j^{\top} \\
	&=U_j G_j\sum_{j' \neq j}\bigg[ \left(I_{r_1}\right)\otimes \cdots \otimes \left(I_{r_{j'-1}}\right) \otimes  \left(U_{j'}^{\top}\hU^{(0)}_{j', \ell}\hU_{j', \ell}^{(0)\top}U_{j'}-I_{r_{j'}}\right)\otimes \left(U_{j'+1}^{\top}\hU^{(0)}_{j'+1, \ell}\hU_{j'+1, \ell}^{(0)\top}U_{j'+1}\right)\\
	&\quad\quad  \otimes \cdots \otimes\left(U_{j-1}^{\top}\hU^{(0)}_{j-1, \ell}\hU_{j-1, \ell}^{(0)\top}U_{j-1}\right) \otimes \left(U_{j+1}^{\top}\hU^{(0)}_{j+1, \ell}\hU_{j+1, \ell}^{(0)\top}U_{j+1}\right)\otimes \cdots \otimes \left(U_{J}^{\top}\hU^{(0)}_{J, \ell}\hU_{J, \ell}^{(0)\top}U_{J}\right)\bigg] G_j^{\top}U_j^{\top}.
\end{align*}

Note that 
\[U_{j'}^{\top}\hU^{(0)}_{j', \ell}\hU_{j', \ell}^{(0)\top}U_{j'}-I_{r_{j'}} = U_{j'}^{\top}\left(\hU^{(0)}_{j', \ell}\hU_{j', \ell}^{(0)\top}-I_{r_{j'}}\right)U_{j'} = U_{j'}^{\top}\hU^{(0)}_{j', \ell, \perp}\hU_{j', \ell,\perp}^{(0)\top}U_{j'},\]
and by \eqref{eq:UperphU}, 
\[\norm{U_{j'}^{\top}\hU^{(0)}_{j', \ell, \perp}\hU_{j', \ell,\perp}^{(0)\top}U_{j'}}_2 \leq \norm{U_{j'}^{\top}\hU^{(0)}_{j', \ell, \perp}}_2^2 \leq \norm{U_{j'}U_{j'}^{\top}-\hU^{(0)}_{j', \ell}\hU^{(0)\top}_{j', \ell}}_2^2 \lesssim p\sigma^2\lambda_{\min}^{-2}.\]
Therefore,
\[\sup_{j, \ell}\norm{\zeta_{j, \ell, 5}}_2 \leq \sup_{j, \ell} \norm{T_{j} \left[\hU^{(0)}_{-j,\ell}\hU_{-j,\ell}^{(0)\top}-U_{-j}U_{-j}^{\top}\right]T^{\top}_{j}}_2 \lesssim \kappa_0^2p\sigma^2.\]
\end{proof}

\section{Statistical Inference in the Distributed Environment}\label{sec:asymp}

In this section, we provide theoretical analysis for the asymptotic distribution of our proposed distributed method in Algorithm 1 to feature statistical inference of the singular spaces. To establish the asymptotic distribution, we develop a two-iteration distributed procedure that obtains refined local estimators using each individual tensor and then aggregates them by averaging the projection matrices.  By establishing the asymptotic distribution, we provide a concise analysis of how aggregation helps to improve the statistical efficiency in the distributed environment. 

The two-iteration distributed procedure is formally displayed in Algorithm \ref{alg:inf}, and we give the asymptotic distribution of the estimation error $\rho(\hU_j, U_j)$ in the following theorem.


\begin{thm}
	Assume the assumptions in Theorem 2.1 hold. Further assume that $\lambda_{\min}/ \sigma \gtrsim L^{1/2}(pr)^{3/4}$
	and $\max(r^{3}, r^{J-1})/p=o(1)$. Then we have
	\begin{equation}\label{eq:normal}
		\frac{\rho^2\big(\hU_j, U_j\big)  - 2\sigma^2L^{-1}p_j\norm{\Lambda^{-1}_j}_{\rm F}^2}{\sqrt{8p_j}\sigma^2L^{-1}\norm{\Lambda_j^{-2}}_{\rm F}} \stackrel{d}{\to} \calN (0, 1),
	\end{equation}
	for $j\in[J]$, where $\hU_j$ is the output of Algorithm \ref{alg:inf}, and $\Lambda_j$ denotes
	$r_j \times r_j$ singular value matrix of $\calM_j(\calG)$.
	\label{thm:inference}
\end{thm}

The technical reason behind the two-iteration distributed procedure is to ensure a more precise quantification of the local estimation error. Recall that in Theorem 2.1, we assume the initial estimators $\big\{\hU^{(0)}_{j,\ell}\big\}$ have an error rate of the order $O(\sqrt{p}\sigma\lambda_{\min}^{-1})$.  However, to establish the asymptotic normality, we need a finer requirement for the initial estimation error, which may not be satisfied by $\big\{\hU_{j,\ell}^{(0)}\big\}$ but is satisfied by the local estimators $\{\hU_{j,\ell}\}$ obtained in Algorithm 1.  Therefore, we add one more iteration for each individual tensor in Algorithm 1 to obtain refined local estimators.  This two-iteration strategy was first investigated by a recent work \cite{xia2022inference} for single tensor decomposition. In Theorem \ref{thm:inference}, we further provide the asymptotic distribution for our distributed estimator and show that aggregation in a distributed setting helps improve the asymptotic Mean Squared Error (MSE) from a local rate $O(p^2r^2\sigma^4\lambda_{\min}^{-4})$ to the global optimal rate $O(L^{-2}p^2r^2\sigma^4\lambda_{\min}^{-4})$, when the SNR is sufficiently large.

Concretely, Theorem \ref{thm:inference} shows that, when the SNR satisfies $\lambda_{\min}/ \sigma \gtrsim L^{1/2}(pr)^{3/4}$, the squared distance $\rho^2(\hU_j, U_j)$ has an asymptotic bias $2\sigma^2L^{-1}p_j\big\|\Lambda_j^{-1}\big\|_{\rm F}^2$ and an asymptotic standard deviation $\sqrt{8p_j}\sigma^2L^{-1}\big\|\Lambda_j^{-2}\big\|_{\rm F}$. The SNR condition may be further weakened, which we leave for future investigation. Since $\big\|\Lambda_j^{-1}\big\|_{\rm F} \leq \sqrt{r}\big\|\Lambda_j^{-1}\big\|_2 \leq \sqrt{r} \lambda_{\min}^{-1}$,  the asymptotic MSE of $\rho^2(\hU_j, U_j)$ is of the order $O(p^2r^2\sigma^4L^{-2}\lambda_{\min}^{-4})$, which is consistent with the error rate established in Theorem 2.1. Moreover, the pooled estimator $\hU_{\mathrm{pooled}, j}$ (defined in Section 2.1) has the same asymptotic distribution as \eqref{eq:normal}, indicating that our estimator achieves the optimal asymptotic MSE we can obtain in the distributed setting.  

Based on Theorem \ref{thm:inference}, we can further construct confidence regions for $U_j$ using the proposed estimator $\hU_j$. Specifically, given a pre-specified level $1-\xi$,  we construct the  confidence region for $U_j$ as follows,
\begin{equation}\label{eq:confidence-region}
	\Big\{U \in \mathbb{O}^{p_j \times r_j}:\Big|\rho^2(\hU_j, U) -2\widehat{\sigma}^2L^{-1}p_j\big\|\widehat \Lambda_j^{-1}\big\|_{\rm F}^2\Big| \leq z_{1-\frac{\xi}{2}} \sqrt{8p_j}\widehat{\sigma}^2L^{-1}\big\|\widehat \Lambda_j^{-2}\big\|_{\rm F}\Big\},
\end{equation}
where $z_{1-\frac{\xi}{2}}$  denotes the $\big(1-\frac{\xi}{2}\big)$-th quantile of a standard normal distribution, and $\widehat\sigma$ and $\widehat\Lambda_j$ are consistent estimators for $\sigma$ and $\Lambda_j$.  In practice, the noise level $\sigma$ can be estimated by \[\widehat{\sigma} =  \big\|\calT_{\ell}-\calT_{\ell} \times_1 \hU_{1} \times_2 \hU_{2}\cdots\times_J\hU_{J}\big\|_{\rm F} / \sqrt{p_1p_2\cdots p_J}, \] 
and the singular value matrix $\Lambda_j$ can be estimated by
\[\widehat{\Lambda}_{j} = \text{the top $r_j$ singular values of } \calM_j\big(\calT_{\ell} \times_{1} \hU_{1}^{\top} \times_{2} \hU_{2}^{\top} \cdots \times_{j-1}\hU_{j-1}^{\top} \times_{j+1}\hU_{j+1}^{\top}\cdots \times_{J}\hU_{J}^{\top}\big),\] 
on any machine $\ell$. 

\begin{algorithm}[!t]
	\spacingset{1.2}
	\caption{Distributed Tensor PCA for Inference} 
	\label{alg:inf} 
	\vspace*{0.08in} {\bf Input:}
	Tensors distributed on local machines $\{\calT_{\ell}\}$ and initial estimators $\big\{\hU_{1,\ell}^{(0)}, \hU_{2,\ell}^{(0)}, \dots, \hU_{J,\ell}^{(0)}\big\}$, for all $\ell \in [L]$. 
	
	{\bf Output:}  Estimators $\big\{\hU_1, \hU_2, \dots, \hU_J\big\}$.
	\begin{algorithmic}[1]  
		\For{$\ell=1,2,\dots,L$, $j=1,2,\dots,J$, $t=1,2$} \do \\
		\State Compute a local estimator 
		\[
		\hU^{(t)}_{j,\ell} = \mathrm{svd}_{r_j}\Big[\calM_j\big(\calT_{\ell} \times_{1} \hU_{1,\ell}^{(t-1)\top} \times_{2} \hU_{2,\ell}^{(t-1)\top} \cdots \times_{j-1}\hU_{j-1,\ell}^{(t-1)\top} \times_{j+1}\hU_{j+1,\ell}^{(t-1)\top}\cdots \times_{J}\hU_{J,\ell}^{(t-1)\top}\big)\Big];
		\]
	\EndFor
	\State Send  $\big\{\hU^{(2)}_{j,\ell}\big\}_{j \in [J], \ell \in [L]}$ to the central machine; 
	
	\For{$j=1,2,\dots, J$}
	\State On the central machine, compute  $\hU_j = \mathrm{svd}_{r_j}\Big[\frac{1}{L}\sum\limits_{\ell=1}^{L}\hU^{(2)}_{j, \ell}\hU^{(2)\top}_{j, \ell}\Big]$;
	\EndFor
\end{algorithmic} 
\end{algorithm}

\subsection{Proof of the Results for Distributed Inference}
\begin{proof}[Proof for Theorem \ref{thm:inference}]
	
	Following the proof of Theorem 2.1, for $j \in [J]$, $t \in \{1, 2\}$, and $\ell \in [L]$, the estimator
	$\hU_{j,\ell}^{(t)}$ is composed of the first $r_j$  eigenvectors of
	\begin{equation}
		\quad (T_j + Z_{j, \ell}) \hU^{(t-1)}_{-j}\hU_{-j}^{(t-1)\top}(T^{\top}_j + Z^{\top}_{j, \ell})= T_jU_{-j}U_{-j}^{\top} T_j^{\top} +  \frakE^{(t)}_{j, \ell}= U_j\Lambda_j^2U_j^{\top} +  \frakE^{(t)}_{j, \ell},
		\label{eq:decomp_before_SVD}
	\end{equation}
	where $\frakE^{(t)}_{j, \ell}$ is a remainder term defined by
	\begin{equation}
		\begin{aligned}
			\frakE^{(t)}_{j, \ell}&:=\zeta_{j,\ell,1}+\zeta^{\top}_{j,\ell,1}+\zeta^{(t)}_{j,\ell,2}+\zeta^{(t)\top}_{j,\ell,2}+\zeta_{j,\ell,3}+\zeta^{(t)}_{j,\ell,4}+\zeta^{(t)}_{j,\ell,5},\\
			\zeta_{j,\ell,1}&:=T_{j} U_{-j}U_{-j}^{\top}Z^{\top}_{j,\ell}, \\
			\zeta^{(t)}_{j,\ell,2}&:= T_{j} \left[\hU^{(t-1)}_{-j}\hU^{(t-1)\top}_{-j}-U_{-j}U_{-j}^{\top}\right]Z^{\top}_{j,\ell},\\
			\zeta_{j,\ell,3}&:= Z_{j,\ell}U_{-j}U_{-j}^{\top}Z_{j,\ell}^{\top},\\
			\zeta^{(t)}_{j,\ell,4}&:= Z_{j,\ell}\left[\hU^{(t-1)}_{-j}\hU^{(t-1)\top}_{-j}-U_{-j}U_{-j}^{\top}\right]Z_{j,\ell}^{\top},\\
			\zeta^{(t)}_{j,\ell,5}&:=T_j \left[\hU^{(t-1)}_{-j}\hU^{(t-1)\top}_{-j}-U_{-j}U_{-j}^{\top}\right] T_j^{\top}.
		\end{aligned}
		\label{eq:frakE_decomp}
	\end{equation} 
	where 	\[U_{-j} = U_1 \otimes U_2 \otimes \cdots \otimes U_{j-1} \otimes U_{j+1} \otimes \cdots \otimes U_{J},\] 
	\[\hU^{(t)}_{-j} = \hU^{(t)}_1 \otimes \hU^{(t)}_2 \otimes \cdots \otimes \hU^{(t)}_{j-1} \otimes \hU^{(t)}_{j+1} \otimes \cdots \otimes \hU^{(t)}_{J} .\] 
	
	For each $j\in [J]$, $\ell\in [L]$, define a ``locally-good" event: 
	\begin{equation}\label{eq:good-two}
		\begin{aligned}
			&\quad \widetilde{E}_{j,\ell}(C)\\
			&:=\bigg\{\norm{\hU_{j,\ell}^{(0)}\hU_{j,\ell}^{(0)\top}-U_{j}U_{j}^{\top}}_2 \leq C \sqrt{p}\sigma\lambda_{\min}^{-1}, \, \norm{\hU_{j,\ell}^{(1)}\hU_{j,\ell}^{(1)\top}-U_{j}U_{j}^{\top}}_2 \leq C \sqrt{p}\sigma\lambda_{\min}^{-1}, \, 
			\norm{Z_{j,\ell}U_{-j}}_2 \leq C\sigma\sqrt{p}, \\ &\quad \sup_{\substack{X_{j'} \in \R^{p_{j'} \times r_{j'}}\\ \norm{X_{j'}}_2 \leq 1, j'\in [J]\setminus \{j\}}} \norm{Z_{j,\ell}(X_1 \otimes X_2 \otimes \cdots \otimes X_{j-1} \otimes X_{j+1}\otimes \cdots \otimes X_J)}_2 \leq C\sigma\sqrt{pr} \bigg\}.
		\end{aligned}
	\end{equation}
	and a ``globally-good'' event
	\[	\widetilde{E}(C):=\bigcap_{\ell=1}^L\bigcap_{j=1}^J E_{j,\ell}(C).\]
	By the proof of Lemma \ref{lem:1} and Theorem 2.1, we have $\Prob[\widetilde{E}(C_2)] \geq 1-C_1e^{-c_1p}$ for some constants $C_1, c_1, C_2$.
	By Lemma \ref{lem:2}, it holds that
	\begin{equation}
		\begin{aligned}
			&\norm{\zeta_{j,\ell,1}}_2 \leq C_2' \kappa_0\lambda_{\min}\sigma\sqrt{p}, \quad  \norm{\zeta^{(t)}_{j,\ell,2}}_2\leq C_2' \kappa_0p\sigma^2\sqrt{r},\\ &\norm{\zeta_{j,\ell,3}}_2 \leq C_2'p\sigma^2, \quad \norm{\zeta^{(t)}_{j,\ell, 4}}_2 \leq C_2' p^{3/2}\sqrt{r}\sigma^3\lambda_{\min}^{-1}, \quad  \norm{\zeta^{(t)}_{j,\ell, 5}}_2 \leq C_2' \kappa_0^2p\sigma^2,
		\end{aligned}
		\label{eq:frakE_bound} 
	\end{equation}
	for some absolute constant $C_2'>0$, under the event $\widetilde{E}(C_2)$.
	Moreover, similar to 	\eqref{eq:rho_decomp_thm-1}, we have
	\begin{equation}
		\begin{aligned}
			\hU^{(t)}_{j,\ell}\hU^{(t)\top}_{j,\ell} - U_jU_j^{\top} &= U_j\Lambda_{j}^{-2}U_j^{\top}\frakE^{(t)}_{j, \ell}U_{j\perp}U_{j\perp}^{\top} + U_{j\perp}U_{j\perp}^{\top}\frakE^{(t)}_{j, \ell}U_j\Lambda_{j}^{-2}U_j^{\top} \\
			&\quad-U_j\Lambda_{j}^{-2}U_j^{\top}\frakE^{(t)}_{j,\ell}U_{j\perp}U_{j\perp}^{\top}\frakE^{(t)}_{j, \ell}U_j\Lambda_{j}^{-2}U_j^{\top} 
			-U_j\Lambda_{j}^{-2}U_j^{\top}\frakE^{(t)}_{j,\ell}U_j\Lambda_{j}^{-2}U_j^{\top}\frakE^{(t)}_{j, \ell}U_{j\perp}U_{j\perp}^{\top}\\
			&\quad-U_{j\perp}U_{j\perp}^{\top}\frakE^{(t)}_{j,\ell}U_j\Lambda_{j}^{-2}U_j^{\top}\frakE^{(t)}_{j, \ell}U_j\Lambda_{j}^{-2}U_j^{\top}+U_j\Lambda_{j}^{-4}U_j^{\top}\frakE^{(t)}_{j,\ell}U_{j\perp}U_{j\perp}^{\top}\frakE^{(t)}_{j,\ell}U_{j\perp}U_{j\perp}^{\top}\\
			&\quad+U_{j\perp}U_{j\perp}^{\top}\frakE^{(t)}_{j,\ell}U_{j\perp}U_{j\perp}^{\top}\frakE^{(t)}_{j,\ell}U_j\Lambda_{j}^{-4}U_j^{\top}+U_{j\perp}U_{j\perp}^{\top}\frakE^{(t)}_{j,\ell}U_j\Lambda_{j}^{-4}U_j^{\top}\frakE^{(t)}_{j,\ell}U_{j\perp}U_{j\perp}^{\top}+\frakR^{(t)}_{j, \ell},
		\end{aligned}
		\label{eq:rho_decomp_thm}
	\end{equation}
	where $\norm{\frakR^{(t)}_{j, \ell}}_2 \lesssim \kappa_0^3\sigma^3p^{3/2}/\lambda_{\min}^3$.
	Then plugging \eqref{eq:frakE_decomp} into \eqref{eq:rho_decomp_thm} and using the fact that $U^{\top}_{j\perp}T_j=0$, we obtain
	\begin{equation}
		\hU^{(t)}_{j,\ell}\hU^{(t)\top}_{j,\ell} - U_jU_j^{\top}
		=\frakS^{(t)}_{j, \ell, 1}+\frakS^{(t)}_{j, \ell, 2}+\frakS^{(t)}_{j, \ell, 3},
		\label{eq:rho_decomp_plugin}
	\end{equation}
	where for $j \in [J]$ and $t=1, 2$,
	\begin{equation}
		\label{eq:def-Sjlt12}
		\begin{aligned}
			\frakS_{j, \ell, 1}&:=U_j\Lambda_{j}^{-2}U_j^{\top}\zeta_{j, \ell,1}U_{j\perp}U_{j\perp}^{\top} +U_{j\perp}U_{j\perp}^{\top}\zeta^{\top}_{j, \ell,1}U_j\Lambda_{j}^{-2}U_j^{\top}\\
			&= U_j\Lambda_{j}^{-2}G_j U_{-j}^{\top}Z^{\top}_{j, \ell}U_{j\perp}U_{j\perp}^{\top} +  U_{j\perp}U_{j\perp}^{\top}Z_{j, \ell} U_{-j}G_j^{\top}\Lambda_{j}^{-2}U_j^{\top},\\
			\frakS^{(t)}_{j, \ell, 2}&:= U_j\Lambda_{j}^{-2}U_j^{\top}\zeta^{(t)}_{j, \ell, 2}U_{j\perp}U_{j\perp}^{\top}+U_{j\perp}U_{j\perp}^{\top}\zeta^{(t)\top}_{j, \ell, 2}U_j\Lambda_{j}^{-2}U_j^{\top}\\
			&\quad+U_j\Lambda_{j}^{-2}U_j^{\top}\zeta_{j, \ell, 3}U_{j\perp}U_{j\perp}^{\top}+U_{j\perp}U_{j\perp}^{\top}\zeta_{j, \ell, 3}U_j\Lambda_{j}^{-2}U_j^{\top}\\
			&\quad-U_j\Lambda_{j}^{-2}U_j^{\top}\zeta_{j, \ell, 1}U_{j\perp}U_{j\perp}^{\top}\zeta^{\top}_{j, \ell, 1}U_j\Lambda_{j}^{-2}U_j^{\top}\\
			&\quad-U_j\Lambda_{j}^{-2}U_j^{\top}(\zeta_{j, \ell, 1}+\zeta^{\top}_{j, \ell, 1})U_j\Lambda_{j}^{-2}U_j^{\top}\zeta_{j,\ell, 1}U_{j\perp}U_{j\perp}^{\top}\\
			&\quad-U_{j\perp}U_{j\perp}^{\top}\zeta^{\top}_{j, \ell, 1}U_j\Lambda_{j}^{-2}U_j^{\top}(\zeta_{j, \ell, 1}+\zeta^{\top}_{j, \ell, 1})U_j\Lambda_{j}^{-2}U_j^{\top}\\
			&\quad +U_{j\perp}U_{j\perp}^{\top}\zeta^{\top}_{j, \ell, 1}U_j\Lambda_{j}^{-4}U_j^{\top}\zeta_{j, \ell, 1}U_{j\perp}U_{j\perp}^{\top},
		\end{aligned}
	\end{equation}
	and $\frakS_{j, \ell,3}^{(t)}$ is a remainder term determined by the above equations.
	Then
	\begin{equation}
		\begin{aligned}
			&\quad \frac{1}{L}\sum_{\ell=1}^L\hU^{(t)}_{j,\ell}\hU^{(t)\top}_{j,\ell} - U_jU_j^{\top}\\ &=\overline{\frakS_{j, 1}}+ \overline{\frakS^{(t)}_{j, 2}} + \overline{\frakS^{(t)}_{j, 3}}\\
			&= U_j\Lambda_{j}^{-2}G_j U_{-j}^{\top}\overline{Z}^{\top}_{j}U_{j\perp}U_{j\perp}^{\top} +  U_{j\perp}U_{j\perp}^{\top}\overline{Z}_{j} U_{-j}G_j^{\top}\Lambda_{j}^{-2}U_j^{\top}+ \overline{\frakS^{(t)}_{j, 2}} + \overline{\frakS^{(t)}_{j, 3}},
		\end{aligned}
	\end{equation}
	where $\overline{Z}_j:=(1/L)\sum_{\ell=1}^L Z_{j, \ell}$ and $\overline{\frakS^{(t)}_{j,k}}:=(1/L)\sum_{\ell=1}^L \frakS^{(t)}_{j, \ell, k}$ for $k \in \{2,3\}$ satisfies
	\[	\norm{\overline{\frakS^{(t)}_{j,2}}}_2 \lesssim \kappa_0^2pr^{1/2}\sigma^2\lambda_{\min}^{-2}, \quad \norm{\overline{\frakS^{(t)}_{j,3}}}_2 \lesssim \kappa_0^3p^{3/2}r^{1/2}\sigma^3\lambda_{\min}^{-3}.\] 
	
	By \eqref{eq:Sj1}, it holds that
	\[ \norm{\overline{\frakS_{j,1}}}_2=\norm{U_j\Lambda_{j}^{-2}G_j U_{-j}^{\top}\overline{Z}^{\top}_jU_{j\perp}U_{j\perp}^{\top} +  U_{j\perp}U_{j\perp}^{\top}\overline{Z}_{j} U_{-j}G_j^{\top}\Lambda_{j}^{-2}U_j^{\top}}_2 \lesssim \lambda^{-1}_{\min}\sigma\sqrt{p/L},\]
	with probability at least $1-C_1e^{-c_1p}$.
	
	Therefore, we obtain that
	\begin{equation}
		\label{eq:rate-hU2}
		\frac{1}{L}\sum_{\ell=1}^L\hU^{(2)}_{j,\ell}\hU^{(2)\top}_{j,\ell} - U_jU_j^{\top}=\overline{\frakS_{j, 1}}+ \overline{\frakS^{(2)}_{j, 2}} + \overline{\frakS^{(2)}_{j, 3}},
	\end{equation}
	where
	\begin{equation}
		\label{eq:bound-2-S123}
		\norm{\overline{\frakS_{j,1}}}_2 \lesssim \lambda^{-1}_{\min}\sigma\sqrt{p/L}, \quad
		\norm{\overline{\frakS^{(2)}_{j,2}}}_2 \lesssim \kappa_0^2pr^{1/2}\sigma^2\lambda_{\min}^{-2}, \quad \norm{\overline{\frakS^{(2)}_{j,3}}}_2 \lesssim \kappa_0^3p^{3/2}r^{1/2}\sigma^3\lambda_{\min}^{-3},
	\end{equation}
	with probability at least $1-C'_1Le^{-c'_1p}$ for some absolute constant $C_1', c_1'>0$.
	
	Since the columns of $\hU_j$ are the first $r_j$ eigenvectors of $\frac{1}{L}\sum_{\ell=1}^L\hU^{(2)}_{j,\ell}\hU^{(2)\top}_{j,\ell}$, similar to \eqref{eq:decomp-after-SVD-thm-1},
	\begin{equation}
		\label{eq:decomp-after-SVD-thm}
		\begin{aligned}
			\hU_j\hU_j^{\top}-U_jU_j^{\top}&=U_jU_j^{\top}\left(\overline{\frakS_{j, 1}}+ \overline{\frakS^{(2)}_{j, 2}}\right)U_{j\perp}U_{j\perp}^{\top}+ 
			U_{j\perp}U_{j\perp}^{\top}\left(\overline{\frakS_{j, 1}}+ \overline{\frakS^{(2)}_{j, 2}}\right)U_jU_j^{\top} \\ 
			&\quad -U_jU_j^{\top}\overline{\frakS_{j, 1}}U_{j\perp}U_{j\perp}^{\top} \overline{\frakS_{j, 1}}U_jU_j^{\top} + U_{j\perp}U_{j\perp}^{\top}\overline{\frakS_{j, 1}}U_jU_j^{\top}\overline{\frakS_{j, 1}}U_{j\perp}U_{j\perp}^{\top}
			+ \widetilde{\frakR}_j,
		\end{aligned}
	\end{equation}
	where $\norm{\widetilde{\frakR}_j}_2 \lesssim \kappa_0^3p^{3/2}r^{1/2}\sigma^3\lambda_{\min}^{-3}$. By \eqref{eq:bound-2-S123}, we obtain that
	\[\norm{\hU_j\hU_j^{\top}-U_jU_j^{\top}}_{\rm F} \leq \sqrt{r}\norm{\hU_j\hU_j^{\top}-U_jU_j^{\top}}_2 =O_{\Prob}\left(\lambda^{-1}_{\min}\sigma\sqrt{pr/L}+ \kappa_0^2pr\sigma^2\lambda_{\min}^{-2}\right).\]
	Moreover,
	\begin{equation}
		\label{eq:decomp-square}
		\begin{aligned}
			&\quad\norm{\hU_j\hU_j^{\top}-U_jU_j^{\top}}_{\rm F}^2\\
			&=\angles{\hU_j\hU_j^{\top}-U_jU_j^{\top}, \hU_j\hU_j^{\top}-U_jU_j^{\top}}\\
			&=2\tr \left(U_j^{\top}\overline{\frakS_{j, 1}}U_{j\perp}U_{j\perp}^{\top}\overline{\frakS_{j, 1}}U_j\right)+4\tr \left(U_j^{\top}\overline{\frakS^{(2)}_{j, 2}}U_{j\perp}U_{j\perp}^{\top}\overline{\frakS_{j, 1}}U_j\right)+\mathfrak{Q}_j\\
			&=2\tr\left(\Lambda_{j}^{-4}G_jU_{-j}^{\top}\overline{Z}^{\top}_{j}U_{j\perp}U_{j\perp}^{\top}\overline{Z}_{j} U_{-j}G_j^{\top}\right)\\
			&\quad+ 4\tr\left(\Lambda_{j}^{-4}U^{\top}_j\overline{\zeta^{(2)}_{j, 2}}U_{j\perp}U_{j\perp}^{\top}\overline{Z_{j}}U_{-j}G_j^{\top}\right)\\
			&\quad+4\tr\left(\Lambda_{j}^{-4}U^{\top}_j\overline{\zeta_{j, 3}}U_{j\perp}U_{j\perp}^{\top}\overline{Z_{j}}U_{-j}G_j^{\top}\right)\\
			&\quad - 4\tr\left\lbrace\Lambda_{j}^{-4}U^{\top}_j\left[\frac{1}{L}\sum_{\ell=1}^L\left(\zeta_{j,\ell, 1}+\zeta^{\top}_{j,\ell, 1}\right)U_j\Lambda_{j}^{-2}U_j^{\top}\zeta_{j,\ell, 1}\right]U_{j\perp}U_{j\perp}^{\top}\overline{Z_{j}}U_{-j}G_j^{\top}\right\rbrace+\mathfrak{Q}_j,
		\end{aligned}
	\end{equation}
	where $\overline{\zeta^{(t)}_{j, k}}:=(1/L)\sum_{\ell=1}^L\zeta^{(t)}_{j, \ell, k}$ and $\mathfrak{Q}_j$ is a remainder term. By the inequality that $|\angles{A, B}| \leq r\norm{A}_2\norm{B}_2$ for all rank-$r$ matrices $A$ and $B$, we bound the remainder term by $|\mathfrak{Q}_j| \lesssim \kappa_0^4p^2r^2\sigma^4\lambda_{\min}^{-4}$.
	
	Now we provide bounds for the second, third and fourth terms in \eqref{eq:decomp-square}. By \eqref{eq:rho_decomp_plugin},
	\[U_j^{\top}\hU_{j, \ell}^{(1)}\hU_{j, \ell}^{(1)\top}=U_j^{\top}+\Lambda_{j}^{-2}G_jU_{-j}^{\top}Z^{\top}_{j,\ell}U_{j\perp}U_{j\perp}^{\top}+ U_j^{\top}\left(\overline{\frakS^{(1)}_{j, 2}} + \overline{\frakS^{(1)}_{j, 3}}\right).\]
	Hence,
	\begin{equation}
		\begin{aligned}
			&\quad\zeta^{(2)}_{j, \ell, 2}\\
			&=T_j\left[\hU^{(1)}_{-j}\hU^{(1)\top}_{-j}-U_{-j}U_{-j}^{\top}\right]Z^{\top}_{j,\ell}\\
			&=
			U_jG_jU_{-j}^{\top}\left[\hU^{(1)}_{-j}\hU^{(1)\top}_{-j}-U_{-j}U_{-j}^{\top}\right]Z_{j,\ell}^{\top}\\
			&= U_jG_j\left[\bigotimes_{j' \neq j}\left(U_{j'}^{\top}\hU_{j', \ell}^{(1)}\hU_{j', \ell}^{(1)\top} \right) -U_{-j}^{\top} \right]Z_{j,\ell}^{\top}\\
			&= U_jG_j \bigg[\sum_{j'\neq j}U^{\top}_{1}  \otimes \cdots U^{\top}_{j'-1}\otimes \left(\Lambda_{j'}^{-2}G_{j'} U_{-j'}Z_{j',\ell}^{\top}U_{j'\perp}U_{j'\perp}^{\top}\right)\otimes U^{\top}_{j'+1}\cdots U_{j-1}^{\top} \otimes U_{j+1}^{\top} \cdots \otimes U^{\top}_{J}\bigg]Z_{j,\ell}^{\top} \\
			&\quad + \mathfrak{R}_{\zeta^{(2)}_{j,\ell, 2}},
		\end{aligned}
	\end{equation}
	where $\norm{\mathfrak{R}_{\zeta^{(2)}_{j,\ell, 2}}}_2 \leq C_4  \kappa_0^3p^{3/2}r^{1/2}\sigma^3\lambda_{\min}^{-1}$  for some $C_4>0$. As a result, we obtain
	\begin{equation}
		\quad\tr\left(\Lambda_{j}^{-4}U^{\top}_j\overline{\zeta^{(2)}_{j, 2}}U_{j\perp}U_{j\perp}^{\top}\overline{Z_{j}}U_{-j}G_j^{\top}\right)=\sum_{j'\neq j}\mathfrak{M}_{j'}+\mathfrak{R}_{\mathfrak{M}},
	\end{equation}
	where
	\begin{equation}
		\begin{aligned}
			&\mathfrak{M}_{j'}:=\tr\Bigg\{\Lambda_{j}^{-4}G_j\\
			&\cdot \bigg[\frac{1}{L}\sum_{\ell=1}^L\left(U^{\top}_{1}  \otimes \cdots\otimes  U^{\top}_{j'-1}\otimes\left( \Lambda_{j'}^{-2}G_{j'} U_{-j'}Z_{j',\ell}^{\top}U_{j'\perp}U_{j'\perp}^{\top}\right)\otimes U^{\top}_{j'+1}\cdots \otimes U_{j-1}^{\top} \otimes U_{j+1}^{\top} \otimes  \cdots \otimes  U^{\top}_{J}\right)Z^{\top}_{j,\ell}\bigg]\cdot\\ 
			&\quad \quad\quad\quad U_{j\perp} U_{j\perp}^{\top}\overline{Z_{j}}U_{-j}G_j^{\top}\Bigg\},
		\end{aligned}
		\label{eq:decomp-2}
	\end{equation}
	and $\abs{\frakR_{\mathfrak{M}}} \lesssim \kappa_0^4p^2r^2\sigma^4\lambda_{\min}^{-4}$.  Define 
	\[\tZ_{j', \ell, a}:=U^{\top}_{j'\perp}Z_{j',\ell}U_{-j'},\] \[\tZ_{j', \ell, b}:=U^{\top}_{j\perp}Z_{j,\ell}(U_{1}\otimes\cdots \otimes U_{j'-1} \otimes U_{j'\perp} \otimes U_{j'+1} \otimes\cdots \otimes U_{j-1} \otimes U_{j+1} \otimes \cdots \otimes U_{J}),\] 
	\[\tZ_{j, \ell,c}:=U^{\top}_{j\perp}Z_{j,\ell}U_{-j},\]then $\mathfrak{M}_{j'}$ can be simplified as
	\[\mathfrak{M}_{j'}=\tr\left[\Lambda_j^{-4}G_j\frac{1}{L}\sum_{\ell=1}^L\left( I_{r_{1}}  \otimes \cdots I_{r_{j'-1}}\otimes\left(\Lambda^{-2}_{j'}G_{j'}\tZ^{\top}_{j',\ell,a}\right)\otimes I_{r_{j'+1}} \cdots \otimes I_{r_{J}}\right)\tZ^{\top}_{j',\ell,b}\overline{\tZ_{j,c}}G_j^{\top}\right],\]
	where $\overline{\tZ_{j,c}}=(1/L)\sum_{\ell=1}^L\tZ_{j,\ell,c}$.
	
	Let $\ve$ denote the vectorization of a matrix. By assumption, $\ve (Z_{j',\ell}) \sim \calN(0, \sigma^2I)$. Using the identity that $\ve(ABC)=(C^{\top} \otimes A) \ve(B)$ for all matrices $A,B,C$, we obtain that $\ve(\tZ_{j', \ell, a}) \sim \calN(0, \sigma^2I)$, i.e., the entries of $\tZ_{j', \ell, a}$ are i.i.d. $\calN(0,\sigma^2)$. Similarly, $\ve(\tZ_{j', \ell,b}) \sim \calN(0, \sigma^2I)$ and $\ve(\tZ_{j, \ell,c}) \sim \calN(0, \sigma^2I)$. Furthermore, since $U_{j'\perp}^{\top}U_{j'}=0$,  $\E[\ve(\tZ_{j', \ell,b})\ve(\tZ_{j, \ell,c})^{\top}]=0$, which implies that  $\tZ_{j',\ell,b}$ and $\tZ_{j,\ell,c}$ are independent. Therefore, conditional on $\{\tZ_{j',\ell,a}, \tZ_{j,\ell,c}\}_{\ell=1}^L$, 
	\begin{equation}
		\label{eq:normal_M1}
		\begin{aligned}
			&\quad \mathfrak{M}_{j'} \mid \{\tZ_{j',\ell,a}, \tZ_{j,\ell,c}\}_{\ell=1}^L \\
			&\sim \calN\left(0, \frac{\sigma^2}{L^2}\sum_{\ell=1}^L\norm{\overline{\tZ_{j,c}}G^{\top}_j\Lambda_j^{-4}G_j\left(I_{r_{1}}  \otimes \cdots I_{r_{j'-1}}\otimes\left(\Lambda^{-2}_{j'}G_{j'}\tZ^{\top}_{j',\ell,a}\right)\otimes I_{r_{j'+1}} \cdots \otimes I_{r_{J}}\right)}^2_{F}\right).
		\end{aligned}
	\end{equation}
	By the proof of Theorem 2.1, with probability at least $1-C_1e^{-c_1p}$, $\norm{\tZ_{j',\ell,a}}_2 \leq C_2\sigma\sqrt{p}$ and $\norm{\overline{\tZ_{j,c}}}_2 \leq C_2\sigma\sqrt{p/L}$.
	Note that \[\mathrm{rank}\left[\overline{\tZ_{j,c}}G^{\top}_j\Lambda_j^{-4}G_j\left(I_{r_{1}}  \otimes \cdots I_{r_{j'-1}}\otimes\left(\Lambda^{-2}_{j'}G_{j'}\tZ^{\top}_{j',\ell,a}\right)\otimes I_{r_{j'+1}} \cdots \otimes I_{r_{J}}\right)\right] \leq \mathrm{rank}(G_j)=r_j.\] 
	Then, using that $\norm{\Lambda_j^{-1}G_j}_2=1$, we obtain,
	\[\norm{\overline{\tZ_{j,c}}G^{\top}_j\Lambda_j^{-4}G_j\left(I_{r_{1}}  \otimes \cdots I_{r_{j'-1}}\otimes\left(\Lambda^{-2}_{j'}G_{j'}\tZ^{\top}_{j',\ell,a}\right)\otimes I_{r_{j'+1}} \cdots \otimes I_{r_{J}}\right)}_{F} \leq C_2^{2}pr^{1/2}L^{-1/2}\sigma^2\lambda_{\min}^{-3}.\]
	Hence,
	\begin{align*}
		\mathrm{sd}(	\mathfrak{M}_{j'})&:=\sqrt{\frac{\sigma^2}{L^2}\sum_{\ell=1}^L\norm{\overline{\tZ_{j,c}}G^{\top}_j\Lambda_j^{-4}G_j\left(I_{r_{1}}  \otimes \cdots I_{r_{j'-1}}\otimes\left(\Lambda^{-2}_{j'}G_{j'}\tZ^{\top}_{j',\ell,a}\right)\otimes I_{r_{j'+1}} \cdots \otimes I_{r_{J}}\right)}^2_{F}} \\
		&\lesssim  L^{-1} pr^{1/2}\sigma^3\lambda_{\min}^{-3}.
	\end{align*}
	
	By \eqref{eq:normal_M1}, for any $\gamma>0$, 
	\[\Prob\left(\abs{\mathfrak{M}_{j'}} \geq  \mathrm{sd}(\mathfrak{M}_{j'})\sqrt{2\gamma\log p} \quad \big|\quad \{\tZ_{j',\ell,a}, \tZ_{j,\ell,c}\}_{\ell=1}^L \right) \leq 2p^{-\gamma}.\]
	Therefore, with probability at least $1-C_1Le^{-c_1p}-2p^{-\gamma}$, 
	\[\abs{\mathfrak{M}_{j'}} \lesssim \frac{p\sigma^3\sqrt{\gamma r \log p}}{\lambda_{\min}^3L}.\]
	In conclusion, the second term in \eqref{eq:decomp-square} can be bounded by
	\begin{equation}
		\abs{4\tr\left(\Lambda_{j}^{-4}U^{\top}_j\overline{\zeta^{(2)}_{j, 2}}U_{j\perp}U_{j\perp}^{\top}\overline{Z_{j}}U_{-j}G_j^{\top}\right)} \lesssim \frac{p\sigma^3\sqrt{\gamma r \log p}}{\lambda_{\min}^3L} + \kappa_0^4p^2r^2\sigma^4\lambda_{\min}^{-4}, \label{eq:bound-M123}
	\end{equation}
	with probability at least $1-C_1JLe^{-c_1p}-2Jp^{-\gamma}$.

	Next, we deal with the third term in \eqref{eq:decomp-square}, i.e.,
	\[\tr\left(\Lambda_{j}^{-4}U^{\top}_j\overline{\zeta_{j, 3}}U_{j\perp}U_{j\perp}^{\top}\overline{Z_{j}}U_{-j}G_j^{\top}\right)=\frac{1}{L}\sum_{\ell=1}^L\tr \left[\Lambda_j^{-4}\tZ_{j,\ell,d}\tZ_{j,\ell,c}^{\top}\overline{\tZ_{j,c}}G_j^{\top}\right]:=\widetilde{\mathfrak{M}_0},\]
	where $\tZ_{j,\ell,d}:=U_{j}^{\top}Z_{j,\ell}U_{-j}$. Since $\ve([\tZ_{j,\ell,d} \tZ_{j,\ell,c}])=\ve([U_j U_{j\perp}]^{\top}Z_{j,\ell}U_{-j}) \sim \calN(0,\sigma^2 I)$, we have 
	\[\widetilde{\mathfrak{M}_0} \mid \{\tZ_{j, \ell,c}\}_{\ell=1}^L \sim \calN \left(0, \frac{\sigma^2}{L^2} \sum_{\ell=1}^L \norm{\Lambda_j^{-4}G_j\overline{\tZ_{j,c}}^{\top}\tZ_{j,\ell,c}}_{\rm F}^2\right).\]
	
	Using that $\norm{\tZ_{j,\ell,c}}_2 \leq C_2\sigma\sqrt{p}$ and $\norm{\overline{\tZ_{j,c}}}_2 \leq C_2\sigma\sqrt{p/L}$ with probability at least $1-C_1e^{-c_1p}$, we have
	$\mathrm{sd}(\widetilde{\mathfrak{M}}_0) \lesssim \sigma^3 L^{-1}r^{1/2}p\lambda^{-3}_{\min}$. Therefore, with probability at least $1-C_1Le^{-c_1p}-2p^{-\gamma}$, we obtain that 
	\begin{equation}
		\abs{\widetilde{\mathfrak{M}}_0} \lesssim \frac{p\sigma^3\sqrt{\gamma r \log p}}{\lambda_{\min}^3L}. 
		\label{eq:bound-M4}
	\end{equation}
	
	For the fourth term in \eqref{eq:decomp-square}, note that
	\begin{equation}
		\begin{aligned}
			&\quad\tr\left\lbrace\Lambda_{j}^{-4}U^{\top}_j\left[\frac{1}{L}\sum_{\ell=1}^L\left(\zeta_{j,\ell, 1}+\zeta^{\top}_{j,\ell, 1}\right)U_j\Lambda_{j}^{-2}U_j^{\top}\zeta_{j,\ell, 1}\right]U_{j\perp}U_{j\perp}^{\top}\overline{Z_{j}}U_{-j}G_j^{\top}\right\rbrace\\
			&=\frac{1}{L}\sum_{\ell=1}^L\tr\left[\Lambda_{j}^{-4}G_j\tZ^{\top}_{j,\ell,d}\Lambda_{j}^{-2}G_j\tZ_{j,\ell,c}^{\top}\overline{\tZ_{j,c}}G_j^{\top}\right]+\frac{1}{L}\sum_{\ell=1}^L\tr\left[\Lambda_{j}^{-4}\tZ_{j,\ell,d}G_j^{\top}\Lambda_{j}^{-2}G_j\tZ_{j,\ell,c}^{\top}\overline{\tZ_{j,c}}G_j^{\top}\right]\\
			&=:\widetilde{\mathfrak{M}}_1+\widetilde{\mathfrak{M}}_2.
		\end{aligned}
	\end{equation}
	Repeating the analysis for $\widetilde{\mathfrak{M}}_0$ yields the same result that 
	\begin{equation}
		\abs{\widetilde{\mathfrak{M}}_1} \lesssim \frac{p\sigma^3\sqrt{\gamma r \log p}}{\lambda_{\min}^3L}, \quad	\abs{\widetilde{\mathfrak{M}}_2} \lesssim \frac{p\sigma^3\sqrt{\gamma r \log p}}{\lambda_{\min}^3L}. 
		\label{eq:bound-M56}
	\end{equation}
	Combining \eqref{eq:decomp-square}, \eqref{eq:bound-M123}, \eqref{eq:bound-M4} and \eqref{eq:bound-M56} leads to that
	\begin{equation}
		\label{eq:square-remainder}
		\begin{aligned}
			&\quad\abs{\norm{\hU_j\hU_j^{\top}-U_jU_j^{\top}}_{\rm F}^2-2\tr\left(\Lambda_{j}^{-4}G_jU_{-j}^{\top}\overline{Z}^{\top}_{j}U_{j\perp}U_{j\perp}^{\top}\overline{Z}_{j} U_{-j}G_j^{\top}\right)}\\ &=O_{\Prob}\left(\frac{p\sigma^3\sqrt{ r \log p}}{\lambda_{\min}^3L}+\frac{\kappa_0^4p^2r^2\sigma^4}{\lambda_{\min}^{4}}\right).
		\end{aligned}
	\end{equation}
	Now we focus on the first term. By the proof of the final step of Theorem 1 in \cite{xia2022inference}, it holds that
	\begin{align*}
		&\quad2\tr\left(\Lambda_{j}^{-4}G_jU_{-j}^{\top}\overline{Z}^{\top}_{j}U_{j\perp}U_{j\perp}^{\top}\overline{Z}_{j} U_{-j}G_j^{\top}\right)\\
		&=2\norm{\Lambda_{j}^{-2}G_jU_{-j}^{\top}\overline{Z}^{\top}_{j}U_{j\perp}}_{\rm F}^2\\
		&\stackrel{d}{=} \frac{2\sigma^2}{L}\sum\limits_{i=1}^{p_j-r_j}\norm{\Lambda_j^{-1}\bm{z}_i}_2^2,
	\end{align*} 
	where $\bm{z}_i \stackrel{i.i.d.}{\sim} \calN(0, I_{r_j})$. Since $\E\left[\norm{\Lambda_1^{-1}\bm{z}_i}_2^2\right]=\norm{\Lambda_j^{-1}}_{\rm F}^2$ and $\mathrm{Var}\left[\norm{\Lambda_j^{-1}\bm{z}_i}_2^2\right]=2\norm{\Lambda_j^{-2}}_{\rm F}^2$,  by Central Limit Theorem,
	
	\[\frac{2\norm{\Lambda_{j}^{-2}G_jU_{-j}^{\top}\overline{Z}^{\top}_{j}U_{j\perp}}_{\rm F}^2 - 2\sigma^2L^{-1}(p_j-r_j)\norm{\Lambda^{-1}_j}_{\rm F}^2}{\sqrt{8(p_j-r_j)}\sigma^2L^{-1}\norm{\Lambda_j^{-2}}_{\rm F}} \stackrel{d}{\to} \calN (0, 1).\]
	
	Since $\sqrt{8(p_j-r_j)}\sigma^2L^{-1}\norm{\Lambda_j^{-2}}_{\rm F} \gtrsim \sqrt{p_jr_j}L^{-1}\sigma^2\kappa_0^{-2}\lambda_{\min}^{-2}$ and $r_j/p_j=o(1)$, by \eqref{eq:square-remainder}, it holds that
	\[	
	\frac{\norm{\hU_j\hU_j^{\top}-U_jU_j^{\top}}_{\rm F}^2  - 2\sigma^2L^{-1}(p_j-r_j)\norm{\Lambda^{-1}_j}_{\rm F}^2}{\sqrt{8p_j}\sigma^2L^{-1}\norm{\Lambda_j^{-2}}_{\rm F}} \stackrel{d}{\to} \calN (0, 1),\]
	if 
	\[\left(\frac{p\sigma^3\sqrt{ r \log p}}{\lambda_{\min}^3L}+\frac{\kappa_0^4p^2r^2\sigma^4}{\lambda_{\min}^{4}}\right)/\frac{\sqrt{p_jr_j}\sigma^2}{L\kappa_0^{2}\lambda_{\min}^{2}}=o(1),\]
	or
	\[\frac{Lr^{3/2}\kappa_0^6p^{3/2}}{(\lambda_{\min}/\sigma)^2}+\frac{\kappa_0^2\sqrt{p\log p}}{\lambda_{\min}/\sigma}=o(1).\]
	Furthermore, if $r^3/p=o(1)$, we have $r_j\norm{\Lambda^{-1}_j}_{\rm F}^2/\sqrt{p_j}\norm{\Lambda_j^{-2}}_{\rm F}=o(1)$, then 
	\[\frac{\norm{\hU_j\hU_j^{\top}-U_jU_j^{\top}}_{\rm F}^2  - 2\sigma^2L^{-1}p_j\norm{\Lambda^{-1}_j}_{\rm F}^2}{\sqrt{8p_j}\sigma^2L^{-1}\norm{\Lambda_j^{-2}}_{\rm F}} \stackrel{d}{\to} \calN (0, 1).\] 
\end{proof}

\section{Additional Results in Numerical Studies}\label{sec:supp_numerical}

\begin{figure}[!t]
	\centering
	\begin{subfigure}[b]{0.45\textwidth}
		\centering
		\includegraphics[width=\textwidth]{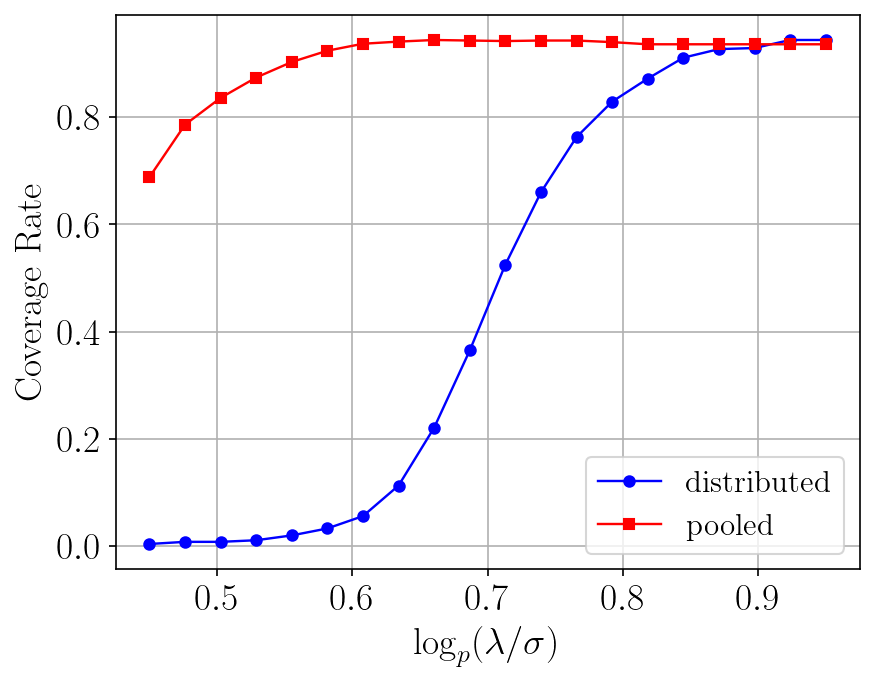}
		\caption{$p=50$}
		\label{fig:inf_p=1}
	\end{subfigure}
	\hfill
	\begin{subfigure}[b]{0.45\textwidth}
		\centering
		\includegraphics[width=\textwidth]{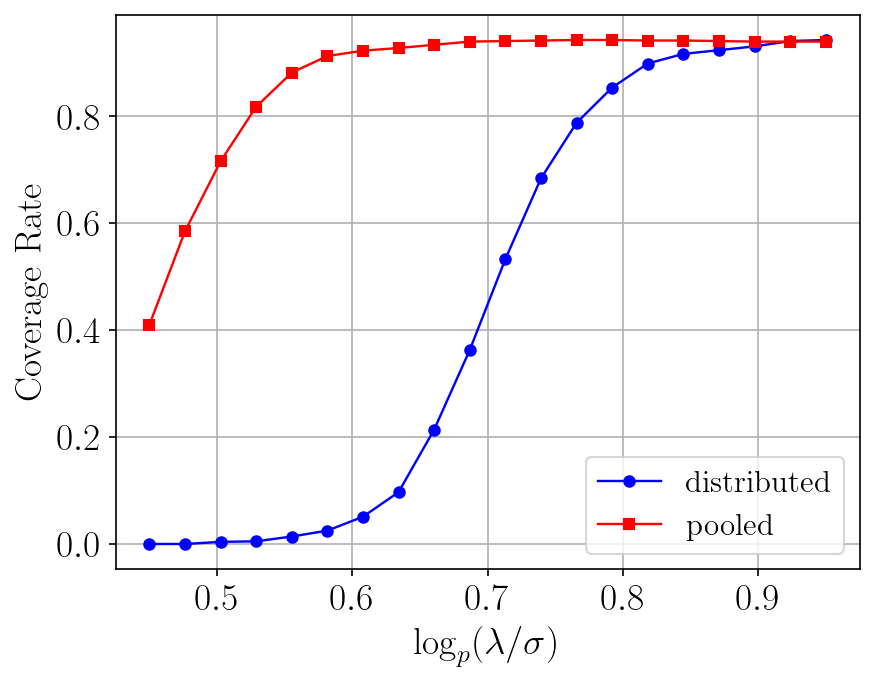}
		\caption{$p=100$}
		\label{fig:inf_p=10}
	\end{subfigure}
	\caption{The coverage rates of different methods under the heterogeneous setting.}
	\label{fig:inf}
\end{figure}

\subsection{Simulations for Asymptotic Distribution} \label{sec:supp-simulation}

In this section, we verify the validity of the asymptotic distribution established in Section \ref{sec:asymp} in the supplementary material by computing the coverage rates of our proposed distributed Algorithm \ref{alg:inf}, that is, the rate that the estimated confidence region \eqref{eq:confidence-region}  covers the truth $U_j$. As clarified in Section \ref{sec:asymp}, we estimate the noise level $\sigma$ by
\[\widehat{\sigma} =  \big\|\calT_{1}-\calT_{1} \times_1 \hU_{1} \times_2 \hU_{2}\cdots\times_J\hU_{J}\big\|_{\rm F} / \sqrt{p_1p_2\cdots p_J}, \] 
and estimate the singular value matrix $\Lambda_{j}$ by 
\[\widehat{\Lambda}_{j} = \text{the top $r_j$ singular values of } \calM_j\big(\calT_{1} \times_{1} \hU_{1}^{\top} \times_{2} \hU_{2}^{\top} \cdots \times_{j-1}\hU_{j-1}^{\top} \times_{j+1}\hU_{j+1}^{\top}\cdots \times_{J}\hU_{J}^{\top}\big),\]
where $\big\{\hU_j\big\}_{j \in [J]}$ are the outputs of Algorithm \ref{alg:inf}. For comparison, we also record the coverage rates of the pooled estimator, where the confidence region for $U_j$ is given by replacing $\hU_j$ in \eqref{eq:confidence-region} with $\hU_{\mathrm{pooled}, j}$.

Specifically, we choose the confidence level $1-\xi$ to be $0.95$ and report the coverage rates of $\hU_1$ and $\hU_{\mathrm{pooled}, 1}$ with $p=50$ and $100$ in Figure \ref{fig:inf}. For both cases, our proposed estimator performs comparably to the pooled estimator when the SNR $\lambda/\sigma$ is sufficiently large, achieving a high coverage rate around the nominal 95\% level. It is worth noting that, compared to  Figure  \ref{fig:homo}, the requirement for SNR to achieve the asymptotic normality is more stringent than that to attain the optimal statistical error rate. This observation is consistent with our theoretical results: Corollary 2.1 guarantees the optimal statistical error rate under the condition $\lambda/\sigma \geq \sqrt{prL}$,  whereas Theorem \ref{thm:inference}, which establishes the asymptotic normality, assumes a stronger condition that $\lambda/\sigma \geq L^{1/2}(pr)^{3/4}$.

\subsection{Knowledge Transfer in Real Data Analysis}\label{sec:supp-real}
\begin{figure}[!t]
\centering
\begin{subfigure}[b]{0.45\textwidth}
	\centering
	\includegraphics[width=\textwidth]{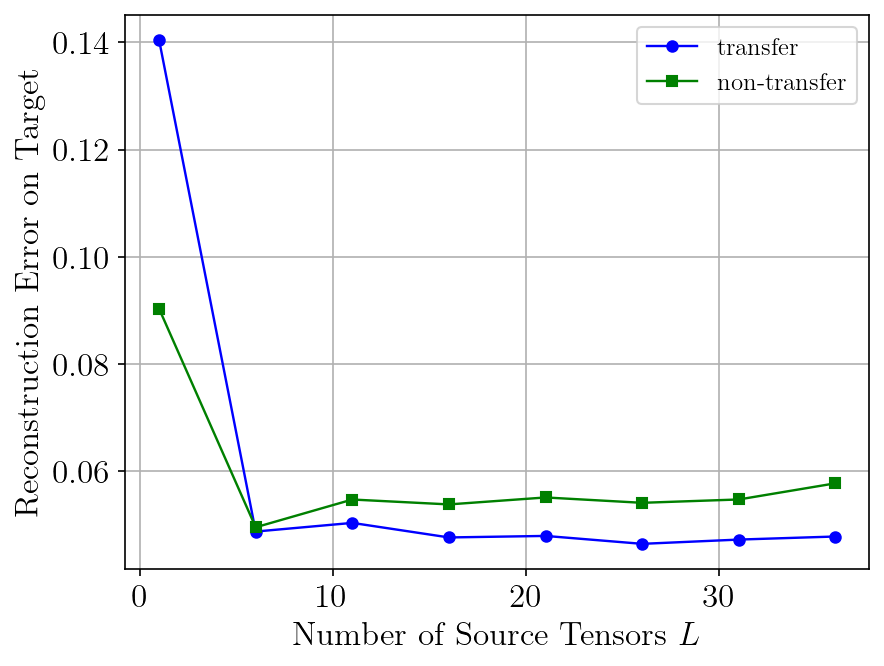}
	\caption{PROTEINS}
	\label{fig:transfer_small_r}
\end{subfigure}
\hfill
\begin{subfigure}[b]{0.45\textwidth}
	\centering
	\includegraphics[width=\textwidth]{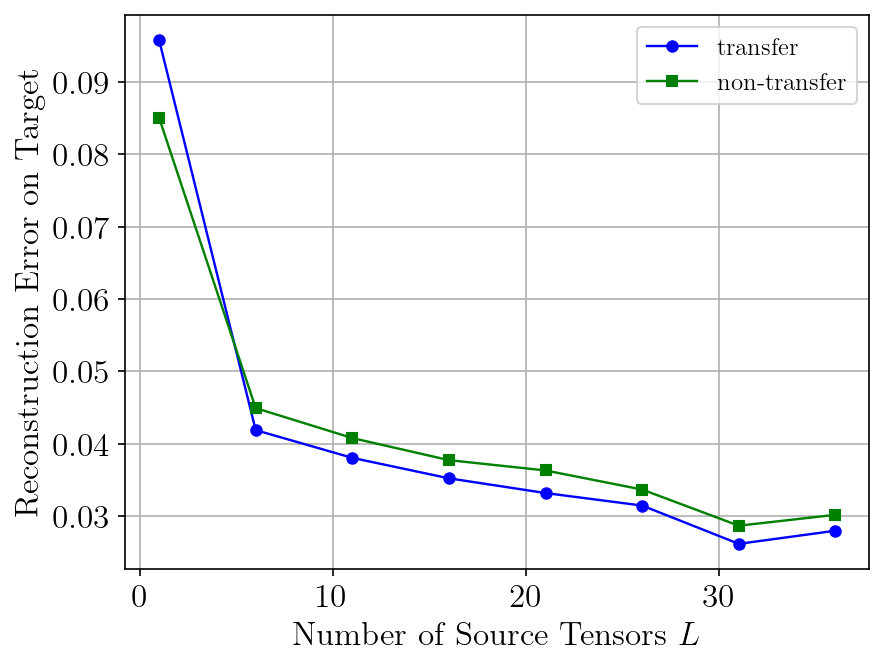}
	\caption{PTC\_FM}
	\label{fig:transferl_large_r}
\end{subfigure}
\caption{The reconstruction errors of transfer learning between class 0 and class 1 of the PROTEINS and PTC\_FM datasets.}
\label{fig:transfer}
\end{figure}
In the real data analysis, we further apply our transfer learning algorithm (Algorithm 3) to conduct knowledge transfer between class 0 and class 1. Since the sample size of class 0 is larger in both datasets, we let class 0 be the source task and class 1 be the target task. Similar to the previous procedure, we randomly select $L$ training samples on both tasks and $L'$ test samples on the target task. The training samples on each task are averaged into a source tensor $\calT_s$ and a target tensor $\calT_t$ and then input into Algorithm 3 to obtain estimators $\big\{[\hU_j \; \hV_{j,t}]\big\}_{j=1,2,3}$. 
For the sample sizes, we still let $L \in [1, 40]$ and $L'=100$. For comparison, we also record the performance of the ``non-transfer'' estimators, which are obtained only using the training samples on the target task. The results are displayed in Figure \ref{fig:transfer}. We observe that the transfer learning method consistently outperforms the ``non-transfer'' method for both datasets when $L>5$, demonstrating the advantage of leveraging knowledge from the source task. 

\end{document}